\documentclass[a4paper]{article}
\pdfoutput=1
\usepackage{amsmath,amsfonts,amsthm,amssymb}
\usepackage{mathtools}
\usepackage{bm}
\usepackage{booktabs} 
\usepackage[usenames,dvipsnames]{xcolor}
\usepackage{tikz-cd}
\usetikzlibrary{matrix,arrows,calc}
\usepackage{calc}
\usepackage[colorlinks=true,linkcolor=blue,citecolor=red,urlcolor=Violet,bookmarksdepth=subsection,final]{hyperref}
\usepackage[margin=3cm]{geometry}

\makeatletter
\renewcommand*\env@matrix[1][*\c@MaxMatrixCols c]{%
  \hskip -\arraycolsep
  \let\@ifnextchar\new@ifnextchar
  \array{#1}}
\makeatother

\tikzset{ampersand replacement=\&}

\allowdisplaybreaks[1] 

\newtheorem{thm}{Theorem}
\newtheorem{dfn}[thm]{Definition}
\newtheorem{prp}[thm]{Proposition}
\newtheorem{lem}[thm]{Lemma}

\theoremstyle{remark}
\newtheorem{rem}{Remark}

  \newcommand{\del}{\partial}
  \newcommand{\oo}{\infty}

  \newcommand{\id}{\mathrm{id}}
  \newcommand{\Lie}{\mathcal{L}}

  \newcommand{\eps}{\varepsilon}
  \newcommand{\la}{\lambda}
  \newcommand{\ka}{\kappa}

  \newcommand{\M}{\mathcal{M}}
  \newcommand{\R}{\mathcal{R}}
\renewcommand{\S}{\mathcal{S}}
\renewcommand{\d}{\mathrm{d}}
  \newcommand{\tr}{\operatorname{tr}}
\renewcommand{\div}{\operatorname{div}}

  \newcommand{\tC}{\tilde{C}}
  \newcommand{\tH}{\tilde{H}}
  
  \newcommand{\tK}{\tilde{K}}
  \newcommand{\tX}{\tilde{X}}
  \newcommand{\tY}{\tilde{Y}}
  \newcommand{\tZ}{\tilde{Z}}
  \newcommand{\Secs}{\Gamma}
  \newcommand{\tnabla}{\tilde{\nabla}}
  \newcommand{\DD}{\mathbb{D}}
  \newcommand{\TT}{\mathbb{T}}

\title{Compatibility complexes of overdetermined PDEs of finite type, with applications to the Killing equation}

\author{Igor Khavkine\\[0.5ex]
	Institute of Mathematics, Czech Academy of Sciences,\\
	\v{Z}itn{\'a} 25, 115 67 Praha 1, Czech Republic\\[0.5ex]
	\texttt{khavkine@math.cas.cz}}

\begin{document}
\maketitle


\begin{abstract}
In linearized gravity, two linearized metrics are considered
gauge-equivalent, $h_{ab} \sim h_{ab} + K_{ab}[v]$, when
they differ by the image of the Killing operator, $K_{ab}[v] =
\nabla_a v_b + \nabla_b v_a$. A universal (or complete)
compatibility operator for $K$ is a differential operator $K_1$ such
that $K_1 \circ K = 0$ and any other operator annihilating $K$ must
factor through $K_1$. The components of $K_1$ can be interpreted as a
complete (or generating) set of local gauge-invariant observables in
linearized gravity. By appealing to known results in the formal theory
of overdetermined PDEs and basic notions from homological algebra, we
solve the problem of constructing the Killing compatibility operator
$K_1$ on an arbitrary background geometry, as well as of extending it to
a full compatibility complex $K_i$ ($i\ge 1$), meaning that for each
$K_i$ the operator $K_{i+1}$ is its universal compatibility operator.
Our solution is practical enough that we apply it explicitly in two
examples, giving the first construction of full compatibility complexes
for the Killing operator on these geometries. The first example consists
of the cosmological FLRW spacetimes, in any dimension. The second
consists of a generalization of the Schwarzschild-Tangherlini black hole
spacetimes, also in any dimension. The generalization allows an
arbitrary cosmological constant and the replacement of spherical
symmetry by planar or pseudo-spherical symmetry.
\end{abstract}

\section{Introduction} \label{sec:intro}

An important aspect of General Relativity is its invariance under
diffeomorphisms, also called \emph{gauge transformations} of this
theory. Of course, this invariance survives linearization about some
fixed background metric $g$ and the linearized diffeomorphisms (or
linearized gauge transformations) change the linearized metric as
$h_{ab} \mapsto h_{ab} + K_{ab}[v]$, where $K_{ab}[v] = \nabla_a v_b +
\nabla_b v_a$ is the \emph{Killing operator} with respect to the
background metric $g$. Solutions of the \emph{Killing equation} $K[v] =
0$ are \emph{Killing vectors} $v_a$. Because two linearized metric
configurations are considered physically equivalent if they differ only
by a linearized gauge transformation, an inescapable part of the study
of linearized gravity (linearized General Relativity) is the need to
separate gauge and physical degrees of freedom; the latter essentially
parametrize equivalence classes of linearized metrics under linearized
gauge transformations.

A \emph{local gauge-invariant observable} is a differential operator
$O[h]$ such that $O[K[v]] = 0$ for an arbitrary argument $v_a$. Clearly,
such differential operators have many potential applications in
linearized gravity and, not surprisingly, their study has a long
history~\cite{stewart-walker}. While not all useful gauge-invariant
observables $O[h]$ are local (where $O$ is \emph{local} if it is a differential,
rather than an integral, operator), the local ones are distinguished by
the property that they preserve supports, $\operatorname{supp} O[h]
\subseteq \operatorname{supp} h$, which helps to disentangle the
gauge-invariant information contained in $h$ from infrared or asymptotic
properties of $h$. Further discussion of these issues, with brief
surveys of previous work, can be found~\cite{fhk}, in the context of
cosmological perturbations, and in~\cite{jezierski,swaab,ab-kerr}, in
the context of black hole perturbations.

In this work, we are interested in the problem of explicitly
constructing complete (or generating) sets of local gauge-invariant
observables on spacetime backgrounds of physical interest. Completeness
refers to the ability to express any local gauge-invariant observable in
terms of linear combinations of derivatives of a given set. For
technical reasons~\cite{kh-peierls, kh-calabi}, it also becomes
important to identify complete sets of differential relations between
them, complete sets of differential relations between these differential
relations, and so on.
For instance, once a complete set of local gauge-invariant observables
is known, the complete set of differential relations between them may
allow us to reduce the number of independent invariants at the price of
integrating some differential equations (if any sufficiently simple ones
can be identified among the differential relations). This is the way in
which Teukolsky-Starobinsky identities are used to relate the Teukolsky
scalars to other invariants on the Kerr background~\cite{whiting-price,
ab-kerr}. Also, the differential relations among invariants may play a
role in the construction of a wave equation satisfied by the invariants,
just as the Bianchi identities $\nabla_{[a}F_{bc]}=0$ play a role in
obtaining the wave equation for the Maxwell tensor $F_{bc} = \nabla_{[a}
A_{b]}$ (whose components constitute a complete set of invariants in
electrodynamics). For linearized gravity on maximally symmetric spaces,
this idea was used in~\cite{kh-calabi} to identify wave equations
satisfied by the invariants, the differential relations between them, as
well as all higher differential relations. Finally, knowledge of all the
higher differential relations has interesting applications to the
analysis of the symplectic and Poisson structures on the space of
solutions of linearized gravity~\cite{kh-peierls}
\cite[Sec.5]{kh-calabi}.
Thus, phrased in mathematical terms, given a background
metric $g$, we are interested in constructing a (full) compatibility
complex for the corresponding Killing operator $K[v]$, where \emph{full}
refers to the continuation of the sequence of differential relations
until it terminates (becomes identically zero), a property that is usually required implicitly.

An unfortunate aspect of the study of local gauge-invariant observables
$O[h]$ is that their structure depends strongly on the background metric
$g$, since the Killing operator $K[v]$, which determines the structure
of gauge-equivalence classes, depends on $g$ in an essential way. Thus,
in principle, this problem needs to be attacked anew for each background
metric of interest. Unfortunately, a full solution (a complete set of
gauge-invariants, relations between them, etc.)\ can be found in the
literature only in very few cases, even if we restrict ourselves only to
the construction of complete sets of gauge-invariants (and not relations between them, etc.). To our knowledge,
the full Killing compatibility complex is known only for flat
(Minkowski) and constant curvature (de~Sitter or anti-de~Sitter)
spacetimes~\cite{kh-calabi}. In principle, the methods of~\cite{gg83,
gg88} could have been used to generate the compatibility complex on
locally symmetric spacetimes (those with a covariantly constant Riemann
tensor), but to our knowledge they have never been explicitly elaborated
in the Lorentzian setting~\cite{cahen-wallach}. In addition, complete
sets of local gauge-invariant observables are known only for
cosmological (inflationary FLRW) spacetimes in any dimension, due to the
recent construction in~\cite{fhh, cdk, fhk}, and for the $4$-dimensional
Kerr black hole, as recently highlighted in~\cite{ab-kerr}. Full
proofs of the results announced in~\cite{ab-kerr} will appear
in~\cite{aabkw} and will be based on the methods to be presented in
this work.

The major obstacle to solving the problem that we have posed (the
construction of a compatibility complex for the Killing operator) has so
far been proving completeness (of a set of gauge-invariants, of a set of
relations between them, etc.). In the flat and constant curvature cases,
the proof was basically due to Calabi~\cite{calabi,kh-calabi}, and was
specific to those geometries. In the cosmological case, the proof is due
to~\cite{fhh}, but is somewhat ad-hoc and without clear generalizations.

The main innovation in this work is the application of methods from the
formal theory of PDEs~\cite{spencer, goldschmidt-lin, tarkhanov, seiler-inv} and
homological algebra~\cite{weibel} to the problem of constructing
Killing compatibility complexes. In fact, a method for systematically
constructing a complete compatibility operator for any overdetermined
linear differential operator (under mild regularity conditions) has been
known for a long time~\cite{goldschmidt-lin} (it was this method
that was applied in~\cite{gg83, gg88}). Unfortunately, it is rather
cumbersome to apply directly. There do exist computer algebra
implementations of this method~\cite{janet}, but they suffer from the
problem that the input and output of this computer algebra construction
must be matrices of scalar differential operators written in some
explicit coordinates, which often destroys any manifest symmetry or
other structure that the original linear differential operator had. This
is certainly an undesirable feature when dealing with the Killing
operator on a spacetime with some symmetry, product or warped product
structure. However, there is a significant simplification of the general
systematic construction when we restrict our attention to differential
operators of \emph{regular finite type}, of which the Killing operator is often an
example. We will take full advantage of this simplification, together
with some basic notions from homological algebra, to give a practical
sufficient condition (Lemma~\ref{lem:compat-sufficient}) for the
completeness of a given set of local gauge-invariant observables in
linearized gravity (or more generally, the completeness of a
compatibility operator for any operator). In practice, this criterion
also leads to a way to construct (Theorem~\ref{thm:fintype-compat}) the
full Killing compatibility complex (or more generally, the compatibility
complex for any operator of regular finite type), which can preserve various
structural properties of a given background spacetime geometry.

In Section~\ref{sec:compat}, we introduce some ideas from homological
algebra, applied to linear differential operators, and use it to show
how to explicitly construct a compatibility complex for a PDE of regular finite
type (the appropriateness of our definition of \emph{finite type}
operator is discussed at length in Appendix~\ref{sec:fintype}). This
technique is applied to the Killing equation in
Section~\ref{sec:killing}, with some examples. In particular, we treat
in detail the examples of spacetimes of constant curvature
(Section~\ref{sec:cc}), cosmological FLRW spacetimes
(Section~\ref{sec:flrw}) and Schwarzschild-Tangherlini black holes
(Section~\ref{sec:schw}). In each case we make some remarks about the
relation of our results with the literature. Appendix~\ref{sec:notation}
gives a helpful reference for the notation used in different subsections
of Section~\ref{sec:killing}. In all examples, we keep the spacetime
dimension $n$ general (that is, we allow at least $n\ge 4$). The results
of Sections~\ref{sec:flrw} and~\ref{sec:schw} are new. Finally, we
conclude with a discussion of further work in Section~\ref{sec:discuss}.

Whenever speaking of differential operators, we will specifically mean a
linear differential operator with smooth coefficients acting on smooth
functions. More precisely, we will consider differential operators that
map between sections of vector bundles, say $V_1 \to M$ and $V_2 \to M$,
on some fixed manifold $M$, $K\colon \Gamma(V_1) \to \Gamma(V_2)$. The
source and target bundle of a differential operator, $V_1\to M$ and
$V_2\to M$ respectively in the last example, will be considered as part
of its definition and will most often be omitted from the notation. We
will denote the composition of two differential operators $K$ and $L$ by
$K\circ L$, or simply by $K L$, if no confusion is possible. A
\emph{local section} of a vector bundle $V\to M$ is a section of the
restriction bundle $V|_U \to U$ for some open $U \subset M$. A local
section $v$ that solves the differential equation $K[v] = 0$ on its
domain of definition is a \emph{local solution}.

\section{Compatibility operators} \label{sec:compat}

We start by introducing some basic notions from \emph{homological
algebra}~\cite{weibel}.

\begin{dfn} \label{def:homalg}
A (possibly infinite) composable sequence $K_l$ of linear maps,
$l=l_{\min}, \ldots, l_{\max}$, such that
$K_{l+1} \circ K_l = 0$ when possible, is called a \emph{(cochain)
complex}. Given complexes $K_l$ and $K'_l$ a sequence $C_l$ of linear
maps, as in the diagram
\begin{equation}
\begin{tikzcd}[column sep=large,row sep=large]
	\cdots \ar{r} \&
	\bullet \ar{r}{K_{l-1}} \ar{d}{C_{l-1}} \&
	\bullet \ar{r}{K_l} \ar{d}{C_l} \&
	\bullet \ar{r}{K_{l+1}} \ar{d}{C_{l+1}} \&
	\bullet \ar{r} \ar{d}{C_{l+2}} \&
	\cdots
	\\
	\cdots \ar{r} \&
	\bullet \ar[swap]{r}{K'_{l-1}} \&
	\bullet \ar[swap]{r}{K'_l} \&
	\bullet \ar[swap]{r}{K'_{l+1}} \&
	\bullet \ar{r} \&
	\cdots
\end{tikzcd} ,
\end{equation}
such that its squares commute, that is $K'_l \circ C_l = C_{l+1} \circ
K_l$ when possible, is called a \emph{cochain map} or a \emph{morphism}
between complexes. A \emph{homotopy} between complexes $K_l$ and $K'_l$
(which could also be the same complex, $K_l = K'_l$) is a sequence of
morphism, as the dashed arrows in the diagram
\begin{equation}
\begin{tikzcd}[column sep=large,row sep=large]
	\cdots \ar{r} \&
	\bullet \ar{r}{K_{l-1}} \ar{d}{C_{l-1}} \&
	\bullet \ar{r}{K_l} \ar{d}{C_l} \ar[dashed]{dl}{H_{l-1}} \&
	\bullet \ar{r}{K_{l+1}} \ar{d}{C_{l+1}} \ar[dashed]{dl}{H_l} \&
	\bullet \ar{r} \ar{d}{C_{l+2}} \ar[dashed]{dl}{H_{l+1}}\&
	\cdots
	\\
	\cdots \ar{r} \&
	\bullet \ar[swap]{r}{K'_{l-1}} \&
	\bullet \ar[swap]{r}{K'_l} \&
	\bullet \ar[swap]{r}{K'_{l+1}} \&
	\bullet \ar{r} \&
	\cdots
\end{tikzcd} ,
\end{equation}
and the sequence of maps $C_l = K'_{l-1}\circ H_{l-1} + H_l\circ K_l$ is
said to be a \emph{morphism induced by} the homotopy $H_l$. An
\emph{equivalence up to homotopy} between complexes $K_l$ and $K'_l$ is
a pair of morphisms $C_l$ and $D_l$ between them, as in the diagram
\begin{equation}
\begin{tikzcd}[column sep=large,row sep=4.5em]
	\ar[loop left]{}{\tH_{l_{\min}-1}}
	\bullet \ar{r}{K_{l_{\min}}}
		\ar[swap,shift right]{d}{C_{l_{\min}}} \&
	\ar[r,phantom,"\cdots"]
		\ar[dashed,bend left]{l}{H_{l_{\min}}} \&
	\bullet \ar{r}{K_l}
		\ar[swap, shift right]{d}{C_l} \&
	\bullet \ar[r,phantom,"\cdots"]
		\ar[swap,shift right]{d}{C_{l+1}}
		\ar[dashed,bend left]{l}{H_l} \&
	\ar{r}{K_{l_{\max}}} \&
	\bullet \ar[swap,shift right]{d}{C_{l_{\max}+1}}
		\ar[dashed,bend left]{l}{H_{l_{\max}}}
		\ar[loop right]{}{\tH_{l_{\max}+1}}
	\\
	\ar[loop left]{}{\tH'_{l_{\min}-1}}
	\bullet \ar[swap]{r}{K'_{l_{\min}}}
		\ar[swap,shift right]{u}{D_{l_{\min}}} \&
	\ar[r,phantom,"\cdots"]
		\ar[swap,dashed,bend right]{l}{H'_{l_{\min}}} \&
	\bullet \ar[swap]{r}{K'_l}
		\ar[swap,shift right]{u}{D_l} \&
	\bullet \ar[r,phantom,"\cdots"]
		\ar[swap,shift right]{u}{D_{l+1}}
		\ar[swap,dashed,bend right]{l}{H'_l} \&
	\ar[swap]{r}{K'_{l_{\max}}} \&
	\bullet \ar[swap,shift right]{u}{D_{l_{\max}+1}}
		\ar[swap,dashed,bend right]{l}{H'_{l_{\max}}}
	\ar[loop right]{}{\tH'_{l_{\max}+1}}
\end{tikzcd} ,
\end{equation}
such that $C_l$ and $D_l$ are mutual inverses up to homotopy ($H_l$ and
$H'_l$), that is
\begin{align}
	D_l \circ C_l &= \id - K_{l-1} \circ H_{l-1} - H_l \circ K_l ,
	\\
	C_l \circ D_l &= \id - K'_{l-1} \circ H'_{l-1} - H'_l \circ K'_l ,
\end{align}
with the special end cases
\begin{align}
	D_{l_{\min}} \circ C_{l_{\min}}
	&= \id - \tH_{l_{\min}-1} - H_{l_{\min}} \circ K_{l_{\min}} , &
	K_{l_{\min}} \circ \tH_{l_{\min}-1} &= 0 ,
	\\
	C_{l_{\min}} \circ D_{l_{\min}}
	&= \id - \tH'_{l_{\min}-1} - H'_{l_{\min}} \circ K'_{l_{\min}} , &
	K'_{l_{\min}} \circ \tH'_{l_{\min}-1} &= 0 ,
	\\
	D_{l_{\max}+1} \circ C_{l_{\max}+1}
	&= \id - H_{l_{\max}} \circ K_{l_{\max}} - \tH_{l_{\max}+1} , &
	\tH_{l_{\max}+1} \circ K_{l_{\max}}  &= 0 ,
	\\
	C_{l_{\max}+1} \circ D_{l_{\max}+1}
	&= \id - H'_{l_{\max}} \circ K'_{l_{\max}} - \tH'_{l_{\max}+1} , &
	\tH'_{l_{\max}+1} \circ K'_{l_{\max}}  &= 0 ,
\end{align}
where the $\tH$ maps are allowed to be arbitrary, as long as they
satisfy the given identities.
\end{dfn}

Note that our definition of equivalence up to homotopy between complexes
of finite length is set up in a way that allows an equivalence between
longer complexes to be truncated and still remain an equivalence.

Next, we restrict our attention to the case where all maps are given by
differential operators.

\begin{dfn}[{cf.~\cite[Def.10.5.4]{seiler-inv}, \cite[Def.1.2.2]{tarkhanov}}] \label{def:compat}
Given a differential operator $K$, any composable differential operator
$L$ such that $L\circ K = 0$ is a \emph{compatibility operator} for $K$.
If $K_1$ is a compatibility operator for $K$, it is called \emph{complete}
or \emph{universal} when any other compatibility operator $L$ can be
factored through $L = L'\circ K_1$ for some differential operator $L'$. A
complex of differential operators $K_l$, $l=0,1,\ldots$ is called a
compatibility complex for $K$ when $K_0 = K$ and, for each $l\ge 1$,
$K_l$ is a complete compatibility operator for $K_{l-1}$.
\end{dfn}

\begin{dfn} \label{def:loc-exact}
Given a (possibly infinite) complex of differential operators $K_l$,
$l=l_{\min}, l_{\min}+1, \ldots, l_{\max}$, we say that it is
\emph{locally exact} at a point $x$ when, for every $l_{\min} < l <
l_{\max}$, for every smooth function $f_l$ defined on an open
neighborhood $U \ni x$ such that $K_l[f_l] = 0$, there exists a smooth
function $g_{l-1}$ defined on a possibly smaller open neighborhood $V
\ni x$ such that $f_l = K_{l-1}[g_{l-1}]$. \emph{Locally exact} (without
specifying a point $x$) means locally exact at every $x$.
\end{dfn}

Note that a complete compatibility operator, say $K_1$, need not be
unique. But, by its universal factorization property, any two
compatibility operators, say $K_1$ and $K'_1$, must factor through each
other, $K_1 = L_1 \circ K'_1$ and $K'_1 = L'_1\circ K_1$ for some
differential operators $L_1$ and $L'_1$.

Given two composable operators, $K$ and $K_1$, the compatibility
condition $K_1\circ K = 0$ is very easy to check. On the other hand, it
may be quite challenging to check completeness/universality. One way to
do it is to compare $K$ and $K_1$ with another pair of operators which
are already known to satisfy the universality condition.

\begin{lem} \label{lem:compat-sufficient}
Consider two complexes of differential operators $K_l$ and $K'_l$, for
$l=0,1$. If these complexes are equivalent up to homotopy, as in the
diagram
\begin{equation}
\begin{tikzcd}[column sep=huge,row sep=huge]
	\bullet \ar{r}{K_0}
		\ar[swap,shift right]{d}{C_0} \&
	\bullet \ar{r}{K_1}
		\ar[swap,shift right]{d}{C_1}
		\ar[dashed,bend left]{l}{H_0} \&
	\bullet
		\ar[swap,shift right]{d}{C_2}
		\ar[dashed,bend left]{l}{H_1}
	\\
	\bullet \ar[swap]{r}{K'_0}
		\ar[swap,shift right]{u}{D_0} \&
	\bullet \ar[swap]{r}{K'_1}
		\ar[swap,shift right]{u}{D_1}
		\ar[swap,dashed,bend right]{l}{H'_0} \&
	\bullet
		\ar[swap,shift right]{u}{D_2}
		\ar[swap,dashed,bend right]{l}{H'_1}
\end{tikzcd} ,
\end{equation}
where we really only require all squares to be commutative and the
identities $D_1\circ C_1 = \id - K_0 \circ H_0 - H_1 \circ K_1$ and
$C_1\circ D_1 = \id - K'_0 \circ H'_0 - H'_1 \circ K'_1$ to hold, then
$K_1$ is universal iff $K'_1$ is universal.

Furthermore, the complex $K_l$, $l=0,1$, is locally exact iff the
complex $K'_l$, $l=0,1$, is locally exact.
\end{lem}

\begin{proof}
Without loss of generality, assume that $K'_1$ is universal. Let $L\circ
K_0 = 0$. Then $(L \circ D_1) \circ K'_0 = L \circ K_0 \circ D_0 = 0$.
By universality of $K'_1$, there exists a differential operator $L'$
such that $L \circ D_1 = L'\circ K'_1$. Recall that from our hypotheses
that $D_1\circ C_1 = \id - K_0 \circ H_0 - H_1 \circ K_1$. But then
\begin{align}
\notag
	L &= L\circ (D_1\circ C_1 + K_0\circ H_0 + H_1\circ K_1) \\
\notag
	&= L' \circ (K'_1\circ C_1) + (L\circ H_1) \circ K_1 \\
\notag
	&= (L'\circ C_2) \circ K_1 + (L\circ H_1) \circ K_1
	= L'' \circ K_1 ,
\end{align}
where $L'' = L'\circ C_2 + L\circ H_1$. This demonstrates the
universality of $K_1$.

Next, without loss of generality, assume that $K'_l$ is locally exact.
Pick a point $x$, an open neighborhood $U \ne x$, and a smooth function
$f$ such that $K_1[f] = 0$. Then $K'_1[C_1[f]] = C_2[K_1[f]] = 0$.
Hence, by local exactness, there exists a smooth $g'$ defined on a
possibly smaller open neighborhood $V \ni x$ such that $K'_0[g'] =
C_1[f]$. Setting $g = D_0[g'] + H_0[f]$ on $V\ni x$, direct
calculation shows that
\begin{align}
	K_0[g]
\notag
	&= K_0[D_0[g']] + K_0[H_0[f]]
	= D_1[K'_0[g']] + K_0[H_0[f]] \\
\notag
	&= D_1\circ C_1[f] + K_0[H_0[f]]
	= f - H_1[K_1[f]] \\
	&= f ,
\end{align}
which shows that the $K_l$ complex is also locally exact.
\end{proof}

Next, we will show how to construct a universal compatibility operator
for a differential operator $K$ if it is equivalent, in the sense of a
complex consisting of one operator, to some operator with a known
universal compatibility operator. This construction is not unknown in
homological algebra, but we include a proof for completeness.

\begin{lem}[{cf.~\cite[Prp.1.2.7]{tarkhanov}}] \label{lem:lift-compat}
Consider differential operators $K_0$ and $K'_0$. Suppose that $K_0$ and
$K'_0$ are equivalent up to homotopy, in the sense of the diagram
\begin{equation}
\begin{tikzcd}[column sep=huge,row sep=huge]
	\ar[loop left]{}{\tH}
	\bullet \ar{r}{K_0}
		\ar[swap,shift right]{d}{C_0} \&
	\bullet
		\ar[swap,shift right]{d}{C_1}
		\ar[dashed,bend left]{l}{H_0}
	\\
	\ar[loop left]{}{\tH'}
	\bullet \ar[swap]{r}{K'_0}
		\ar[swap,shift right]{u}{D_0} \&
	\bullet 
		\ar[swap,shift right]{u}{D_1}
		\ar[swap,dashed,bend right]{l}{H'_0} 
\end{tikzcd} ,
\end{equation}
where we require all squares to be commutative and the identities
$D_0\circ C_0 = \id - \tH - H_0 \circ K_0$, $C_0 \circ D_0 = \id
- \tH' - H'_0 \circ K'_0$ to hold, with $K_0 \circ \tH = 0$ and $K'_0
\circ \tH' = 0$. Then, if a universal compatibility operator $K'_1$ for
$K'_0$ is known, we can complete the above diagram to the following
equivalence up to homotopy
\begin{equation} \label{eq:lifted-compat}
\begin{tikzcd}[column sep=8em,row sep=7em]
	\bullet \ar{r}{K_0}
		\ar[swap,shift right]{d}{C_0} \&
	\bullet \ar{r}{K_1
		= \begin{bmatrix}
			\id - K_0\circ H_0 - D_1\circ C_1 \\
			K'_1 \circ C_1\end{bmatrix}}
		\ar[swap,shift right]{d}{C_1}
		\ar[dashed,bend left]{l}{H_0} \&
	\bullet
		\ar[swap,shift right]{d}{C_2
			= \begin{bmatrix}
				0 & \id\end{bmatrix}}
		\ar[dashed,bend left]{l}{H_1
			= \begin{bmatrix}
				\id & 0\end{bmatrix}}
	\\
	\bullet \ar[swap]{r}{K'_0}
		\ar[swap,shift right]{u}{D_0} \&
	\bullet \ar[swap]{r}{K'_1}
		\ar[swap,shift right]{u}{D_1}
		\ar[swap,dashed,bend right]{l}{H'_0} \&
	\bullet
		\ar[swap,shift right]{u}{D_2
			= \begin{bmatrix}
				D'_2 \\ \id - K'_1\circ H'_1\end{bmatrix}}
		\ar[swap,dashed,bend right]{l}{H'_1}
\end{tikzcd} ,
\end{equation}
with some differential operators $H'_1$, $D_2'$.
\end{lem}

\begin{proof}
From our hypotheses, $C_0$ and $D_0$ are mutual inverses, up to a
homotopy correction. Our first observation is that the same property
then holds for $C_1$ and $D_1$. Namely,
\begin{align}
	(\id - K_0 \circ H_0 - D_1\circ C_1) \circ K_0
\notag
	&= K_0 - K_0 \circ H_0 \circ K_0 - K_0 \circ (D_0 \circ C_0) \\
\notag
	&= K_0 \circ (\id - H_0\circ K_0 - D_0 \circ C_0) \\
	&= K_0 \circ \tH = 0 , \\
	(\id - K'_0 \circ H'_0 - C_1\circ D_1) \circ K'_0
	&= 0 ,
\end{align}
where the second identity is completely analogous to the first one. Then
we also have
\begin{equation}
	(\id - K_0 \circ H_0 - D_1\circ C_1) \circ D_1 \circ K'_0
	= (\id - K_0 \circ H_0 - D_1\circ C_1) \circ K_0 \circ D_0
	= 0 .
\end{equation}
Since we know that $K'_1$ is a universal compatibility operator for
$K'_0$, there must exist differential operators $H'_1$ and $D'_2$ such
that
\begin{align}
	\id - K'_0 \circ H'_0 - C_1\circ D_1
	&= H'_1 \circ K'_1 ,
	\\
	(\id - K_0 \circ H_0 - D_1\circ C_1) \circ D_1
	&= D'_2 \circ K'_1 .
\end{align}

Next, defining the operators $K_1$, $H_1$, $C_2$ and $D_2$ as in the
diagram~\eqref{eq:lifted-compat}, the remaining identities needed to show
that this diagram is a homotopy equivalence are
\begin{subequations}
\begin{align}
	D_1 \circ C_1 &= \id - K_0\circ H_0 - H_1 \circ K_1 , \\
	C_2 \circ K_1 &= K'_1 \circ C_1 , \\
	D_2 \circ K'_1 &= K_1 \circ D_1 , \\
	(\id - K_1 \circ H_1 - D_2\circ C_2) \circ K_1 &= 0 , \\
	(\id - K'_1 \circ H'_1 - C_2\circ D_2) \circ K'_1 &= 0 .
\end{align}
\end{subequations}
The last two, (d) and (e), follow from the same argument as in the first
paragraph of this proof. The first two, (a) and (b), follow from direct
calculation and the identities that we have already established earlier
in the proof. To get (c), it remains to check the following identity
\begin{equation}
	(K'_1 \circ C_1) \circ D_1
	= K'_1 \circ (\id - K'_0 \circ H'_0 - H'_1\circ K'_1)
	= (\id - K'_1 \circ H'_1) \circ K'_1 .
\end{equation}
This completes the proof.
\end{proof}

\begin{dfn} \label{def:flat-conn}
A \emph{(linear) connection} $\DD$ on a vector bundle $V\to M$ is a
linear differential operator $\DD\colon \Secs(V) \to \Secs(T^*M
\otimes_M V)$ that in local coordinates $(x^a)$ has the form $\DD_a v(x)
= \frac{\del}{\del x^a} v(x) + \gamma_a(x) v(x)$, with the $\gamma_a(x)$
being smooth local sections of the endomorphism bundle of $V\to M$. The
connection is \emph{flat} when locally its components commute,
$[\mathbb{D}_a, \mathbb{D}_b] = 0$. The equation $\DD f = 0$
is called the \emph{($\DD$-)flat section equation}.

A \emph{flat connection} $\DD$ gives rise to a complex of differential
operators $d^\DD_l \colon \Secs(\Lambda^l T^*M \otimes_M V) \to
\Secs(\Lambda^{l+1} T^*M \otimes_M V)$, $l=0,1,\ldots,n,\ldots$, with
$d^\DD_0 = \DD$, locally $(d^\DD_l w)_{a_1 \cdots a_{l+1}} = (l+1)
\DD_{[a_1} w_{a_2\cdots a_{l+1}]}$ for $0<l<n$, and $d^\DD_l = 0$ for $l\ge
n$, the \emph{($\DD$-)twisted de~Rham complex}.
\end{dfn}

\begin{dfn} \label{def:fintype}
A differential operator $K$ defines a \emph{PDE of regular finite type} $K[f] =
0$, when $K_0 = K$ is equivalent to a flat connection $K'_0=\DD$, in
the sense of one operator complexes (Definition~\ref{def:homalg}), with
the extra requirement that $\tH_{-1} = 0$ and $\tH'_{-1} = 0$.
\end{dfn}

\begin{rem}
One can find in the literature different definitions for PDEs of
\emph{finite type}. Our Definition~\ref{def:fintype}, with the extra
\emph{regular} modifier, is most convenient for the purposes of this
work and is well-known to be equivalent to other common definitions when
the regularity conditions are precisely specified. The relation between
the different definitions is discussed in the Appendix and an elementary
proof of the equivalence is given in Proposition~\ref{prp:fintype-conn}.
\end{rem}

Once we know that we are faced with a PDE of regular finite type, we can exploit
its equivalence to the equation for $\DD$-flat sections, together
with our preceding results and the following well known proposition.

\begin{prp} \label{prp:de-rham}
Given a flat connection $\DD$, the corresponding twisted de~Rham
complex $d^\DD_l$, $l=0,1,\ldots,n,\ldots$ is locally exact and is
also a compatibility complex for $\DD = d^\DD_0$.
\end{prp}
\begin{proof}
A well-known property of flat connections is that they are locally
equivalent to $\operatorname{rk}_M V$ copies of exterior derivative
operator (we can set the $\gamma_a(x)$ matrices to zero by locally
choosing the fiber coordinates on $V\to M$ adapted to the foliating by
flat sections). Thus, locally, the twisted de~Rham complex turns into
$\operatorname{rk}_M V$ copies of the ordinary de~Rham complex, which is
known to satisfy all the desired properties (it suffices to combine the
Poincar\'e lemma~\cite[Prp.1.2.41]{tarkhanov} with
\cite[Prps.1.2.13,1.2.39]{tarkhanov}).
\end{proof}

\begin{thm} \label{thm:fintype-compat}
Let $K[f] = 0$ be PDE of regular finite type. Then, starting from an equivalence
of $K$ with a flat connection $\DD$, we can explicitly construct a
locally exact compatibility complex $K_l$ for $K$, $l=0,1,\ldots$.
\end{thm}

\begin{proof}
The proof is by induction. Setting $K_0 = K$ and $K'_0 = \DD =
d^\DD_0$, the regular finite type hypothesis implies that we can satisfy the
hypotheses of Lemma~\ref{lem:lift-compat} and extend the equivalence of
$K_0$ and $K'_0$ to an equivalence of $K_l$ and $K'_l$, $l=0,1$, with
$K_1$ explicitly constructed. Suppose, inductively, that we have an
equivalence up to homotopy between the complexes $K_l$ and $K'_l =
d^\DD_l$, $l=0,1,\ldots,m$, for some $m>0$. Iterating the previous
argument, we can extend it to an equivalence up to homotopy between the
complexes $K_l$ and $K'_l = d^\DD_l$, $l=0,1,\ldots,m+1$, where
$K_{m+1}$ is an explicitly constructed as in
Lemma~\ref{lem:lift-compat}.

Since the above construction of the complex $K_l$, $l=0,1,\ldots$, comes
with an equivalence up to homotopy with the twisted de~Rham complex
$K'_l = d^\DD_l$, $l=0,1,\ldots$, Lemma~\ref{lem:compat-sufficient}
and Proposition~\ref{prp:de-rham} allow us to conclude that $K_l$ is both
locally exact and is a compatibility complex for $K$. That is, $K_1$ is
a universal compatibility operator for $K_0 = K$ and $K_{l+1}$ is a
universal compatibility operator for $K_l$, for $l>1$.
\end{proof}

Note that, even though the twisted de~Rham complex $d^\DD_l$
terminates after $l=n$, or rather becomes trivial $d^\DD_l = 0$ for
$l\ge n$, the result of Theorem~\ref{thm:fintype-compat} is not
guaranteed to produce a complex that eventually terminates in the same
way. For instance, the simple example $K = \id$, produces the complex
$K_l$, $l=0,1,\ldots$, where $K_{2k} = \id$ and $K_{2k+1} = 0$, with the
source and target of every operator $K_l$ being the same as for $K$. Of
course, in this simple example, we can force this complex to terminate
by setting the target of $K_1 = 0$ to be zero dimensional, and setting
$K_l = 0$ as a map between zero dimensional spaces, for each $l>1$.

Since our goal is not a fully algorithmic construction of compatibility
complexes, but rather one where human intervention is allowed along the
way, we can apply similar simplifications at each step of the iterative
construction given in the proof of Theorem~\ref{thm:fintype-compat}. At
each step, before proceeding to the next one, having obtained the
operator $K_{l+1}$, we can replace it by a potentially simpler operator
$\tilde{K}_{l+1}$ without breaking its universality property. Here, a
trivial, but helpful, observation is that an operator $\tilde{K}_{l+1}$
such that $K_{l+1} = \tilde{K}'_{l+1} \circ \tilde{K}_{l+1}$ is a
universal compatibility operator for $K_l$ whenever $K_{l+1}$ is. This
way, on general principles, we expect to be able to produce
a compatibility complex $K_l$, $l=0,1,\ldots$, that becomes trivial $K_l
= 0$ for $l>n$.

We will not try to give a rigorous proof of the fact that we can always
produce a compatibility complex $K_l$ that trivializes to $K_l = 0$ for
$l>n$. Instead, in the next section, we will present examples of PDEs of
regular finite type with compatibility complexes of finite length. In each case,
we will give an explicit equivalence up to homotopy of a given complex
to a twisted de Rham complex, which together with
Lemma~\ref{lem:compat-sufficient} and Proposition~\ref{prp:de-rham}
serves as a witness to the fact that it is a compatibility complex.
However, the reader should understand that this compatibility complex
was produced by the construction given in the proof of
Theorem~\ref{thm:fintype-compat}, with intermediate simplifications as
described above.

\section{Killing equation} \label{sec:killing}

Consider an $n$-dimensional (pseudo-)Riemannian manifold $(M,g)$, with
Levi-Civita connection $\nabla$. In Lorentzian signature, we refer to
$(M,g)$ as a \emph{spacetime}. The \emph{Killing equation} is an
equation on sections $v \in \Secs(TM)$, namely
\begin{equation}
	K_{ab}[v] = \nabla_a v_b + \nabla_b v_a = 0 .
\end{equation}
The Lie derivative identity $K[v] = \Lie_v g$ implies that solutions of
the Killing equations are infinitesimal isometries of $(M,g)$. In the
context of \emph{linearized gravity} (that is, the theory of
\emph{linearized Einstein equations}), metric perturbations $h \in
\Secs(S^2T^*M)$ are grouped into \emph{gauge equivalence classes}, $h
\sim h + K[v]$ for $v \in \Secs(TM)$. A differential operator $L[h]$
such that $L\circ K = 0$, a compatibility operator for $K$ in the
terminology of Section~\ref{sec:compat}, is interpreted as a
\emph{(local) gauge-invariant observable} or \emph{gauge-invariant field
combination}~\cite{stewart-walker}. The components of a complete compatibility operator $K_1$
for the Killing operator $K$ can be interpreted, by the universality property, as a generating set for
all gauge-invariants, also known as a \emph{complete set of
gauge-invariant observables}.

It is well known that the Killing equation is of regular finite type
(Definition~\ref{def:fintype}), provided a precise regularity condition holds.
The quickest way to see that is to put it into the so-called
\emph{tractor form}~\cite{eastwood} or the form of the \emph{Killing
transport} equation~\cite[App.B]{geroch-killing}. Namely, we have the
equivalence up to homotopy
\begin{equation} \label{eq:killing-transport-equiv}
\setlength\arraycolsep{1pt}
\renewcommand{\arraystretch}{.8}
\begin{tikzcd}[column sep=4cm,row sep=4cm,every label/.append style =
{font = \small}]
	v^a \ar{r}{v^a \mapsto K_{ab}[v]}
		\ar[swap,shift right]{d}{v^a \mapsto \begin{bmatrix} v_a \\ \nabla_{[b} v_{c]} \end{bmatrix}} \&
	h_{ab} \ar[swap,shift right]{d}{h_{ab} \mapsto \begin{bmatrix} h_{a_1a}/2 \\ \nabla_{[b} h_{c]a_1} \end{bmatrix}}
		\ar[dashed,bend left]{l}{0}
	\\
	\begin{bmatrix} v_a \\ w_{[bc]} \end{bmatrix}
		\ar[swap]{r}{\begin{bmatrix} v_a \\ w_{[bc]} \end{bmatrix} \mapsto \TT_{a_1} \begin{bmatrix} v_a \\ w_{[bc]} \end{bmatrix}}
		\ar[swap,shift right]{u}{\begin{bmatrix} v_a \\ w_{[bc]} \end{bmatrix} \mapsto v^a} \&
	\begin{bmatrix} v_{a_1:a} \\ w_{a_1:[bc]} \end{bmatrix}
		\ar[swap,shift right]{u}{\begin{bmatrix} v_{a_1:a} \\ w_{a_1:[bc]} \end{bmatrix} \mapsto 2 v_{(a:b)}}
		\ar[swap,dashed,bend right]{l}{
			\begin{bmatrix} v_{a_1:a} \\ w_{a_1:[bc]}\end{bmatrix}
			\mapsto
			\begin{bmatrix}0  \\ -v_{[b:c]} \end{bmatrix}
		}
\end{tikzcd} ,
\end{equation}
where the connection operator%
	\footnote{The form of this connection was already derived
	in~\cite[Eq.(B.2)]{geroch-killing}, though there it has a typo. The
	sign of $R_{a_1 dbc}v^d$ is opposite compared to ours.} %
is
\begin{equation} \label{eq:killing-transport}
	\TT_{a_1} \begin{bmatrix}
		v_a \\ w_{[bc]}
	\end{bmatrix}
	= \begin{bmatrix}
		\nabla_{a_1} v_a - w_{[a_1 a]} \\
		\nabla_{a_1} w_{[bc]} + R_{a_1 d bc} v^d
	\end{bmatrix} ,
\end{equation}
which uses the Riemann tensor%
	\footnote{We follow the curvature conventions of~\cite{wald-gr}.}
$R_{abc}{}^d v_d = 2\nabla_{[a} \nabla_{b]} v_c$. We should mention that
the $:$ appearing for instance in $w_{a_1:[bc]}$ or $v_{[b:c]}$ is only
there to visually separate particular groups of indices. Also, we use
square brackets to indicate that some of the tensors must have
anti-symmetrized indices from the outset.

To check the commutativity of the square with downward arrows, first
note that for $w_{bc} = \nabla_{[b} v_{c]} = \nabla_b v_c
- \frac{1}{2} K_{bc}[v]$ we have
\begin{equation}
	2\nabla_{[a} w_{b]c} - R_{abcd} v^d = -\nabla_{[a} K_{b]c}[v] .
\end{equation}
Then, the antisymmetry $w_{cb} = -w_{bc}$ and the algebraic Bianchi
identity $R_{[abc]d} = 0$ allow us to write
\begin{align}
	2(\nabla_a w_{bc} + R_{adbc} v^d)
\notag
	&= (2\nabla_{[a} w_{b]c} - R_{abcd} v^d)
		- (2\nabla_{[b} w_{c]a} - R_{bcad} v^d)
		+ (2\nabla_{[c} w_{a]b} - R_{cabd} v^d) \\
	&= 2 \nabla_{[b} K_{c]a}[v] .
\end{align}
In general, the connection $\TT_a$ is not a flat. In fact, it is flat
iff $(M,g)$ is of constant curvature, $R_{abcd} = \alpha (g_{ac} g_{bd}
- g_{ad} g_{bc})$ for some constant $\alpha$. In general, the dimension
of the space of local solutions of $\TT [\begin{smallmatrix} v \\ w
\end{smallmatrix}] = 0$ need not be constant over $M$. However, when it
is (this is the needed regularity condition), the local solutions span a
vector bundle and the restriction of $\TT$ to this sub-bundle is flat
and the corresponding flat section equation is equivalent to the
original Killing equation (Lemma~\ref{lem:conn-flat-conn}), hence
implying its regular finite type. Thus, the required regularity
condition is that the local solutions of the Killing equation, in
tractor form, span a sub-bundle. Or, equivalently, the pointwise
dimension of the span of these local solutions is constant. However, in
special cases, this particular way of reducing the Killing operator to a
flat connection may not be the preferred one, and a different reduction
might be more convenient.

Consider a tensor $T[g]$ built covariantly out of the metric $g$, the
Riemann tensor $R$, and the covariant derivatives $\nabla R$, $\nabla
\nabla R$, \ldots. Define its linearization $\dot{T}$ about $g$ by the
identity $T[g + \eps h] = T[g] + \eps \dot{T}[h] + O(\eps^2)$. Recall
the standard identity between the Lie derivative, $T$ and $\dot{T}$:
\begin{gather}
\label{eq:lie-killing}
	\Lie_v T[g] = \dot{T}[\Lie_v[g]] = \dot{T}[K[v]] , \\
\label{eq:cov-lie}
	\text{where} \quad
	\Lie_v T_{a\cdots}^{b\cdots}
	= v^c \nabla_c T_{a\cdots}^{b\cdots}
		+ T_{c\cdots}^{b\cdots} \nabla_a v^c
		- T_{a\cdots}^{c\cdots} \nabla_c v^b
		+ \cdots \, ,
\end{gather}
which guarantees that $\dot{T}\circ K = 0$ for the linearization $g
\mapsto g + \eps h$ whenever $T[g] = 0$ or some expression involving
only constants and Kronecker $\delta$'s. This result is sometimes known
as the Stewart-Walker lemma~\cite[Lem.2.2]{stewart-walker}.
Alternatively, when $T[g] \ne 0$, this identity can be used to extract some components of
$v_a$ or $\nabla_a v_b$ by applying $\dot{T}$ to $K[v]$.

Section~\ref{sec:flrw} and~\ref{sec:schw} below will involve some algebraic constructions with
tensors for which it is convenient to introduce the following notation.
For $A_{ac}$ and $B_{bd}$ symmetric tensors,
we denote the \emph{Kulkarni-Nomizu} product by
\begin{equation} \label{eq:kn-prod}
	(A\odot B)_{abcd}
	= A_{ac} B_{bd}
	- A_{bc} B_{ad}
	- A_{ad} B_{bc}
	+ A_{bd} B_{ac} .
\end{equation}
Clearly $A\odot B = B \odot A$ and the result has the symmetry type of
the Riemann tensor. For tensors with two or four indices, we define the
contractions
\begin{equation} \label{eq:dot-con}
	(A\cdot B)_{ab}
	= A_{a}{}^c B_{cb} ,
	\quad \text{and} \quad
	(R \cdot S)_{abcd}
	= R_{ab}{}^{ef} S_{efcd} .
\end{equation}
With these definitions, when $A$, $B$, $C$ and $D$ are symmetric, we
have the useful identities
\begin{gather}
	[(A\odot B) \cdot (C\odot D)]_{abcd}
	= 2[(A\cdot C)\odot(B\cdot D) + (A\cdot D)\odot(B\cdot C)]_{abcd} , \\
	(A\odot B)_{ab}{}^b{}_d
	= [A\cdot B - (\tr A) B - A (\tr B) + B\cdot A]_{ad} .
\end{gather}

\subsection{Constant curvature spacetime} \label{sec:cc}

An $n$-dimensional \emph{constant curvature} spacetime $(M,g)$ of
sectional curvature $\alpha$ is defined by a Riemann curvature tensor of
the form $R_{abcd} = \alpha (g_{ac} g_{bd} - g_{ad} g_{bc})$, with
$\alpha$ a constant. It is well-known that in this case the
\emph{Killing transport} (or \emph{tractor}) connection defined
in~\eqref{eq:killing-transport} is actually flat. Thus, we could use the
methods of Section~\eqref{sec:compat}, in particular
Theorem~\ref{thm:fintype-compat}, to construct a compatibility complex
for the Killing operator $K$ on $(M,g)$. However, in this particular
situation, the compatibility complex for $K$ is already known
independently. It is sometimes called the \emph{Calabi
complex}~\cite{kh-calabi}. We will denote it by $C_i$, $i=0,1,\ldots$,
with $C_0 = K$ and $C_l = 0$ for $l\ge n$. The operator $C_1$ is
essentially the linearized Riemann tensor $R_{a}{}^b{}_c{}^d[g]$, while
$C_2$ is the linearized differential Bianchi identity. The remaining
operators $C_i$, for $i>2$, are essentially higher rank Bianchi
identities. These operators have the following explicit formulas
(see~\cite[Sec.2.2]{kh-calabi}):
\begin{subequations} \label{eq:calabi}
\begin{align}
	C_0[v]_{a:b} &= \nabla_a v_b + \nabla_b v_a , \\
	C_1[h]_{ab:cd} &= (\nabla\nabla \odot h)_{abcd}
		+ \alpha (g\odot h)_{abcd} , \\
	C_2[r]_{abc:de} &= 3 \nabla_{[a} r_{bc]de} , \\
	C_3[b]_{abcd:ef} &= 4 \nabla_{[a} b_{bcd]ef} , \\
	& \quad \vdots \\
	C_i[b]_{a_0\cdots a_i : bc} &= (i+1) \nabla_{[a_0} b_{a_1\cdots a_i] bc}
		\quad (i\ge 2) .
\end{align}
\end{subequations}
The $:$ notation only serves to visually separate groups of indices that
are independently antisymmetric. $C_0$ has the symmetric index pair
$a{:}b$, while $C_1$ has the index group $ab{:}cd$ satisfying the algebraic
symmetries of the Riemann tensor. More generally, the tensor symmetry
type of the target of each $C_i$ operator is best described using Young
symmetrizers (see~\cite[Sec.2.1]{kh-calabi} for complete details).
Ignoring corresponding algebraic symmetry conditions on the tensors
entering into the Calabi complex may violate its property of being a
compatibility complex. Below we list the Young symmetry types and ranks
of the tensor bundles serving as domains and codomains for the operators
of the Calabi complex:
\begin{center}
\begin{tabular}{c|c|cccc}
	& Young type & $n=2$ & $n=3$ & $n=4$ & $n\ge 2$ \\ \hline
	& $(1)$ & $2$ & $3$ & $4$ & $n$ \\
	$C_0$ & & \\
	& $(2)$ & $3$ & $6$ & $10$ & $\frac{n(n+1)}{2}$ \\
	$C_1$ & & \\
	& $(2,2)$ & $1$ & $6$ & $20$ & $\frac{n}{2} \binom{n+1}{3}$ \\
	$C_2$ & & \\
	& $(2,2,1)$ & & $3$ & $20$ & $\frac{n(3-1)}{2} \binom{n+1}{3+1}$ \\
	$\vdots$ & & & & & $\vdots$ \\
	& $(2,2,1^{i-2})$ & & & & $\frac{n(i-1)}{2} \binom{n+1}{i+1}$ \\
	\llap{($1<i$)\quad}
	$C_i$ & & \\
	& $(2,2,1^{i-1})$ & & & & $\frac{ni}{2} \binom{n+1}{i+2}$ \\
	$\vdots$ & & & & & $\vdots$ \\
	$C_{n-1}$ & & \\
	& $(2,2,1^{n-2})$ & $1$ & $3$ & $6$ & $\frac{n(n-1)}{2}$
\end{tabular}
\end{center}

In the diagram~\eqref{eq:killing-transport-equiv} for the equivalence up
to homotopy between $K$ and $\TT$, the top and bottom lines can be
extended to their compatibility complexes, the Calabi $C_l$ complex for
$K$ and the twisted de~Rham complex $d^\TT_l$
(Definition~\ref{def:flat-conn}) for $\TT$. Then, using the same argument
as at the beginning of the proof of Lemma~\ref{lem:lift-compat}, the
vertical equivalence maps can be propagated to the rest of the
complexes, thus giving a full equivalence up to homotopy between them

\begin{equation} \label{eq:calabi-equiv}
\begin{tikzcd}[column sep=large,row sep=large]
	\bullet
		\ar[swap,shift right]{d}
		\ar{r}{C_0 = K} \&
	\bullet
		\ar[dashed,bend left]{l}
		\ar[swap,shift right]{d}
		\ar{r}{C_1} \&
	{}
		\ar[dashed,bend left]{l}
		\ar[phantom,"\cdots"]{r} \&
	{}
		\ar{r}{C_{n-1}} \&
	\bullet
		\ar[dashed,bend left]{l}
		\ar[swap,shift right]{d}
		\ar{r}{0} \&
	\cdots
	\\
	\bullet
		\ar[swap,shift right]{u}
		\ar[swap]{r}{d^\TT_0 = \TT} \&
	\bullet
		\ar[dashed,bend right]{l}
		\ar[swap,shift right]{u}
		\ar[swap]{r}{d^\TT_1} \&
	{}
		\ar[dashed,bend right]{l}
		\ar[phantom,"\cdots"]{r} \&
	{}
		\ar[swap]{r}{d^\TT_{n-1}} \&
	\bullet
		\ar[dashed,bend right]{l}
		\ar[swap,shift right]{u}
		\ar{r}{0} \&
	\cdots
\end{tikzcd}
\end{equation}

We will not discuss the explicit formulas for the vertical equivalence
differential operators. For our purposes it is sufficient to know that
they exist. However, if these operators were to be given explicitly,
then according to Lemma~\ref{lem:compat-sufficient} the equivalence
diagram~\eqref{eq:calabi-equiv} would constitute an independent proof of
the fact that the operators $C_l$ constitute a compatibility complex for
the Killing operator $K$ on the constant curvature spacetime $(M,g)$.

\subsection{FLRW spacetimes} \label{sec:flrw}

Consider an FLRW spacetime $(M,g)$, where $M = I \times F$, with
$I\subset \mathbb{R}$ an open interval with coordinate $t$ and $\dim F =
m$, and $g = -dt^2 + f^2 \tilde{g}^F$, where \emph{the scale factor} $f
= f(t)$ is a positive scalar function and $\tilde{g}^F_{ab} = (\pi^*
g^F)_{ab}$ is the pullback of a constant curvature Riemannian metric
(with sectional curvature $\alpha$) on $F$ along the standard projection
$\pi\colon I\times F \to F$. Let us denote $U_a = -(dt)_a$, and note
that $U_a U^a = -1$ and $f^2 \tilde{g}^F_{ab} = g_{ab} + U_a U_b$.

Below, it will be convenient to extensively rely on the product
structure $M = I \times F$  and naturally decompose all tensors on $M$
with respect to it. For instance, $v^a_b = v_t^t U^a U_b + \tilde{v}^t_b
U^a + \tilde{v}^a_t U_b + \tilde{v}^a_b$, where $v_t^t$ is a scalar and
$\tilde{v}^a_t$, $\tilde{v}^a_b$ and $\tilde{v}^a_b$ are respectively
sections of the pullback bundles $\pi^* TF$, $\pi^* T^*F$ and
$\pi^*(T\otimes T^*)F$ over $M$, or equivalently simply sections $TM$,
$T^*M$ and $(T\otimes T^*)M$ that are annihilated by $g$-contraction
with $U_a$ on any index. It will also be convenient to insert extra
factors of $f$ into some tensor decompositions of this type, with the
sole purpose of simplifying some forthcoming formulas (which we have
unfortunately managed to do only in an ad-hoc way). If $X_a \in
\Secs(T^*F)$ then we will denote its pullback by $\tilde{X}_a = \pi^* X
\in \Secs(T^*M)$ and note that it satisfies $U^a \tilde{X}_a = 0$, where
the metric $g$ is used for contractions. On the other hand, if no such
$X_a$ was introduced previously (which will be true in most cases), we
used $\tilde{X}_a \in \Secs(T^*M)$ to denote any section that satisfies
$U^a \tilde{X}_a = 0$, even if it has non-trivial dependence on $t$. This
should not generate any confusion for the reader. The same convention is
extended to all purely covariant tensors.

Each of the factors has an auxiliary pseudo-Riemannian structure, $(I,
-dt^2)$ and $(F, g^F)$, with corresponding Levi-Civita connections,
$\nabla^I$ and $\nabla^F$, which we can extend to all covariant tensor
fields on the product manifold $M = I \times F$. Let us denote the
extensions of $\nabla^I$ and $\nabla^F$ respectively by $-U_a \del_t$ and $\tnabla$, which
incidentally defines the convenient differential operator $\del_t =
U^a(-U_a \del_t)$. The defining properties of these operators are the
usual Leibniz rule, together with $\del_t t = 1$, $\del_t U_a = 0$, and
$\del_t \tilde{X} = \del_t (\pi^* X) = 0$ for any covariant tensor field
$X$ on $F$, and also $\tnabla t = 0$, $\tnabla U_a = 0$ and $\tnabla
\tilde{X} = \widetilde{\nabla^F X} = \pi^*(\nabla^F X)$ for any
covariant tensor field $X$ on $F$. Since any covariant tensor on $M$ is
locally a limit of sums of products of covariant tensors pulled back
from $I$ and $M$, these properties are sufficient to uniquely define
$\del_t$ and $\tnabla$ as linear differential operators. Note that we
will also frequently use the abbreviation $(-)' = \del_t (-)$.

By the above remarks, any covector field $v_a$ on $M$ can be
parametrized as
\begin{equation} \label{eq:flrw-v-param}
	v_a = A f U_b + f^2 \tX_a
\end{equation}
where $A$ is a scalar function on $M$ and $\tX\in \Secs(\pi^* T^*F)
\subset \Secs(T^*M)$, so that $U^a \tX_a = 0$. Now note that with our
conventions, for any scalar $A$, its exterior derivative is given by
$(dA)_a = -A' U_a + \tnabla_a A$, while the Levi-Civita connection on
$(M,g)$ is given in the above parametrization by
\begin{equation}
	\nabla_a (A f U_b + f^2 \tX_b)
	= (dA)_a f U_b + A f' g_{ab}
		- 2 f' f U_{[a} \tX_{b]}
		- f^2 U_a \tX'_b
		+ f^2 \tnabla_a \tX_b .
\end{equation}
Parametrizing symmetric 2-tensors $h_{ab}$ on $M$ as
\begin{equation} \label{eq:flrw-h-param}
	h_{ab} = p U_a U_b - 2 f^2 U_{(a} \tilde{Y}_{b)} + f^2 \tilde{Z}_{ab} ,
\end{equation}
the Killing operator on covectors $h_{ab} = K_{ab}[v] = 2 \nabla_{(a}
v_{b)}$ becomes
\begin{equation} \label{eq:flrw-killing}
	\begin{bmatrix} p \\ \cmidrule(lr){1-1} \tilde{Y} \\ \tilde{Z} \end{bmatrix}
	= K \begin{bmatrix}
		A \\ \cmidrule(lr){1-1}
		\tilde{X}
	\end{bmatrix}
	= \begin{bmatrix}
		-2(Af)' \\
		\cmidrule(lr){1-1}
		\tX' - f^{-1} \tnabla A \\
		\tilde{K}[\tX] + 2 A f' \tilde{g}^F
	\end{bmatrix}
	= \begin{bmatrix}[c|c]
		-2\del_t f & 0 \\
		\cmidrule(lr){1-1} \cmidrule(lr){2-2}
		-f^{-1} \tnabla & \del_t \\
		2f'\tilde{g}^F & \tilde{K}
	\end{bmatrix}
	\begin{bmatrix}
		A \\ \cmidrule(lr){1-1}
		\tilde{X}
	\end{bmatrix}
\end{equation}
in block matrix operator notation, where $\tilde{K}_{bc}[\tilde{X}] =
2\tnabla_{(b} \tX_{c)}$, which is simply our extension of the Killing
operator on covectors from $(F,g^F)$ to $(M,g)$.

\begin{rem} \label{rem:block-matrix}
Each entry in our block operator matrices is a linear differential
operator between some covariant tensor bundles over $M$. We use vertical
and horizontal lines to further partition block operator matrices with
respect to some special direct sum decomposition of their domain or
codomain. We will use $\id$ to denote the identity endomorphism on any
vector bundle and $0$ to denote the zero morphism between any two vector
bundles. In each case, the domains and codomains of these operators can
be deduced from the context.
\end{rem}

It is worth noting that setting the $A$ component of $v_a$ to zero
simplifies the Killing operator to
\begin{equation} \label{eq:simp-killing-flrw}
	K \begin{bmatrix} 0 \\ \cmidrule(lr){1-1} \id \end{bmatrix}
	= \begin{bmatrix}
		0 \\ \cmidrule(lr){1-1} \del_t \\ \tilde{K}
	\end{bmatrix} .
\end{equation}

For a generic FLRW spacetime (see~\cite[Def.2.1]{cdk} for a breakdown of
FLRW geometries into special and generic classes, based on the properties of the scale
factor $f$), it is well-known that the only Killing vectors are those
that reduce to the Killing vectors of the spatial slices $(F,g^F)$,
appropriately propagated in time. We will see shortly that,
equivalently, each Killing vector on $(M,g)$ has the form $v_a = 0 + f^2
\tX_a$, where $\tilde{K}_{ab}[\tX] = 0$ and $\del_t \tX_a = 0$. Now,
since the spatial slices $(F,g^F)$ are of constant curvature, the
spatial Killing operator $\tilde{K}$ is of the type discussed in
Section~\ref{sec:cc}. This means that $\tilde{K}$ is equivalent up to
homotopy to the flat spatial Killing transport connection $\tilde{\TT}$,
so that the following two operators are also equivalent up to homotopy:
\begin{equation}
	\begin{bmatrix} \del_t \\ \tilde{K} \end{bmatrix}
	\quad \text{and} \quad
	\TT = \begin{bmatrix} \del_t \\ \tilde{\TT} \end{bmatrix} .
\end{equation}
Since both $[\del_t,\tilde{K}] = 0$ and $[\del_t,\tilde{\TT}] = 0$, it
is easy to see that $\TT$ itself defines a flat connection. Let the
operators $\tC_i$ be the extensions of the Calabi
complex~\eqref{eq:calabi} from the constant curvature geometry $(F,g^F)$
to $M$, where we have simply replaced $\nabla^F$ by $\tnabla$ and $g^F$
by $\tilde{g}^F$ whenever necessary. Since such an extension preserves
all operator identities, including the suitably extended equivalence up
to homotopy in~\eqref{eq:calabi-equiv}. Now, it is a straightforward
exercise to check that the twisted de~Rham complex $d^\TT_l$ can be
represented as the bottom line in the following diagram, and also to use
the mentioned identities to construct the corresponding the vertical
differential operators that complete this diagram to an equivalence up
to homotopy:
\begin{equation} \label{eq:flrw-calabi}
\begin{tikzcd}[column sep=2.2cm,row sep=huge]
	\bullet
		\ar[swap,shift right]{d}
		\ar{r}{\begin{bmatrix}\del_t \\ \tC_0 = \tilde{K}\end{bmatrix}} \&
	\bullet
		\ar[dashed,bend left]{l}
		\ar[swap,shift right]{d}
		\ar{r}{\begin{bmatrix}-\tC_0 & \del_t \\ 0 & \tC_1\end{bmatrix}} \&
	{}
		\ar[dashed,bend left]{l}
		\ar[phantom,"\cdots"]{r} \&
	{}
		\ar{r}{\begin{bmatrix}-\tC_{m-1} & \del_t\end{bmatrix}} \&
	\bullet
		\ar[dashed,bend left]{l}
		\ar[swap,shift right]{d}
		\ar{r}{0} \&
	\cdots
	\\
	\bullet
		\ar[swap,shift right]{u}
		\ar[swap]{r}{d^\TT_0 = \TT = \begin{bmatrix}\del_t \\ \tilde{\TT}\end{bmatrix}} \&
	\bullet
		\ar[dashed,bend right]{l}
		\ar[swap,shift right]{u}
		\ar[swap]{r}{d^\TT_1 = \begin{bmatrix}-d^{\tilde{\TT}}_0 & \del_t \\ 0 & d^{\tilde{\TT}}_1 \end{bmatrix}} \&
	{}
		\ar[dashed,bend right]{l}
		\ar[phantom,"\cdots"]{r} \&
	{}
		\ar[swap]{r}{d^\TT_{m-1} = \begin{bmatrix}-d^{\tilde{\TT}}_{m-1} & \del_t\end{bmatrix}} \&
	\bullet
		\ar[dashed,bend right]{l}
		\ar[swap,shift right]{u}
		\ar{r}{0} \&
	\cdots
\end{tikzcd}
\end{equation}
Hence, by Lemma~\ref{lem:compat-sufficient}, the top complex
in~\eqref{eq:flrw-calabi} is also a compatibility complex.

Now, we need to examine the integrability conditions that will help us
establish an explicit equivalence of the full Killing equation
$K_{ab}[v]=0$ with the system of equations $\del_t \tX_a = 0$,
$\tilde{K}_{ab}[\tX] = 0$, whose compatibility complex is given by the
top line of~\eqref{eq:flrw-calabi}. All of that crucially depends on the
structure of the curvature of $(M,g)$.

The Riemann curvature tensor, the Ricci tensor and the Ricci scalar of
$(M,g)$ are (recalling the notation from~\eqref{eq:kn-prod}) given by
\begin{subequations} \label{eq:flrw-curv}
\begin{align}
	R_{abcd} &= \left( g \odot \left[
		\frac{1}{2} \left(\frac{f'^2}{f^2}+\frac{\alpha}{f^2}\right) g
		- \left(\left(\frac{f'}{f}\right)'-\frac{\alpha}{f^2}\right) UU \right] \right)_{abcd} , \\
\notag
	R_{ab} &= -(m-1) \left[\left(\frac{f'}{f}\right)'-\frac{\alpha}{f^2}\right] U_a U_b \\
	& \quad {}
		+ \left[\left(\left(\frac{f'}{f}\right)'-\frac{\alpha}{f^2}\right)
		+ m\left(\frac{f'^2}{f^2}+\frac{\alpha}{f^2}\right)\right] g_{ab} , \\
	\R &= \left(\left(\frac{f'}{f}\right)' - \frac{\alpha}{f^2}\right)
		+ (m+1) \left(\frac{f'^2}{f^2} + \frac{\alpha}{f^2}\right)
	.
\end{align}
\end{subequations}
We suppose that the FLRW spacetime is \emph{non-degenerate}, that is
both that $f'/f \ne 0$ and that the scalar curvature $\R$ is not
constant,
\begin{equation}
	\R'
	= \left[\left(\left(\frac{f'}{f}\right)' - \frac{\alpha}{f^2}\right)
		+ (m+1) \left(\frac{f'^2}{f^2} + \frac{\alpha}{f^2}\right)\right]'
	\ne 0 .
\end{equation}
To make use of identity~\eqref{eq:lie-killing}, we compute the Lie
derivative
\begin{equation}
	\Lie_v \R = v^a \nabla_a \R = A f \R' .
\end{equation}
Thus, defining the operator
\begin{equation} \label{eq:flrw-J}
	J[h] = \frac{1}{f\R'} \dot{\R}[h] ,
\end{equation}
we have the identities
\begin{equation} \label{eq:flrw-J-K}
	J\circ K
		\begin{bmatrix} A \\ \cmidrule(lr){1-1} \tX \end{bmatrix} = A
	\quad \text{or} \quad
	J\circ K = \begin{bmatrix}[c|c] \id & 0 \end{bmatrix} .
\end{equation}
The last equation also implies that
\begin{equation}
	J \circ \left( K \begin{bmatrix} 0 \\ \cmidrule(lr){1-1} \id \end{bmatrix} \right)
	= \begin{bmatrix}[c|c] \id & 0 \end{bmatrix}
		\begin{bmatrix} 0 \\ \cmidrule(lr){1-1} \id \end{bmatrix}
	= 0
	\quad \text{and} \quad
	J \circ \left( K \begin{bmatrix} \id \\ \cmidrule(lr){1-1} 0 \end{bmatrix} \right)
	= \begin{bmatrix}[c|c] \id & 0 \end{bmatrix}
		\begin{bmatrix} \id \\ \cmidrule(lr){1-1} 0 \end{bmatrix}
	= \id .
\end{equation}
Let us denote the block matrix components of $J$ as $J =
\begin{bmatrix}[c|cc] J^p & J^Y & J^Z \end{bmatrix}$ and also $\tilde{J}
= \begin{bmatrix} J^Y & J^Z \end{bmatrix}$. Combining the above
identities with formula~\eqref{eq:simp-killing-flrw}, we get
\begin{equation} \label{eq:flrw-tJ}
	\tilde{J} \circ
		\left(K \begin{bmatrix} 0 \\ \cmidrule(lr){1-1} \id \end{bmatrix}\right)
	= \begin{bmatrix} J^Y & J^Z \end{bmatrix}
		\begin{bmatrix} \del_t \\ \tK \end{bmatrix}
	= 0 .
\end{equation}
Hence, knowing the first compatibility operator from the top line
of~\eqref{eq:flrw-calabi}, it must be possible to factor
\begin{equation} \label{eq:flrw-tHJ}
	\tilde{J} = \tilde{H}_J \begin{bmatrix}
		-\tC_0 & \del_t \\
		0 & \tC_1
	\end{bmatrix} .
\end{equation}
Of course, the full operator $J$ can be then factored as
\begin{equation} \label{eq:flrw-HJ}
	J = H_J \begin{bmatrix}[c|cc]
		\id & 0 & 0 \\
		\cmidrule(lr){1-1} \cmidrule(lr){2-3}
		0 & -\tC_0 & \del_t \\
		0 & 0 & \tC_1
	\end{bmatrix} ,
\end{equation}
where $H_J = \begin{bmatrix}[c|c] J^p & \tilde{H}_J \end{bmatrix}$. For
much of what follows, we will only need the fact that $\tilde{H}_J$ or
$H_J$ exists. However, a direct calculation shows that its explicit form
can be deduced from the identity
\begin{multline} \label{eq:flrw-curv-scalar}
	\dot{\R} \begin{bmatrix} p \\ \cmidrule(lr){1-1} \tY \\ \tZ \end{bmatrix}
	= f^{-2} \tilde{\Delta} p + m(f/f')[(f'^2/f^2) p]' + m(m+1) (f'^2/f^2) p \\
		- f^{-m-1} [ f^{m+1} (2\tilde{\div}\, \tY - \tilde{\tr} \tZ') ]'
		+ f^{-2} [-\tilde{\Delta} \, \tilde{\tr} \, \tZ
			+ \tilde{\div} \, \tilde{\div} \, \tZ_{ab}
			- (m-1) \, \alpha \, \tilde{\tr} \, \tZ] \\
	= f^{-2} \tilde{\Delta} p + \frac{m}{f'f^m} [f^{m+1} (f'^2/f^2) p]'
		+ f^{-m-1} [ f^{m+1} \tilde{\tr} (\del_t \tZ - \tC_0[\tY]) ]'
		- \frac{1}{2} f^{-2} \tilde{\tr} \, \tilde{\tr} \, \tC_1[\tZ]
	,
\end{multline}
meaning that
\begin{align}
\label{eq:flrw-J-expl}
	J &= \frac{1}{f\R'} \begin{bmatrix}[c|cc]
			f^{-2} \tilde{\Delta} + \frac{m}{f'f^m}\del_t f^{m+1} \frac{f'^2}{f^2} &
			-\frac{1}{f^{m+1}} \del_t f^{m+1} \tC_0 &
			\frac{1}{f^{m+1}} \del_t f^{m+1} \tilde{\tr}
				-\frac{1}{2} f^{-2} \tilde{\tr} \, \tilde{\tr} \tC_1
		\end{bmatrix} , \\
	\intertext{and hence}
\label{eq:flrw-HJ-expl}
	H_J &= \frac{1}{f\R'} \begin{bmatrix}[c|cc]
		f^{-2} \tilde{\Delta} + \frac{m}{f'f^m}\del_t f^{m+1} \frac{f'^2}{f^2} &
		\frac{1}{f^{m+1}} \del_t f^{m+1} \, \tilde{\tr} & 
		-\frac{1}{2} f^{-2} \, \tilde{\tr} \, \tilde{\tr}
	\end{bmatrix} ,
\end{align}
where of course we have defined $\tilde{\tr}$, $\tilde{\div}$ and
$\tilde{\Delta}$ such that
\begin{equation} \label{eq:flrw-lapdivtr}
	\tilde{\Delta} \tX = \widetilde{(g^F)^{ab} \nabla_a \nabla_b X} ,
	\quad
	\tilde{\div} \, \tY = \widetilde{\nabla^a Y_{a\cdots}} ,
	\quad
	(\tilde{\tr} \, \tilde{X})_{ab} = \widetilde{(g^F)^{cd} X_{acbd}} ,
	\quad \text{and} \quad
	\tilde{\tr} \, \tilde{Z} = \widetilde{(g^F)^{ab} Z_{ab}} .
\end{equation}

Now we are ready to follow the proof of Theorem~\ref{thm:fintype-compat}
to construct a compatibility complex for the Killing operator $K$ by
lifting the compatibility complex from~\eqref{eq:flrw-calabi}. The
results of these calculations will be presented below directly in
diagrammatic form, where the arrows in the diagrams satisfy the
identities introduced in Section~\ref{sec:compat}. All the relevant
identities are easily checked by direct calculation, relying on the key
identity~\eqref{eq:flrw-J-K}, the basic commutation relations $[\del_t,
\tnabla] = [\del_t, \tC_i] = 0$, the compatibility identities $\tC_{i+1}
\circ \tC_i = 0$ of the operators of the Calabi complex, which were
introduced in Section~\ref{sec:cc}.

We start by applying the information obtained above to give an explicit
reduction of the Killing equation to the first operator from the top
line of~\eqref{eq:flrw-calabi}:

\begin{equation} \label{eq:flrw-lift1}
\begin{tikzcd}[column sep=5cm,row sep=4.5cm]
	\bullet
		\ar{r}{K = \begin{bmatrix}[c|c]
			-2\del_t f & 0 \\
			\cmidrule(lr){1-1} \cmidrule(lr){2-2}
			-f^{-1} \tnabla & \del_t \\
			2f'\tilde{g}^F & \tilde{K}
		\end{bmatrix}}
		\ar[swap,shift right]{d}{\begin{bmatrix}[c|c]0 & \id\end{bmatrix}} \&
	\bullet
		\ar[dashed,bend left]{l}{\begin{bmatrix}J \\ \cmidrule(lr){1-1} 0\end{bmatrix}}
		\ar[swap,shift right]{d}{\begin{bmatrix}[c|cc]
			0 & \id & 0 \\
			0 & 0 & \id 
		\end{bmatrix} \left( \id - K \begin{bmatrix} J \\ \cmidrule(lr){1-1} 0 \end{bmatrix} \right)}
	\\
	\bullet
		\ar[swap]{r}{\begin{bmatrix} \del_t \\ \tC_0 = \tilde{K} \end{bmatrix}}
		\ar[swap,shift right]{u}{\begin{bmatrix}0 \\ \cmidrule(lr){1-1} \id\end{bmatrix}} \&
	\bullet
		\ar[swap,dashed,bend right]{l}{0}
		\ar[swap,shift right]{u}{\begin{bmatrix}0 & 0 \\ \cmidrule(lr){1-2} \id & 0 \\ 0 & \id\end{bmatrix}}
\end{tikzcd}
\end{equation}

Next, we proceed by iterating the construction from
Lemma~\ref{lem:lift-compat}, while simultaneously applying the
simplifications discussed after the proof of
Theorem~\ref{thm:fintype-compat}. The following diagram should be
appended on the right to~\eqref{eq:flrw-lift1}:

\begin{equation} \label{eq:flrw-lift2}
\begin{tikzcd}[column sep=5cm,row sep=4.8cm]
	\bullet
		\ar{r}{\begin{bmatrix}[c|cc]
			\id & 0 & 0 \\
			\cmidrule(lr){1-1} \cmidrule(lr){2-3} 
			0 & -\tC_0 & \del_t \\
			0 & 0 & \tC_1
		\end{bmatrix} \left( \id - K \begin{bmatrix} J \\ \cmidrule(lr){1-1} 0 \end{bmatrix} \right)
		= \left(
			\id
			- \begin{bmatrix}[c|cc]
				\id & 0 & 0 \\
				\cmidrule(lr){1-1} \cmidrule(lr){2-3} 
				0 & -\tC_0 & \del_t \\
				0 & 0 & \tC_1
			\end{bmatrix}
			K \begin{bmatrix} H_J \\ \cmidrule(lr){1-1} 0 \end{bmatrix}
		\right)
		\begin{bmatrix}[c|cc]
			\id & 0 & 0 \\
			\cmidrule(lr){1-1} \cmidrule(lr){2-3} 
			0 & -\tC_0 & \del_t \\
			0 & 0 & \tC_1
		\end{bmatrix}
		}
		\ar[swap,shift right]{d}
		\&
	\bullet
		\ar[swap,shift right]{d}{\begin{bmatrix}[c|cc]0 & \id & 0 \\ 0 & 0 & \id\end{bmatrix}}
		\ar[dashed,bend left]{l}{\begin{bmatrix}[c|cc]
			\id & 0 & 0 \\
			\cmidrule(lr){1-1} \cmidrule(lr){2-3} 
			0 & 0 & 0 \\
			0 & 0 & 0
		\end{bmatrix}}
	\\
	\bullet
		\ar[swap]{r}{\begin{bmatrix} -\tC_0 & \del_t \\ 0 & \tC_1 \end{bmatrix}}
		\ar[swap,shift right]{u} \&
	\bullet
		\ar[swap,shift right]{u}{
			\left( \id
			- \begin{bmatrix}[c|cc]
				\id & 0 & 0 \\
				\cmidrule(lr){1-1} \cmidrule(lr){2-3} 
				0 & -\tC_0 & \del_t \\
				0 & 0 & \tC_1
			\end{bmatrix}
			K \begin{bmatrix} H_J \\ \cmidrule(lr){1-1} 0 \end{bmatrix}
			\right)
			\begin{bmatrix}
				0 & 0 \\
				\cmidrule(lr){1-2} 
				\id & 0 \\
				0 & \id
			\end{bmatrix}
			}
		\ar[swap,dashed,bend right]{l}{
			\begin{bmatrix}[c|cc]
				0 & \id & 0 \\
				0 & 0 & \id
			\end{bmatrix} K
			\begin{bmatrix} \tH_J \\ \cmidrule(lr){1-1} 0 \end{bmatrix}
		}
\end{tikzcd}
\end{equation}
Note that we do not repeat the labels on the left-most vertical arrows,
which can be read off as the right-most vertical arrows
in~\eqref{eq:flrw-lift1}.

Two more iterations of Lemma~\ref{lem:lift-compat} (with simultaneous
simplifications) gives the following diagram, to be appended on the
right to~\eqref{eq:flrw-lift2}:

\begin{equation} \label{eq:flrw-lift3}
\begin{tikzcd}[column sep=5cm,row sep=5cm]
	\bullet
		\ar[swap,shift right]{d}
		\ar{r}{\begin{bmatrix}[c|cc]
			\multicolumn{3}{c}{H_J} \\
			\cmidrule(lr){1-3}
			0 & -\tC_1 & \del_t \\
			0 & 0 & \tC_2
		\end{bmatrix}} \&
	\bullet
		\ar[dashed,bend left]{l}{\begin{bmatrix}[c|cc]
			\id & 0 & 0 \\
			\cmidrule(lr){1-1} \cmidrule(lr){2-3} 
			0 & -\tC_0 & \del_t \\
			0 & 0 & \tC_1
		\end{bmatrix} K \begin{bmatrix}[c|cc]
			\id & 0 & 0 \\
			\cmidrule(lr){1-1} \cmidrule(lr){2-3}
			0 & 0 & 0
		\end{bmatrix}}
		\ar[swap,shift right]{d}{\begin{bmatrix}[c|cc]0 & \id & 0 \\ 0 & 0 & \id\end{bmatrix}}
		\ar{r}{\begin{bmatrix}[c|cc] 0 & -\tC_2 & \del_t \\ 0 & 0 & \tC_3 \end{bmatrix}} \&
	\bullet
		\ar[swap,shift right]{d}{\begin{bmatrix}\id & 0\\ 0 & \id\end{bmatrix}}
		\ar[dashed,bend left]{l}{0}
	\\
	\bullet
		\ar[swap,shift right]{u}
		\ar[swap]{r}{\begin{bmatrix} -\tC_1 & \del_t \\ 0 & \tC_2 \end{bmatrix}} \&
	\bullet
		\ar[swap,dashed,bend right]{l}{0}
		\ar[swap,shift right]{u}{\begin{bmatrix}0 & 0\\ \cmidrule(lr){1-2} \id & 0\\ 0 & \id\end{bmatrix}}
		\ar[swap]{r}{\begin{bmatrix} -\tC_2 & \del_t \\ 0 & \tC_3 \end{bmatrix}} \&
	\bullet
		\ar[swap,shift right]{u}{\begin{bmatrix}\id & 0\\ 0 & \id\end{bmatrix}}
		\ar[swap,dashed,bend right]{l}{0}
\end{tikzcd}
\end{equation}

From this point on, the compatibility complex for $K$ and the top line
of~\eqref{eq:flrw-calabi} become identical.

\begin{thm} \label{thm:flrw}
Consider a non-degenerate FLRW spacetime $(M,g)$, $M = I\times F$,
as introduced at the top of Section~\ref{sec:flrw}, which spatially has
the structure of an $m$-dimensional constant curvature space $(F,g^F)$,
with sectional curvature $\alpha$. The full compatibility complex
$K_i$ for the Killing operator $K_0 = K$~\eqref{eq:flrw-killing} is
given by
\begin{subequations} \label{eq:flrw-complex}
\begin{align}
	K_0
	&= \begin{bmatrix}[c|c]
			-2\del_t f & 0 \\
			\cmidrule(lr){1-1} \cmidrule(lr){2-2} 
			-f^{-1} \tnabla & \del_t \\
			2f'\tilde{g}^F & \tilde{K}
		\end{bmatrix} , \\
	K_1
	&= \begin{bmatrix}[c|cc]
			\id & 0 & 0 \\
			\cmidrule(lr){1-1} \cmidrule(lr){2-3} 
			0 & -\tC_0 & \del_t \\
			0 & 0 & \tC_1
		\end{bmatrix}
		\left( \id - K \begin{bmatrix} J \\ \cmidrule(lr){1-1} 0 \end{bmatrix} \right) , \\
	K_2
	&= \begin{bmatrix}[c|cc]
			\multicolumn{3}{c}{H_J} \\
			\cmidrule(lr){1-3}
			0 & -\tC_1 & \del_t \\
			0 & 0 & \tC_2
		\end{bmatrix} , \\
	K_3
	&= \begin{bmatrix}[c|cc]
			0 & -\tC_2 & \del_t \\
			0 & 0 & \tC_3
		\end{bmatrix} , \\
	K_{i}
	&= \begin{bmatrix}
			-\tC_{i-1} & \del_t \\
			0 & \tC_i
		\end{bmatrix} \quad (3<i<m) , \\
	K_m
	&= \begin{bmatrix}
			-\tC_{m-1} & \del_t
		\end{bmatrix} , \\
	K_{i} &= 0  \quad (m<i) ,
\end{align}
\end{subequations}
where the operator $H_J$ is defined in~\eqref{eq:flrw-HJ}, $\del_t$ and
$\tnabla$ are the covariant derivatives pulled back along the product
structure $t\colon I\times F \to I$ and $I\times F \to F$, while $\tC_i$
are the operators from the Calabi complex associated to the constant
curvature space $(M,g^F)$, as introduced in Section~\ref{sec:cc}. (See
Appendix~\ref{sec:flrw-notation} for a more complete summary of the
notation.)
\end{thm}

\begin{proof}
The argument given around diagram~\eqref{eq:flrw-calabi} shows that its
top line constitutes a full compatibility complex, which coincides
with the bottom line of the diagram obtained by gluing (from left to
right) the diagrams~\eqref{eq:flrw-lift1}, \eqref{eq:flrw-lift2} and
\eqref{eq:flrw-lift3}, which are continued by identifying the top and
bottom rows. From the preceding discussion in the current section, it is
clear that each pair of consecutive squares in this glued diagram
satisfies the hypotheses of Lemma~\ref{lem:compat-sufficient}. Thus, the
top line of this glued diagram is itself a full compatibility
complex, but that complex consists precisely of the operators $K_i$
in~\eqref{eq:flrw-complex}.
\end{proof}

The non-vanishing ranks of the vector bundles in the $K_i$ complex have
the following pattern, which can be compared to a similar table for the
constant curvature case at the end of Section~\ref{sec:cc} (where
$n=m+1$, for comparison):

\begin{center}
\begin{tabular}{c|cccc}
	& $m=2$ & $m=3$ & $m=4$ & $m\ge 2$ \\ \hline
	& $3$ & $4$ & $5$ & $m+1$ \\
	$K_0$ & \\
	& $6$ & $10$ & $15$ & $\frac{m(m+1)}{2} + m + 1$ \\
	$K_1$ & \\
	& $5$ & $13$ & $31$ & $\frac{m}{2}\binom{m+1}{3}+\frac{m(m+1)}{2}+1$ \\
	$K_2$ & \\
	& $2$ & $10$ & $41$ & $m\binom{m+1}{4}+\frac{m}{2}\binom{m+1}{3}+1$ \\
	$K_3$ & \\
	& & $3$ & $26$ & $\frac{3m}{2}\binom{m+1}{5}+m\binom{m+1}{4}$ \\
	$\vdots$ & & & & $\vdots$ \\
	& & & & $\frac{m(i-1)}{2}\binom{m+1}{i+1}+\frac{m(i-2)}{2}\binom{m+1}{i}$ \\
	\llap{($3<i$)\quad}
	$K_i$ & \\
	& & & & $\frac{mi}{2}\binom{m+1}{i+2}+\frac{m(i-1)}{2}\binom{m+1}{i+1}$ \\
	$\vdots$ & & & & $\vdots$ \\
	$K_m$ & \\
	& $2$ & $3$ & $6$ & $\frac{m(m-1)}{2}$
\end{tabular}
\end{center}

\begin{rem}
In $n=4$ ($m=3$) dimensions, the common choice for gauge-invariant
variables on cosmological FLRW spacetimes are the so-called
\emph{Bardeen potentials}~\cite{bardeen}. They include two scalars
components, $\Phi$ and $\Psi$, a divergence-free spatial vector field,
$\hat{\Phi}_a$, with $3$ independent components, and a divergence-free
trace-free spatial symmetric 2-tensor, $\hat{E}_{ab}$, with $5$
independent components. Hence, the total number of independent
components (not counting the differential relations coming from the
divergence-free conditions) is $1+1+3+5 = 10$, which is less than the
$13$ components of $K_1$ that we have counted above. The difference of
course between our $K_1$ and the Bardeen potentials is that our
expressions are all local (given only in terms of differential
operators), while the Bardeen potentials are non-local (their definition
involves inverted spatial Laplacians; see the Introduction to~\cite{fhk}
for details). While we cannot claim that our construction gives the
minimal number of local gauge-invariant quantities, it is not surprising
that we get a larger number than a construction that allows non-local
expressions.
\end{rem}

\begin{rem} \label{rem:flrw-K1}
It is worth noting that the $K_i$ complex presented above is not
continuously deformable through the class of generic FLRW spacetimes to
the constant or zero curvature cases, which correspond to special
choices of the scale factor $f(t)$. The main reason is that the operator
$J$ introduced in Equation~\eqref{eq:flrw-J} is proportional to $1/\R'$,
which diverges when the background scalar curvature becomes constant.
Since the operator $J$ and the related operator $H_J$ appear in several
places in the formulas~\eqref{eq:flrw-complex} for the operators $K_i$
(for $i\ge 1$), it seems difficult to directly compare them with the
corresponding operators $C_i$ on constant curvature backgrounds
in~\eqref{eq:calabi}. In particular, an explicit expression for the
linearized Riemann tensor $C_1$ on $(M,g)$ would be rather long and
unenlightening in our notation, as can already be glimpsed from the
formula for the linearized scalar curvature
in~\eqref{eq:flrw-curv-scalar}. A more fruitful comparison would be to
try to express the components of our $K_1$ in terms of the linearized
IDEAL characterization tensors that were recently constructed for the
FLRW geometries~\cite{cdk}.

However, we might gain some qualitative insight into the components of
$K_1$, which can be interpreted as a complete set of local
gauge-invariant observables, from the alternative formula that was
indicated in diagram~\eqref{eq:flrw-lift2}:
\begin{equation}
	\begin{bmatrix}
		q \\ \cmidrule(lr){1-1}
		\tY_{a:b} \\
		\tZ_{ab:cd}
	\end{bmatrix}
	=
	K_1 \begin{bmatrix} p \\ \cmidrule(lr){1-1} \tY_a \\ \tZ_{a:b} \end{bmatrix}
	= \left(
			\id
			- \begin{bmatrix}[c|cc]
				\id & 0 & 0 \\
				\cmidrule(lr){1-1} \cmidrule(lr){2-3} 
				0 & -\tC_0 & \del_t \\
				0 & 0 & \tC_1
			\end{bmatrix}
			K \begin{bmatrix} H_J \\ \cmidrule(lr){1-1} 0 \end{bmatrix}
		\right)
		\begin{bmatrix}[c|cc]
			\id & 0 & 0 \\
			\cmidrule(lr){1-1} \cmidrule(lr){2-3} 
			0 & -\tC_0 & \del_t \\
			0 & 0 & \tC_1
		\end{bmatrix}
		\begin{bmatrix} p \\ \cmidrule(lr){1-1} \tY_a \\ \tZ_{ab} \end{bmatrix} .
\end{equation}
The $\tZ_{ab:cd}$ invariants are roughly coming from the linearized
Riemann operator $\tC_1$ on the spatial slice $(F,g^F)$. While these
components can be obtained from a purely spatial projection of the
linearized Riemann operator on the FLRW spacetime $(M,g)$, since the
components of the linearized spacetime Riemann tensor are not by
themselves invariant, they need to be deformed through the $H_J$ and
subsequent operators in the above formula to be truly invariant. The $q$
scalar invariant comes from the difference between the $p$ component of
$h_{ab}$ and the $p$ component of $K[J[h]]_{ab}$. When precomposed with
$h_{ab} = K[v]_{ab}$, both terms in the difference depend only on the
$A$ component of $v_a$ and in exactly the same way, hence cancelling to
give an invariant quantity. The $\tY_{a:b}$ components are harder to
interpret in familiar terms.
\end{rem}

\subsection{Schwarzschild-Tangherlini spacetimes} \label{sec:schw}

Consider an $n$-dimensional spacetime $(\bar{\M},\bar{g})$ where
$\bar{\M} = \M\times \S$, where $\dim \M = 2$ and $\dim \S =
n-2$~\cite{kis-n+m, ki-master, ki-stab}. Let Greek indices
$\mu,\nu,\ldots$ correspond to tensors om $\bar{\M}$. We take the second
factor $(\S,\Omega)$ to be a maximally symmetric space, hence a constant
curvature Riemannian space with sectional curvature $\alpha$, where
$\alpha = 1$ for a unit sphere, $\alpha = 0$ for Euclidean space, and
$\alpha = -1$ for hyperbolic space (or pseudo-sphere). Let upper case
Latin indices $A,B,\ldots$ correspond to tensors on $\S$. Let us denote
local coordinates on $\S$ by $\theta^A$ and the Levi-Civita connection
on $(\S,\Omega)$ by $D_A$; its curvature tensor is then
\begin{equation}
	R_{ABCD} = \alpha (\Omega_{AC} \Omega_{BD} - \Omega_{AD} \Omega_{BC})
	= \frac{\alpha}{2} (\Omega \odot \Omega)_{ABCD} ,
\end{equation}
where we have used the Kulkarni-Nomizu product~\eqref{eq:kn-prod}. The
other factor $(\M,g)$ has signature $({-}{+})$. Let lower case Latin
indices $a,b,\ldots$ correspond to tensors on $\M$, which we will presume
has a timelike Killing vector $t^a$. Let us denote local coordinates
on $\M$ by $y^a$ and the Levi-Civita connection on $(\M,g)$ by
$\nabla_a$. Because $\dim \M = 2$, its curvature is given by
\begin{equation} \label{eq:schw-curv-rt}
	R_{abcd} = \frac{\R}{2} (g_{ac} g_{bd} - g_{ad} g_{bc})
	= \frac{\R}{4} (g \odot g)_{abcd} ,
\end{equation}
where $\R = R_{ab}{}^{ab}$ is the corresponding Ricci scalar.
The above notation and index conventions are chosen to match those
of~\cite{martel-poisson} as closely as possible.

We are interested in warped product~\cite[Ch.7]{oneill}
metrics of the form~\cite{kis-n+m, ki-master, ki-stab}
\begin{equation}
	\bar{g} = g_{ab} dy^a dy^b + r^2 \Omega_{AB} d\theta^A d\theta^B ,
\end{equation}
where $r=r(y)$ and $g_{ab}$ is static in the (Schwarzschild) coordinates
$(y^a) = (t,r)$,
\begin{equation} \label{eq:gen-g}
	g_{ab} = -f dt_a dt_b + \frac{1}{f} dr_a dr_b ,
\end{equation}
with $f = f(r)$. In these coordinates, the timelike Killing vector has
the form $t^a = (\del_t)^a$. For convenience, we also introduce the
notation $t_a = g_{ab} t^b = -f dt_a$ and $r_a = dr_a$. They are
related as $t^a = -\eps^{ab} r_b$, where $\eps_{ab} = (dt \wedge
dr)_{ab}$. Then, of course, $r_a r^a = f$ and $t_a t^a = -f$.

As we will see shortly, under our assumptions, the Einstein equations
with a cosmological constant $\Lambda$,
\begin{equation} \label{eq:einstein-eq}
	\bar{R}_{ab} - \frac{1}{2} \bar{\R} \bar{g}_{ab} + \Lambda
\bar{g}_{ab} = 0 ,
\end{equation}
are solved by~\cite[Eq.(2.15)]{ki-master}
\begin{equation} \label{eq:gen-f}
	f(r) = \alpha - \frac{2M}{r^{n-3}} - \frac{2\Lambda}{(n-1)(n-2)} r^2 ,
\end{equation}
where $M$ is a constant. When $\alpha=1$, $\Lambda = 0$ and $n\ge 4$,
this metric describes the higher dimensional spherically symmetric
static black holes, the so-called \emph{Schwarzschild-Tangherlini}
solutions, specializing to the Schwarzschild solution when $n=4$. When
$n=3$, we are forced to have $\alpha=0$ and the spacetime is actually of
constant curvature. With $n=3$ and $\Lambda < 0$, we get the BTZ
metric~\cite{btz}. In terms of the parameter $M$, the black hole mass is
given by
\begin{equation}
	\frac{(n-2) A_{n-2}}{8\pi} \frac{M}{G}
	= \frac{(n-2)}{2} \frac{A_{n-2}}{A_2} \frac{M}{G} ,
\end{equation}
where $G$ is the $n$-dimensional Newton's constant and
\begin{equation}
	A_{n-2} = \frac{2 \pi^{(n-1)/2}}{\Gamma[(n-1)/2]}
\end{equation}
is the area of the unit $(n-2)$-sphere. When $\alpha=0$, we get the
higher dimensional version of \emph{Taub's plane-symmetric spacetime}
\cite{taub}, \cite[Eq.(15.29)]{stephani-sols}, \cite[Eq.(2.2)]{bcpds}.
When $\alpha=-1$, we get the higher dimensional version of a
\emph{pseudo-Schwarzschild wormhole spacetime}~\cite{lobo-mimoso}.

In what follows, we restrict our attention to $n\ge 4$, which is
physically reasonable, but is also forced upon us by some of our
formulas, which have poles at $n=1,2$ or $3$.

For convenience, let us introduce the notations
\begin{align}
\label{eq:gen-f1}
	f_1(r) = r f'(r)
	&= (n-3) \frac{2M}{r^{n-3}} - \frac{4\Lambda}{(n-1)(n-2)} r^2 , \\
\label{eq:gen-f2}
	f_2(r) = r f_1'(r)
	&= -(n-3)^2 \frac{2M}{r^{n-3}} - \frac{8\Lambda}{(n-1)(n-2)} r^2 ,
\end{align}
as well as note that the formula~\eqref{eq:gen-f} for $f$ parametrized
by the constants $M$ and $\Lambda$, with $\alpha$ fixed, gives the
general solution to the differential equation
\begin{equation} \label{eq:f-ode}
	f_2 + (n-5) f_1 - 2(n-3) (f-\alpha) = 0 .
\end{equation}

Any tensor on $\M$ decomposes as $T_a = T_t dt_a + T_r dr_a$ or $T_a \to
(T_t,T_r)$, with obvious extension to higher rank tensors. With respect
to this decomposition and the coordinates $(t,r)$, the Levi-Civita
connection on $(\M,g)$ is then~\cite[Eq.(2.18)]{ki-master}
\begin{equation} \label{eq:schw-cov-coord}
	\nabla_a v_b
	\to \begin{bmatrix}
		\del_t v_t & \del_t v_r \\
		\del_r v_t & \del_r v_r
	\end{bmatrix}
	+ \begin{bmatrix}
		0 & -\frac{f_1}{2r f} \\
		-\frac{f_1}{2r f} & 0
	\end{bmatrix} v_t
	+ \begin{bmatrix}
		-f\frac{f_1}{2r} & 0 \\
		0 & \frac{1}{f} \frac{f_1}{2r}
	\end{bmatrix} v_r .
\end{equation}
Equivalently, we can summarize this information by giving the covariant
derivatives of the frame $(t_a,r_a)$,
\begin{equation} \label{eq:schw-cov-frame}
	\nabla_a t_b = \frac{f_1}{2r} \eps_{ab}
	\quad \text{and} \quad
	\nabla_a r_b = \frac{f_1}{2r} g_{ab} .
\end{equation}
A direct calculation gives the Ricci scalar on $(\M,g)$ as
\begin{equation}
	\R = \frac{f_1 - f_2}{r^2} .
\end{equation}
And finally, symmetrizing the covariant derivative as written
in~\eqref{eq:schw-cov-coord} or~\eqref{eq:schw-cov-frame}, the explicit
form of the Killing operator on $(\M,g)$ is
\begin{align}
	2\nabla_{(a} v_{b)}
\notag
	&= -2t_{(a} \nabla_{b)} \frac{v_t}{f}
			- 2 t_{a} t_{b} \frac{f_1}{2 r f} v_r
			- 2 t_{(a} r_{b)} \frac{1}{f} \del_t v_r
		+ r_a r_b \frac{1}{f} (2f \del_r v_r + \frac{f_1}{r} v_r)
	\\
	&\to \begin{bmatrix}
		f (2 \del_t \frac{1}{f} v_t - \frac{f_1}{r} v_r) & \del_t v_r + f\del_r \frac{1}{f} v_t \\
		\del_t v_r + f \del_r \frac{1}{f} v_t & \frac{1}{f} (2f\del_r v_r + \frac{f_1}{r} v_r)
	\end{bmatrix} .
\end{align}

Greek indices $\mu,\nu,\ldots$ on $\bar{\M}$-tensors are raised and
lowered by $\bar{g}_{\mu\nu}$. Lower case Latin indices $a,b,c,\ldots$
on $\M$-tensors are raised and lowered by $g_{ab}$. And upper case Latin
indices $A,B,C,\ldots$ on $\S$-tensors are raised and lowered by
$\Omega_{AB}$. Any $\bar{\M}$-tensor decomposes into sectors,%
	\footnote{While our tensor sector formalism with (pseudo-)spherical
	symmetry in $n$-dimensions is strongly inspired
	by~\cite{martel-poisson}, where it was presented for $n=4$, our
	conventions differ by the introduction of $r$-weights in the spherical
	sectors.} %
according to $T_\mu = T_a (dy^a)_\mu + r T_A (\d\theta^A)_\mu \to (T_a,
r T_A)$ and $T^\mu = T^a (\del_a)^\mu + \frac{1}{r} T^A (\del_A)^\mu \to
(T^a, \frac{1}{r} T^A)$, with obvious extension to higher rank tensors.
With a slight departure from this convention, let us define some
$\bar{\M}$-tensors by their sector decomposition
\begin{equation}
	g_{\mu\nu} \to \begin{bmatrix} g_{ab} & 0 \\ 0 & 0 \end{bmatrix} , ~
	g^{\mu\nu} \to \begin{bmatrix} g^{ab} & 0 \\ 0 & 0 \end{bmatrix} ,
	\quad \text{and} \quad
	\Omega_{\mu\nu} \to \begin{bmatrix} 0 & 0 \\ 0 & \Omega_{AB} \end{bmatrix} , ~
	\Omega^{\mu\nu} \to \begin{bmatrix} 0 & 0 \\ 0 & \Omega^{AB} \end{bmatrix} ,
\end{equation}
so that $\bar{g}_{\mu\nu} = g_{\mu\nu} + r^2 \Omega_{\mu\nu}$ and
$\bar{g}^{\mu\nu} = g^{\mu\nu} + r^{-2} \Omega^{\mu\nu}$.

The pair $(\nabla_a, D_A)$ defines a connection on $\bar{\M} = \M \times
\S$, which differs~\cite[App.A]{kis-n+m} from the Levi-Civita connection
$\bar{\nabla}_\mu$ as follows
\begin{equation} \label{eq:schw-cov}
	\bar{\nabla}_\mu v_\nu \to
	\begin{bmatrix}
		\nabla_a v_b & r\nabla_a v_B \\
		r D_A \frac{v_b}{r} & r^2 D_A \frac{v_B}{r}
	\end{bmatrix}
	+ \begin{bmatrix}
		0 & 0 \\
		0 & r^2 \Omega_{AB} \frac{r^c}{r} 
	\end{bmatrix} v_c
	+ \begin{bmatrix}
		0 & 0 \\
		-r_b \delta_A^C & 0
	\end{bmatrix} v_C .
\end{equation}

\begin{rem}
In giving the formula for $\bar{\nabla}_\mu$, we have essentially
extended the action of $\nabla_a$ and $D_A$ as linear differential
operators to tensors defined on $\bar{\M} = \M \times
\S$. The extension is done in exact analogy with the procedure described
at the start of Section~\ref{sec:flrw}. Recall that the covariant
derivatives $\nabla_a$ and $D_A$ simply act as coordinate derivatives on
scalars on $\M$ and $\S$. Suitably extending these coordinate
derivatives to $\bar{\M}$, the same can be said for
$(\M,\S)$-mixed tensors on $\bar{\M}$. But for the sake
of uniformity in notation, we continue to use the notation $\nabla_a$
and $D_A$ even when they act on $\M$- and $\S$-scalars
respectively.
\end{rem}

\bigskip

Next, we need to carefully study the structure of the curvature tensor.
The spacetime Riemann curvature tensor on $(\bar{\M},\bar{g})$
is~\cite[App.A]{kis-n+m}
\begin{align}
	\bar{R}_{\mu\nu\la\ka}
\notag
	&= \frac{\R}{2} (g\odot g)_{\mu\nu\la\ka}
	+ \frac{1}{2r^2} (\alpha - r_a r^a)
		(r^2\Omega \odot r^2\Omega)_{\mu\nu\la\ka}
	- \left(\frac{\nabla\nabla r}{r} \odot r^2\Omega\right)_{\mu\nu\la\ka} \\
	&= \frac{f_1-f_2}{4r^2} (g\odot g)_{\mu\nu\la\ka}
	+ \frac{(\alpha - f)}{2r^2} (r^2\Omega \odot r^2\Omega)_{\mu\nu\la\ka}
	- \frac{f_1}{2r^2} (g \odot r^2\Omega)_{\mu\nu\la\ka} ,
\label{eq:schw-curv-riem}
\end{align}
with the corresponding Ricci tensor
\begin{align}
	\bar{R}_{\mu\la}
\notag
	&= \frac{f_1 - f_2}{4r^2} 2 g_{\mu\la}
	+ \frac{(\alpha-f)}{2r^2} 2 (n-3) r^2 \Omega_{\mu\la}
	- \frac{f_1}{2r^2} (2 r^2 \Omega + (n-2) g)_{\mu\la} \\
	&= -\frac{f_2 + (n-3)f_1}{2r^2} g_{\mu\la}
	- \frac{f_1 + (n-3) (f-\alpha)}{r^2} r^2 \Omega_{\mu\la} .
\label{eq:schw-curv-ricc}
\end{align}
To satisfy Einstein's equations in the presence of a cosmological
constant~\eqref{eq:einstein-eq}, we must have
\begin{equation}
	\bar{R}_{\mu\la}
	= \frac{2 \Lambda}{(n-2)} \bar{g}_{\mu\la}
	= \frac{2 \Lambda}{(n-2)} (g + r^2 \Omega)_{\mu\la} ,
\end{equation}
which implies
\begin{align}
	f_1 &= -(n-3) (f-\alpha) - \frac{2 \Lambda}{(n-2)} r^2 , \\
	f_2 &= -(n-3) f_1 - \frac{4 \Lambda}{(n-2)} r^2 .
\end{align}
Eliminating the explicit dependence on $\Lambda$, we obtain precisely
the second order ODE~\eqref{eq:f-ode} whose general solution is given by
$f(r)$ in~\eqref{eq:gen-f}.

Recalling the definition of the Kulkarni-Nomizu and contraction
products~\eqref{eq:kn-prod} and~\eqref{eq:dot-con}, we get the following
useful identities, where we used $\bar{g}_{\mu\nu}$ for contractions:
\begin{align}
	(g\odot g) \cdot (g\odot g)
	&= 4 (g\odot g) , \\
	(r^2\Omega \odot r^2\Omega) \cdot (r^2\Omega \odot r^2\Omega)
	&= 4 (r^2\Omega \odot r^2\Omega) , \\
	(g\odot r^2\Omega) \cdot (g\odot r^2\Omega)
	&= 2 (g\odot r^2\Omega) , \\
	(g\odot r^2\Omega) \cdot (g\odot g)
	&= 0 , \\
	(g\odot r^2\Omega) \cdot (r^2\Omega\odot r^2\Omega)
	&= 0 , \\
	(g\odot g) \cdot (r^2\Omega\odot r^2\Omega)
	&= 0 , \\
	(g\odot g)_{\la\nu}{}^\nu{}_\ka
	&= -2 g_{\la\ka} , \\
	(r^2\Omega \odot r^2\Omega)_{\la\nu}{}^\nu{}_\ka
	&= -2(n-3) (r^2\Omega)_{\la\ka} , \\
	(g\odot r^2\Omega)_{\la\nu}{}^\nu{}_\ka
	&= -2 (r^2\Omega)_{\la\ka} - (n-2) g_{\la\ka} , \\
	(g\odot g)_{\mu\la\nu\ka} r^\la r^\ka
	&= 2(r^\la r_\la) g_{\mu\nu} - 2r_\mu r_\nu , \\
	(g\odot r^2\Omega)_{\mu\la\nu\ka} r^\la r^\ka
	&= (r^\la r_\la) (r^2\Omega)_{\mu\nu} , \\
	(r^2\Omega \odot r^2\Omega)_{\mu\la\nu\ka} r^\la r^\ka
	&= 0 .
\end{align}
Defining
\begin{align}
	\bar{T}_{\mu\nu\la\ka}
\notag
	&= \bar{R}_{\mu\nu\la\ka}
		- \frac{\Lambda}{(n-1)(n-2)} (\bar{g}\odot \bar{g})_{\mu\nu\la\ka} \\
	&= \frac{M}{r^{n-1}} \left[
		\frac{(n-2)(n-3)}{2} (g\odot g)_{\mu\nu\la\ka}
	+ (r^2\Omega \odot r^2\Omega)_{\mu\nu\la\ka}
	- (n-3) (g\odot r^2\Omega)_{\mu\nu\la\ka} \right] ,
\label{eq:schw-curv-lambda}
\end{align}
we also get the identities
\begin{align*}
	\bar{T}\cdot \bar{T}
	&= \left(\frac{M}{r^{n-1}}\right)^2 \left[
		(n-2)^2(n-3)^2 (g\odot g)
		+ 4 (r^2\Omega\odot r^2\Omega)
		+ 2 (n-3)^2 (g\odot r^2\Omega) \right] \\
	\bar{T}\cdot \bar{T} \cdot \bar{T}\cdot \bar{T}
	&= \left(\frac{M}{r^{n-1}}\right)^4 \left[
		4 (n-2)^4(n-3)^4 (g\odot g)
		+ 64 (r^2\Omega\odot r^2\Omega)
		+ 8 (n-3)^4 (g\odot r^2\Omega) \right] \\
	(\bar{T}\cdot \bar{T})_{\mu\nu}{}^{\nu}{}_\ka
	&= -\left(\frac{M}{r^{n-1}}\right)^2 \left[
		2(n-2)^2(n-3)^2 g
		+ 8(n-3) r^2\Omega
		+ 2(n-3)^2 (2 r^2 \Omega + (n-2) g) \right]_{\mu\ka} \\
	&= -\left(\frac{M}{r^{n-1}}\right)^2 \left[
		2(n-1)(n-2)(n-3)^2 g
		+ 4(n-1)(n-3) r^2\Omega \right]_{\mu\ka} \\
	(\bar{T}\cdot \bar{T})_{\mu\nu}{}^{\mu\nu}
	&= \left(\frac{M}{r^{n-1}}\right)^2 \left[
		4(n-1)(n-2)(n-3)^2
		+ 4(n-1)(n-2)(n-3) \right] \\
	&= 4(n-1)(n-2)^2(n-3) \left(\frac{M}{r^{n-1}}\right)^2 \\
	\frac{\bar{\nabla}_\la (\bar{T}\cdot \bar{T})_{\mu\nu}{}^{\mu\nu}}
		{(\bar{T}\cdot \bar{T})_{\mu\nu}{}^{\mu\nu}}
	&= -2(n-1) \frac{r_\la}{r}
\end{align*}

We would like to use these identities to write $(r^2\Omega)_{\la\ka}$
and a simple $r$-dependent scalar as covariant expression in the
curvature. For the latter, the simplest choice seems to be
\begin{equation}
	\bar{T}^{(1)}[\bar{g}]
	:= \frac{(\bar{T} \cdot \bar{T})_{\mu\nu}{}^{\mu\nu}}{4(n-1)(n-2)^2(n-3)}
	= \left(\frac{M}{r^{n-1}}\right)^2 .
\end{equation}
Next, we encounter a slight dimension dependence in the expression for
$(r^2\Omega)_{\la\ka}$. When $n>4$, we can use
%
%
\begin{equation}
	(r^2 \Omega)_{\la\ka}
	=
	\frac{2(n-2)^2}{(n-1)(n-4)} 
	\frac{(\bar{T}\cdot \bar{T})_{\la\nu}{}^{\nu}{}_\ka}{(\bar{T}\cdot\bar{T})_{\mu\nu}{}^{\mu\nu}}
	+ \frac{(n-2)(n-3)}{(n-1)(n-4)} \bar{g}_{\la\ka}
	=: \bar{T}^{(2)}_{\la\ka}[\bar{g}] ,
\end{equation}
while for $n=4$ the simplest expression we could find is
%
%
%
\begin{multline}
	(r^2 \Omega)_{\la\ka}
	=
	-\frac{2(n-1)(n-2)^4}{(n-3)^3 [\bar{\nabla}(\bar{T}\cdot\bar{T})_{\mu\nu}{}^{\mu\nu}]^2}
	\\
	\left(\bar{T}\cdot \bar{T} \cdot \bar{T}\cdot \bar{T}
		- \frac{(n-3)}{(n-1)} (\bar{T}\cdot\bar{T})_{\mu\nu}{}^{\mu\nu} \,
			\bar{T}\cdot \bar{T}
	\right)_{\la\rho\ka\sigma}
	\frac{\bar{\nabla}^\rho (\bar{T}\cdot\bar{T})_{\mu\nu}{}^{\mu\nu}}{(\bar{T}\cdot\bar{T})_{\mu\nu}{}^{\mu\nu}}
	\frac{\bar{\nabla}^\sigma (\bar{T}\cdot\bar{T})_{\mu\nu}{}^{\mu\nu}}{(\bar{T}\cdot\bar{T})_{\mu\nu}{}^{\mu\nu}}
	\\
	=: \bar{T}^{(3)}_{\la\ka}[\bar{g}] .
\end{multline}
Although, since it also works for $n>4$, if desired, the more
complicated expression $\bar{T}^{(3)}[\bar{g}]$ could actually be used
in higher dimensions too.

To make use of identity~\eqref{eq:lie-killing}, we compute the Lie
derivatives
\begin{align}
	\Lie_v \left(\frac{M}{r^{n-1}}\right)^2
	&= -2(n-1) \frac{r^c v_c}{r} \left(\frac{M}{r^{n-1}}\right)^2 , \\
	\Lie_v (r^2\Omega)_{\mu\nu}
	&\to \begin{bmatrix}
			0 & r (r\nabla_a \frac{v_B}{r}) \\
			r (r\nabla_b \frac{v_A}{r}) &
				2r^2 (D_{(A} \frac{1}{r} v_{B)} + \Omega_{AB} \frac{r^c v_c}{r})
		\end{bmatrix} .
\end{align}
Hence, defining the linear operators
\begin{align}
	\label{eq:schw-J1}
	J_1[h] &:= -\frac{1}{2(n-1)} \frac{r}{f} \left(\frac{M}{r^{n-1}}\right)^{-2}
		\dot{\bar{T}}^{(1)}[h] , \\
	\label{eq:schw-J2}
	J_2[h]_{aB} &:= \frac{1}{r^2}
		\begin{cases}
			\dot{\bar{T}}^{(2)}_{aB}[h] & (n > 4) \\
			\dot{\bar{T}}^{(3)}_{aB}[h] & (n = 4)
		\end{cases} ,
\end{align}
the Lie derivative formula~\eqref{eq:lie-killing} implies the
compositional identities
\begin{align}
	J_1 \circ \bar{K}[v] &= \frac{r^c v_c}{f} , \\
	J_2 \circ \bar{K}[v]_{aB} &= \nabla_a \frac{v_B}{r} .
\end{align}

Based on Equation~\eqref{eq:schw-cov}, the explicit expression for the
Killing operator is
\begin{equation}
	\bar{K}_{\mu\nu}[v] = 2 \bar{\nabla}_{(\mu} v_{\nu)}
	\to
	\begin{bmatrix}
		2\nabla_{(a} v_{b)} & r^2\nabla_a \frac{1}{r} v_B + r D_B \frac{v_a}{r} \\
		r^2\nabla_b \frac{1}{r} v_A + r D_A \frac{v_b}{r} & 2 r^2 D_{(A} \frac{1}{r} v_{B)}
			+ 2r^2 \frac{r^c v_c}{r} \Omega_{AB}
	\end{bmatrix} .
\end{equation}

For further convenience, we parametrize
\begin{equation} \label{eq:schw-vh-param}
	v_\mu \to \begin{bmatrix} u_t\, f dt_a + u_r dr_a \\ r (r X_A) \end{bmatrix}
	\quad \text{and} \quad
	h_{\mu\nu} \to \begin{bmatrix}
		p \, r_a r_b - 2t_{(a} w_{b)} & r (r Y_{aB}) \\
		r (r Y_{bA}) & r^2 Z_{AB}
		\end{bmatrix} .
\end{equation}
The Killing equation $h = \bar{K}[v]$ then becomes
\begin{equation} \label{eq:schw-killing}
	\begin{bmatrix}
		p \\ \cmidrule(lr){1-1}
		w \\
		Y \\
		Z
	\end{bmatrix}
	= \bar{K} \begin{bmatrix}
		u_r \\ \cmidrule(lr){1-1}
		u_t \\
		X
	\end{bmatrix}
	= \begin{bmatrix}[c|cc]
		\frac{1}{f} (2f \del_r + \frac{f_1}{r}) & 0 & 0 \\
		\cmidrule(lr){1-1} \cmidrule(lr){2-3}
		dr\,\frac{1}{f} \del_t - dt\,\frac{f_1}{2r} & \nabla & 0 \\
		dr\, \frac{f}{r^2} D \frac{1}{f} & dt\, \frac{f}{r^2} D & \nabla  \\
		2\Omega \frac{f}{r} & 0 & C_0
	\end{bmatrix} \begin{bmatrix}
		u_r \\ \cmidrule(lr){1-1}
		u_t \\
		X
	\end{bmatrix} ,
\end{equation}
where $C_0[X]_{AB} = D_A X_B + D_B X_A$ is the Killing operator on the
constant curvature factor $(\S,\Omega)$, and hence the first operator of
the Calabi complex $C_i$, $i\ge 0$, which constitutes a compatibility
complex for $C_0$ (Section~\ref{sec:cc}).

\begin{rem}
Note that we are continuing here to use the block matrix notation for
differential operators, as discussed previously in
Remark~\ref{rem:block-matrix}.
\end{rem}

With the above parametrizations for $v$ and $h$, the compositional
identities for the operators $J_1$~\eqref{eq:schw-J1} and
$J_2$~\eqref{eq:schw-J2} simplify to
\begin{subequations} \label{eq:schw-J1J2K}
\begin{align}
	J_1 \circ \bar{K}
	&= \begin{bmatrix}[c|ccc] J_1^p & J_1^w & J_1^Y & J_1^Z \end{bmatrix}
		\circ \bar{K}
	= \begin{bmatrix}[c|cc] \id & 0 & 0 \end{bmatrix} , \\
	J_2 \circ \bar{K}
	&= \begin{bmatrix}[c|ccc] J_2^p & J_2^w & J_2^Y & J_2^Z \end{bmatrix}
		\circ \bar{K}
	= \begin{bmatrix}[c|cc] 0 & 0 & \nabla \end{bmatrix} .
\end{align}
\end{subequations}

Now we have all the information that we need to use the methods of
Section~\ref{sec:compat} to construct a compatibility complex for the
Killing operator $\bar{K}$. We will follow roughly the same outline as
we did in the Section~\ref{sec:flrw} on cosmological FLRW geometries.

From now on, our strategy will be to show that our Killing operator
$\bar{K}_0 = \bar{K}$ is equivalent to each of the operators
\begin{equation}
	\tK_0
	\begin{bmatrix}
		u_t \\
		X
	\end{bmatrix}
	= \begin{bmatrix}
		\bar{d}_0 = \bar{\nabla} \to \begin{bmatrix} \nabla \\ D \end{bmatrix} &
			\begin{bmatrix} 0 \\ 0 \end{bmatrix} \\
		dt\,\frac{f}{r^2}\,D & d_0 = \nabla \\
		0 & C_0
	\end{bmatrix}
	\begin{bmatrix}
		u_t \\
		X
	\end{bmatrix}
	\quad \text{and} \quad
	K_0
	\begin{bmatrix}
		u_t \\
		X
	\end{bmatrix}
	= \begin{bmatrix}
		\bar{d}_0 & 0 \\
		0 & d_0 \\
		0 & C_0
	\end{bmatrix}
	\begin{bmatrix}
		u_t \\
		X
	\end{bmatrix} ,
\end{equation}
where we have introduced the notation $\bar{d}_i$ and $d_i$, $i\ge 0$,
for the usual exterior derivatives acting on $i$-forms on $\bar{\M}$ and
$\M$ respectively (hence the corresponding de~Rham complexes). In the
sequel, we will use the notations $\bar{d}_0$ and $\left[
\begin{smallmatrix} \nabla \\ D \end{smallmatrix} \right]$ completely
interchangeably. Then, we will lift the known compatibility complex for
$K_0$ first to $\tK_0$ and finally to $\bar{K}_0$. This known
compatibility complex has the form

\begin{subequations} \label{eq:schw-calabi}
\begin{align}
	K_0 &=
	\begin{bmatrix}
		\bar{d}_0 & 0 \\
		0 & d_0 \\
		0 & C_0
	\end{bmatrix} , \\
	K_1 &=
	\begin{bmatrix}
		\bar{d}_1 & 0 & 0 \\
		0 & d_1 & 0 \\
		0 & -C_0 & d_0 \\
		0 & 0 & C_1
	\end{bmatrix} , \\
	K_2 &=
	\begin{bmatrix}
		\bar{d}_2 & 0 & 0 & 0 \\
		0 & C_0 & d_1 & 0 \\
		0 & 0 & -C_1 & d_0 \\
		0 & 0 & 0 & C_2
	\end{bmatrix} , \\
	K_i &=
	\begin{bmatrix}
		\bar{d}_i & 0 & 0 & 0 \\
		0 & C_{i-2} & d_1 & 0 \\
		0 & 0 & -C_{i-1} & d_0 \\
		0 & 0 & 0 & C_i
	\end{bmatrix} 
		\quad (2 < i < n-2), \\
	K_{n-2} &=
	\begin{bmatrix}
		\bar{d}_{n-2} & 0 & 0 & 0 \\
		0 & C_{n-4} & d_1 & 0 \\
		0 & 0 & -C_{n-3} & d_0
	\end{bmatrix} , \\
	K_{n-1} &=
	\begin{bmatrix}
		\bar{d}_{n-1} & 0 & 0 \\
		0 & C_{n-3} & d_1
	\end{bmatrix} , \\
	K_i &= 0
		\quad (n \le i) .
\end{align}
\end{subequations}
It is straightforward to construct an equivalence between this complex
and a twisted de~Rham complex, similar to how it was done
in~\eqref{eq:flrw-calabi}, thus showing that each of the above
compatibility operators is complete.

\bigskip

We start with the explicit reduction of $\bar{K}_0$ to $\tK_0$ and then
to $K_0$. Here and in each subsequent step, we give pairs of diagrams,
which could be concatenated vertically, illustrating the passage from
the $\bar{K}_i$ to the $\tK_i$ and to the $K_i$ sequences. All the
diagrams below illustrate equivalences up to homotopy, as discussed in
Section~\ref{sec:compat}. All the required identities can be checked by
direct calculation, making careful use of the known identities $d_{i+1}
\circ d_i = 0$, $\bar{d}_{i+1} \circ \bar{d}_i = 0$, $C_{i+1} \circ C_i
= 0$, as well as the compositional identities~\eqref{eq:schw-J1J2K}.

\begin{subequations} \label{eq:schw-lift1}

\begin{equation} \label{eq:schw-lift1a}
\begin{tikzcd}[column sep=7.5cm,row sep=5cm]
	\bullet
		\ar{r}{\bar{K}_0 = \bar{K}}
		\ar[swap,shift right]{d}{\begin{bmatrix}[c|cc]0 & \id & 0 \\ 0 & 0 & \id\end{bmatrix}} \&
	\bullet
		\ar[dashed,bend left=15]{l}{\begin{bmatrix}J_1 \\ \cmidrule(lr){1-1} 0 \\ 0\end{bmatrix}}
		\ar[swap,shift right,pos=0.6]{d}{\begin{bmatrix}[c|ccc]
				0 & \id & 0 & 0 \\
				0 & 0 & \id & 0 \\
				0 & 0 & 0 & \id
			\end{bmatrix}
			\left( \id - \bar{K} \begin{bmatrix} J_1 \\ \cmidrule(lr){1-1} 0 \\ 0 \end{bmatrix} \right)}
	\\
	\bullet
		\ar[swap]{r}{\tK_0 = \begin{bmatrix}
			\nabla & 0 \\
			dt\,\frac{f}{r^2} D & \nabla \\
			0 & C_0
		\end{bmatrix}}
		\ar[swap,shift right]{u}{\begin{bmatrix}0 & 0 \\ \cmidrule(lr){1-2} \id & 0 \\ 0 & \id\end{bmatrix}} \&
	\bullet
		\ar[swap,dashed,bend right=15]{l}{0}
		\ar[swap,shift right]{u}{
			\begin{bmatrix}
				0 & 0 & 0 \\
				\cmidrule(lr){1-3}
				\id & 0 & 0 \\
				0 & \id & 0 \\
				0 & 0 & \id
			\end{bmatrix}}
\end{tikzcd}
\end{equation}

\begin{equation} \label{eq:schw-lift1b}
\begin{tikzcd}[column sep=7.5cm,row sep=5cm]
	\bullet
		\ar{r}{\tK_0 = \begin{bmatrix}
			\nabla & 0 \\
			dt\,\frac{f}{r^2} D & \nabla \\
			0 & C_0
		\end{bmatrix}}
		\ar[swap,shift right]{d}{\begin{bmatrix}\id & 0 \\ 0 & \id\end{bmatrix}} \&
	\bullet
		\ar[dashed,bend left=20]{l}{0}
		\ar[swap,shift right]{d}{\begin{bmatrix}
			\begin{bmatrix} \id \\ r^2 dt\cdot J_2^w \end{bmatrix} &
				\begin{bmatrix} 0 \\ r^2 dt\cdot (J_2^Y-\id) \end{bmatrix} &
				\begin{bmatrix} 0 \\ r^2 dt\cdot J_2^Z \end{bmatrix} \\
			J_2^w & J_2^Y & J_2^Z \\
			0 & 0 & \id
		\end{bmatrix}}
	\\
	\bullet
		\ar[swap]{r}{K_0 = \begin{bmatrix}
			\bar{d}_0 & 0 \\
			0 & d_0 \\
			0 & C_0
		\end{bmatrix}}
		\ar[swap,shift right]{u}{\begin{bmatrix}\id & 0 \\ 0 & \id\end{bmatrix}} \&
	\bullet
		\ar[swap,dashed,bend right=20]{l}{0}
		\ar[swap,shift right]{u}{
			\begin{bmatrix}[@{~[}cc@{]~}cc]
				\id & 0 & 0 & 0 \\
				0 & dt\,\frac{f}{r^2} & \id & 0 \\
				0 & 0 & 0 & \id
			\end{bmatrix}}
\end{tikzcd}
\end{equation}

\end{subequations}

Next, as we did previously in Section~\ref{sec:flrw}, we iterate the
construction from Lemma~\ref{lem:lift-compat}, while applying the
simplifications discussed after the proof of
Theorem~\ref{thm:fintype-compat}. As before, the $\bar{K}_i$ and $\tK_i$
complexes are built up by appending the following diagrams to the right
of the diagrams in~\eqref{eq:schw-lift1}. Also as before, we do not
repeat the labels on the vertical arrows if they can be read off a
preceding diagram.

The resulting operators $\bar{K}_1$ and $\tK_1$ will be, respectively,
compatibility operators for $\bar{K}_0$ and $\tK_0$. Some of the
auxiliary arrows in these diagrams use the operators $\tH_{J_1}$ and
$H_{J_1}$, which are defined as follows. Noting that
\begin{align*}
	J_1 \bar{K} \begin{bmatrix} 0 & 0 \\ \cmidrule(lr){1-2} \id & 0 \\ 0 & \id \end{bmatrix}
	&= J_1 \begin{bmatrix}
		0 & 0 \\
		\cmidrule(lr){1-2}
		\nabla & 0 \\
		dt\, \frac{f}{r^2} D & \nabla \\
		0 & C_0
		\end{bmatrix} \\
	&= \begin{bmatrix} J_1^w & J_1^Y & J_1^Z \end{bmatrix}
		\begin{bmatrix}
		\nabla & 0 \\
		dt\, \frac{f}{r^2} D & \nabla \\
		0 & C_0
		\end{bmatrix}
	= \begin{bmatrix}[c|cc] \id & 0 & 0 \end{bmatrix}
		\begin{bmatrix} 0 & 0 \\ \cmidrule(lr){1-2} \id & 0 \\ 0 & \id \end{bmatrix}
	= 0 ,
\end{align*}
we must be able to factor
\begin{equation} \label{eq:schw-HJ1-tHJ1}
	\begin{bmatrix}
		J_1^w & J_1^Y & J_1^Z
	\end{bmatrix}
	=
	\tH_{J_1} \tK_1 ,
	\quad \text{and} \quad
	\begin{bmatrix}[c|ccc]
		J_1^p & J_1^w & J_1^Y & J_1^Z
	\end{bmatrix}
	=
	H_{J_1} \begin{bmatrix}[c|c] \id & 0 \\ \cmidrule(lr){1-1} \cmidrule(lr){2-2} 0 & \tK_1 \end{bmatrix}
\end{equation}
through some operators $\tH_{J_1}$ and $H_{J_1}$. A bit more precisely,
$H_{J_1} = \begin{bmatrix}[c|c] J_1^p & \tH_{J_1} \end{bmatrix}$.

\begin{subequations} \label{eq:schw-lift2}

\begin{equation} \label{eq:schw-lift2a}
\begin{tikzcd}[column sep=7cm,row sep=5cm]
	\bullet
		\ar{r}{\bar{K}_1 = \begin{bmatrix}[c|c]
				\id & 0 \\ \cmidrule(lr){1-1} \cmidrule(lr){2-2}
				0 & \tK_1
			\end{bmatrix}
			\left(\id - \bar{K}\begin{bmatrix}J_1 \\ \cmidrule(lr){1-1} 0 \\ 0 \end{bmatrix}\right)}
		\ar[swap,shift right]{d}
		\&
	\bullet
		\ar[swap,shift right]{d}{\begin{bmatrix}[c|c] 0 & \id \end{bmatrix}}
		\ar[dashed,bend left=15]{l}{
			\begin{bmatrix}[c|c]
				\id & 0 \\
				\cmidrule(lr){1-1} \cmidrule(lr){2-2}
				0 & 0 \\
			\end{bmatrix}}
	\\
	\bullet
		\ar[swap]{r}{\tK_1}
		\ar[swap,shift right]{u} \&
	\bullet
		\ar[swap,shift right]{u}{
			\left(\id - \begin{bmatrix}[c|c] \id & 0 \\ \cmidrule(lr){1-1} \cmidrule(lr){2-2} 0 & \tK_1 \end{bmatrix}
				\bar{K} \begin{bmatrix} H_{J_1} \\ \cmidrule(lr){1-1} 0 \\ 0 \end{bmatrix}
			\right)
			\begin{bmatrix}[c|cccc]
				0 & 0 & 0 & 0 & 0 \\ \cmidrule(lr){1-1} \cmidrule(lr){2-5}
				\id & 0 & 0 & 0 & 0 \\ \cmidrule(lr){1-1} \cmidrule(lr){2-5}
				0 & \id & 0 & 0 & 0 \\
				0 & 0 & \id & 0 & 0 \\
				0 & 0 & 0 & \id & 0 \\
				0 & 0 & 0 & 0 & \id
			\end{bmatrix}
		}
		\ar[swap,dashed,bend right=15]{l}{
			\begin{bmatrix}[c|ccc]
				0 & \id & 0 & 0 \\
				0 & 0 & \id & 0 \\
				0 & 0 & 0 & \id
			\end{bmatrix} \bar{K} \begin{bmatrix} \tH_{J_1} \\ \cmidrule(lr){1-1} 0 \\ 0 \end{bmatrix}}
\end{tikzcd}
\end{equation}

\begin{equation} \label{eq:schw-lift2b}
\begin{tikzcd}[column sep=7cm,row sep=5cm]
	\bullet
		\ar{r}{\tK_1 = \begin{bmatrix}
			\begin{bmatrix} 0 & -dr\cdot(-) \end{bmatrix} & dr\cdot(-) & 0 \\
			\cmidrule(lr){1-3}
			\bar{d}_1 \begin{bmatrix} \id & 0 \\ 0 & -r^2 dt\cdot(-) \end{bmatrix} &
				\bar{d}_1 \begin{bmatrix} 0 \\ r^2 dt\cdot(-) \end{bmatrix} & 0 \\
			0 & d_1 & 0 \\
			0 & -C_0 & d_0 \\
			0 & 0 & C_1
		\end{bmatrix}
		\begin{bmatrix}
			\begin{bmatrix}\id \\ 0\end{bmatrix} &
				\begin{bmatrix}0 \\ \id\end{bmatrix} &
				\begin{bmatrix}0 \\ 0\end{bmatrix} \\
			J_2^w & J_2^Y & J_2^Z \\
			0 & 0 & \id
		\end{bmatrix}}
		\ar[swap,shift right]{d}
		\&
	\bullet
		\ar[swap,shift right]{d}{\begin{bmatrix}[c|cccc]
				0 & \id & 0 & 0 & 0 \\
				0 & 0 & \id & 0 & 0 \\
				0 & 0 & 0 & \id & 0 \\
				0 & 0 & 0 & 0 & \id
			\end{bmatrix}}
		\ar[dashed,bend left=15,pos=.577]{l}{\begin{bmatrix}
			0 \\ -\frac{1}{f}\,dr \\ 0
		\end{bmatrix} \begin{bmatrix}[c|cccc] \id & 0 & 0 & 0 & 0 \end{bmatrix}}
	\\
	\bullet
		\ar[swap]{r}{K_1 = \begin{bmatrix}
				\bar{d}_1 & 0 & 0 \\
				0 & d_1 & 0 \\
				0 & -C_0 & d_0 \\
				0 & 0 & C_1
			\end{bmatrix}}
		\ar[swap,shift right]{u} \&
	\bullet
		\ar[swap,shift right]{u}{
			\begin{bmatrix}
				0 & 0 & 0 & 0 \\
				\cmidrule(lr){1-4}
				\id & 0 & 0 & 0 \\
				0 & \id & 0 & 0 \\
				0 & 0 & \id & 0 \\
				0 & 0 & 0 & \id
			\end{bmatrix}
			+ \begin{bmatrix}
				dr\cdot(-) \\
				\cmidrule(lr){1-1}
				\bar{d}_1 \begin{bmatrix} 0 \\ r^2 dt\cdot(-) \end{bmatrix} \\
				d_1 \\
				-C_0 \\
				0
			\end{bmatrix}
			\tH_{J_2}
		}
		\ar[swap,dashed,bend right=15,pos=.6]{l}{
			\begin{bmatrix}
				\begin{bmatrix} 0 \\ -r^2 dt\cdot(-) \end{bmatrix} \\
				-\id \\
				0
			\end{bmatrix} \tH_{J_2}}
\end{tikzcd}
\end{equation}

\end{subequations}

Above, we have used the notations $dt\cdot(-) = dt^a (-)_a$ and
$dr\cdot(-) = dr^a (-)_a$. Also, the operator $\tH_{J_2}$ is defined as
follows. Noting that
\begin{align*}
	J_2 \bar{K} \begin{bmatrix} 0 \\ \cmidrule(lr){1-1} \id \\ 0 \end{bmatrix}
	&= J_2 \begin{bmatrix} 0 \\ \cmidrule(lr){1-1} \nabla \\ dt\,\frac{f}{r^2} D \\ 0\end{bmatrix}
	= \begin{bmatrix} J_2^w & J_2^Y\,dt\,\frac{f}{r^2} \end{bmatrix}
		\left(\bar{d}_0 = \begin{bmatrix} \nabla \\ D \end{bmatrix}\right)
	= \begin{bmatrix}[c|cc] 0 & 0 & \nabla \end{bmatrix}
		\begin{bmatrix} 0 \\ \cmidrule(lr){1-1} \id \\ 0 \end{bmatrix}
	= 0 , \\
	J_2 \bar{K} \begin{bmatrix} 0 \\ \cmidrule(lr){1-1} 0 \\ \id \end{bmatrix}
	&= J_2 \begin{bmatrix} 0 \\ \cmidrule(lr){1-1} 0 \\ \nabla \\ C_0 \end{bmatrix}
	= \begin{bmatrix} J_2^Y & J_2^Z \end{bmatrix}
		\begin{bmatrix} d_0 = \nabla \\ C_0 \end{bmatrix}
	= \begin{bmatrix}[c|cc] 0 & 0 & \nabla \end{bmatrix}
		\begin{bmatrix} 0 \\ \cmidrule(lr){1-1} 0 \\ \id \end{bmatrix}
	= \nabla
	= \begin{bmatrix} \id & 0 \end{bmatrix}
		\begin{bmatrix} \nabla \\ C_0 \end{bmatrix} ,
\end{align*}
we must be able to factor
\begin{equation} \label{eq:schw-tHJ2}
	\begin{bmatrix}
		\begin{bmatrix} J_2^w & J_2^Y \, dt\, \frac{f}{r^2} \end{bmatrix} &
		J_2^Y - \id & J_2^Z
	\end{bmatrix}
	=
	\tH_{J_2}
	\begin{bmatrix}
		\bar{d}_1 & 0 & 0 \\
		0 & d_1 & 0 \\
		0 & -C_0 & d_0 \\
		0 & 0 & C_1
	\end{bmatrix}
\end{equation}
through some operator $\tH_{J_2}$.

For convenience, we note that
\begin{equation*}
	\begin{bmatrix}[c|cccc]
		\id & 0 & 0 & 0 & 0
	\end{bmatrix}
	\tK_1
	=
	\begin{bmatrix}
		dr\cdot J_2^w & dr\cdot(J_2^Y-\id) & dr\cdot J_2^Z
	\end{bmatrix} ,
\end{equation*}
while on the other hand
\begin{equation*} 
	\begin{bmatrix}[c|c]
		0 & \tH_{J_2}
	\end{bmatrix} \tK_1
	= J_2^Y \frac{dr}{f}
		\begin{bmatrix} dr\cdot J_2^w & dr\cdot(J_2^Y-\id) & dr\cdot J_2^Z \end{bmatrix}
	.
\end{equation*}
Then, defining
\begin{equation} \label{eq:schw-HJ2}
	H_{J_2} = \begin{bmatrix}[c|c] -J_2^Y \frac{dr}{f} & \tH_{J_2} \end{bmatrix} ,
	\quad \text{we have} \quad
	H_{J_2} \tK_1 = 0 .
\end{equation}

With the next iteration of Lemma~\ref{lem:lift-compat}, we construct the
compatibility operators $\bar{K}_2$ and $\tK_2$.

\begin{subequations} \label{eq:schw-lift3}

\begin{equation} \label{eq:schw-lift3a}
\begin{tikzcd}[column sep=7cm,row sep=5cm]
	\bullet
		\ar{r}{\bar{K}_2 =
		\begin{bmatrix}[c|c]
			\multicolumn{2}{c}{H_{J_1}} \\
			\cmidrule(lr){1-2}
			0 & \tK_2
		\end{bmatrix}}
		\ar[swap,shift right]{d}
		\&
	\bullet
		\ar[swap,shift right]{d}{\begin{bmatrix}[c|c]
				0 & \id
			\end{bmatrix}}
		\ar[dashed,bend left=15]{l}{
			\begin{bmatrix}[c|c]
				\id & 0 \\
				\cmidrule(lr){1-1} \cmidrule(lr){2-2}
				0 & \tK_1
			\end{bmatrix}
			\bar{K}
			\begin{bmatrix}[c|c]
				\id & 0 \\
				\cmidrule(lr){1-1} \cmidrule(lr){2-2}
				0 & 0 \\
			\end{bmatrix}
		}
	\\
	\bullet
		\ar[swap]{r}{\tK_2}
		\ar[swap,shift right]{u} \&
	\bullet
		\ar[swap,shift right]{u}{
			\begin{bmatrix}
				0 \\ \cmidrule(lr){1-1}
				\id
			\end{bmatrix}}
		\ar[swap,dashed,bend right=15]{l}{0}
\end{tikzcd}
\end{equation}

\begin{equation} \label{eq:schw-lift3b}
\begin{tikzcd}[column sep=7cm,row sep=5cm]
	\bullet
		\ar{r}{\tK_2 =
		\begin{bmatrix}[c|cccc]
			\multicolumn{5}{c}{H_{J_2}} \\
			\cmidrule(lr){1-5}
			0 & \bar{d}_2 & 0 & 0 & 0 \\
			0 & 0 & C_0 & d_1 & 0 \\
			0 & 0 & 0 & -C_1 & d_0 \\
			0 & 0 & 0 & 0 & C_2
		\end{bmatrix}}
		\ar[swap,shift right]{d}
		\&
	\bullet
		\ar[swap,shift right,pos=0.6]{d}{\begin{bmatrix}[c|cccc]
				0 & \id & 0 & 0 & 0 \\
				0 & 0 & \id & 0 & 0 \\
				0 & 0 & 0 & \id & 0 \\
				0 & 0 & 0 & 0 & \id
			\end{bmatrix}}
		\ar[dashed,bend left=15,pos=.577]{l}{
			-\begin{bmatrix}
				dr\cdot(-) \\
				\cmidrule(lr){1-1}
				\bar{d}_1 \begin{bmatrix} 0 \\ r^2 dt\cdot(-) \end{bmatrix} \\
				d_1 \\
				-C_0 \\
				0
			\end{bmatrix}
			\begin{bmatrix}[c|cccc] \id & 0 & 0 & 0 & 0 \end{bmatrix}}
	\\
	\bullet
		\ar[swap]{r}{K_2 = \begin{bmatrix}
				\bar{d}_2 & 0 & 0 & 0 \\
				0 & C_0 & d_1 & 0 \\
				0 & 0 & -C_1 & d_0 \\
				0 & 0 & 0 & C_2
			\end{bmatrix}}
		\ar[swap,shift right]{u} \&
	\bullet
		\ar[swap,shift right]{u}{
			\begin{bmatrix}
				0 & 0 & 0 & 0 \\
				\cmidrule(lr){1-4}
				\id & 0 & 0 & 0 \\
				0 & \id & 0 & 0 \\
				0 & 0 & \id & 0 \\
				0 & 0 & 0 & \id
			\end{bmatrix}}
		\ar[swap,dashed,bend right=15]{l}{0}
\end{tikzcd}
\end{equation}

\end{subequations}

With two more iterations of Lemma~\eqref{lem:lift-compat}, we construct
the compatibility operators $\bar{K}_3$, $\bar{K}_4$ and $\tK_3$,
$\tK_4$.

\begin{subequations} \label{eq:schw-lift4}

\begin{equation} \label{eq:schw-lift4a}
\begin{tikzcd}[column sep=5cm,row sep=5cm]
	\bullet
		\ar[swap,shift right]{d}
		\ar{r}{\bar{K}_3 = \begin{bmatrix}[c|c|cccc]
			0 & 0 & \bar{d}_3 & 0 & 0 & 0 \\
			0 & 0 & 0 & C_1 & d_1 & 0 \\
			0 & 0 & 0 & 0 & -C_2 & d_0 \\
			0 & 0 & 0 & 0 & 0 & C_3
		\end{bmatrix}} \&
	\bullet
		\ar[dashed,bend left]{l}{0}
		\ar[swap,shift right]{d}{\begin{bmatrix}
				\id & 0 & 0 & 0 \\
				0 & \id & 0 & 0 \\
				0 & 0 & \id & 0 \\
				0 & 0 & 0 & \id
			\end{bmatrix}}
		\ar{r}{\bar{K}_4 = \begin{bmatrix}
				\bar{d}_4 & 0 & 0 & 0 \\
				0 & C_2 & d_1 & 0 \\
				0 & 0 & -C_3 & d_0 \\
				0 & 0 & 0 & C_4
			\end{bmatrix}} \&
	\bullet
		\ar[swap,shift right]{d}{\begin{bmatrix}
				\id & 0 & 0 & 0 \\
				0 & \id & 0 & 0 \\
				0 & 0 & \id & 0 \\
				0 & 0 & 0 & \id
			\end{bmatrix}}
		\ar[dashed,bend left]{l}{0}
	\\
	\bullet
		\ar[swap,shift right]{u}
		\ar[swap]{r}{\tK_3 = \begin{bmatrix}[c|cccc]
				0 & \bar{d}_3 & 0 & 0 & 0 \\
				0 & 0 & C_1 & d_1 & 0 \\
				0 & 0 & 0 & -C_2 & d_0 \\
				0 & 0 & 0 & 0 & C_3
			\end{bmatrix}} \&
	\bullet
		\ar[swap,dashed,bend right]{l}{0}
		\ar[swap,shift right]{u}{\begin{bmatrix}
				\id & 0 & 0 & 0 \\
				0 & \id & 0 & 0 \\
				0 & 0 & \id & 0 \\
				0 & 0 & 0 & \id
			\end{bmatrix}}
		\ar[swap]{r}{\tK_4 = \begin{bmatrix}
				\bar{d}_4 & 0 & 0 & 0 \\
				0 & C_2 & d_1 & 0 \\
				0 & 0 & -C_3 & d_0 \\
				0 & 0 & 0 & C_4
			\end{bmatrix}} \&
	\bullet
		\ar[swap,shift right]{u}{\begin{bmatrix}
				\id & 0 & 0 & 0 \\
				0 & \id & 0 & 0 \\
				0 & 0 & \id & 0 \\
				0 & 0 & 0 & \id
			\end{bmatrix}}
		\ar[swap,dashed,bend right]{l}{0}
\end{tikzcd}
\end{equation}

\begin{equation} \label{eq:schw-lift4b}
\begin{tikzcd}[column sep=5cm,row sep=5cm]
	\bullet
		\ar[swap,shift right]{d}
		\ar{r}{\tK_3 = \begin{bmatrix}[c|cccc]
			0 & \bar{d}_3 & 0 & 0 & 0 \\
			0 & 0 & C_1 & d_1 & 0 \\
			0 & 0 & 0 & -C_2 & d_0 \\
			0 & 0 & 0 & 0 & C_3
		\end{bmatrix}} \&
	\bullet
		\ar[dashed,bend left]{l}{0}
		\ar[swap,shift right]{d}{\begin{bmatrix}
				\id & 0 & 0 & 0 \\
				0 & \id & 0 & 0 \\
				0 & 0 & \id & 0 \\
				0 & 0 & 0 & \id
			\end{bmatrix}}
		\ar{r}{\tK_4 = \begin{bmatrix}
				\bar{d}_4 & 0 & 0 & 0 \\
				0 & C_2 & d_1 & 0 \\
				0 & 0 & -C_3 & d_0 \\
				0 & 0 & 0 & C_4
			\end{bmatrix}} \&
	\bullet
		\ar[swap,shift right]{d}{\begin{bmatrix}
				\id & 0 & 0 & 0 \\
				0 & \id & 0 & 0 \\
				0 & 0 & \id & 0 \\
				0 & 0 & 0 & \id
			\end{bmatrix}}
		\ar[dashed,bend left]{l}{0}
	\\
	\bullet
		\ar[swap,shift right]{u}
		\ar[swap]{r}{K_3 = \begin{bmatrix}
				\bar{d}_3 & 0 & 0 & 0 \\
				0 & C_1 & d_1 & 0 \\
				0 & 0 & -C_2 & d_0 \\
				0 & 0 & 0 & C_3
			\end{bmatrix}} \&
	\bullet
		\ar[swap,dashed,bend right]{l}{0}
		\ar[swap,shift right]{u}{\begin{bmatrix}
				\id & 0 & 0 & 0 \\
				0 & \id & 0 & 0 \\
				0 & 0 & \id & 0 \\
				0 & 0 & 0 & \id
			\end{bmatrix}}
		\ar[swap]{r}{K_4 = \begin{bmatrix}
				\bar{d}_4 & 0 & 0 & 0 \\
				0 & C_2 & d_1 & 0 \\
				0 & 0 & -C_3 & d_0 \\
				0 & 0 & 0 & C_4
			\end{bmatrix}} \&
	\bullet
		\ar[swap,shift right]{u}{\begin{bmatrix}
				\id & 0 & 0 & 0 \\
				0 & \id & 0 & 0 \\
				0 & 0 & \id & 0 \\
				0 & 0 & 0 & \id
			\end{bmatrix}}
		\ar[swap,dashed,bend right]{l}{0}
\end{tikzcd}
\end{equation}

\end{subequations}

From this point on, the complexes $\bar{K}_i$, $\tK_i$ become identical
with $K_i$ from~\eqref{eq:schw-calabi}.

\bigskip

\begin{thm} \label{thm:schw}
Consider the family of $n$-dimensional ($n\ge4$) spacetimes $(\bar{\M},
\bar{g})$ introduced at the top of Section~\ref{sec:schw}, warped
products of a static $2$-dimensional factor $(\M,g)$ and a constant
curvature factor $(\S,\Omega)$ with sectional curvature $\alpha$, which
includes the higher dimensional Schwarzschild
(Schwarzschild-Tangherlini), Taub and pseudo-Schwarzschild solutions,
possibly with a nonzero cosmological constant. The full
compatibility complex $\bar{K}_i$ for the Killing operator $\bar{K}_0 =
\bar{K}$~\eqref{eq:schw-killing} is given by
\begin{subequations} \label{eq:schw-complex}
\begin{align}
	\bar{K}_0 &=
	\begin{bmatrix}[c|cc]
		\frac{1}{f} (2f \del_r + \frac{f_1}{r}) & 0 & 0 \\
		\cmidrule(lr){1-1} \cmidrule(lr){2-3}
		dr\,\frac{1}{f} \del_t - dt\,\frac{f_1}{2r} & \nabla & 0 \\
		dr\, \frac{f}{r^2} D \frac{1}{f} & dt\, \frac{f}{r^2} D & \nabla  \\
		2\Omega \frac{f}{r} & 0 & C_0
	\end{bmatrix}
	, \\
	\bar{K}_1 &=
	\begin{bmatrix}[c|ccc]
		\id & 0 & 0 & 0 \\
		\cmidrule(lr){1-1} \cmidrule(lr){2-4}
		0 & \begin{bmatrix} 0 & -dr\cdot(-) \end{bmatrix} & dr\cdot(-) & 0 \\[1ex]
		\cmidrule(lr){1-1}
		\cmidrule(lr){2-4}
		0 &
			\bar{d}_1 \begin{bmatrix} \id & 0 \\ 0 & -r^2 dt\cdot(-) \end{bmatrix} &
			\bar{d}_1 \begin{bmatrix} 0 \\ r^2 dt\cdot(-) \end{bmatrix} & 0 \\
		0 & 0 & d_1 & 0 \\
		0 & 0 & -C_0 & d_0 \\
		0 & 0 & 0 & C_1
	\end{bmatrix}
	\begin{bmatrix}[c|ccc]
		\id & 0 & 0 & 0 \\
		\cmidrule(lr){1-1} \cmidrule(lr){2-4}
		0 &
			\begin{bmatrix}\id \\ 0\end{bmatrix} &
			\begin{bmatrix}0 \\ \id\end{bmatrix} &
			\begin{bmatrix}0 \\ 0\end{bmatrix} \\
		0 & J_2^w & J_2^Y & J_2^Z \\
		0 & 0 & 0 & \id
	\end{bmatrix} \notag \\
	& \qquad {}
	\left(
		\begin{bmatrix}[c|ccc]
			\id & 0 & 0 & 0 \\
			\cmidrule(lr){1-1} \cmidrule(lr){2-4}
			0 & \id & 0 & 0 \\
			0 & 0 & \id & 0 \\
			0 & 0 & 0 & \id
		\end{bmatrix}
		- \bar{K} \begin{bmatrix} J_1 \\ \cmidrule(lr){1-1} 0 \\ 0 \end{bmatrix} \right) , \\
	\bar{K}_2 &=
	\begin{bmatrix}[c|c|cccc]
		\multicolumn{6}{c}{H_{J_1}} \\
		\cmidrule(lr){1-6}
		0 & \multicolumn{5}{c}{H_{J_2}} \\
		\cmidrule(lr){1-1}
		\cmidrule(lr){2-6}
		0 & 0 & \bar{d}_2 & 0 & 0 & 0 \\
		0 & 0 & 0 & C_0 & d_1 & 0 \\
		0 & 0 & 0 & 0 & -C_1 & d_0 \\
		0 & 0 & 0 & 0 & 0 & C_2
	\end{bmatrix} , \\
	\bar{K}_3 &=
	\begin{bmatrix}[c|c|cccc]
		0 & 0 & \bar{d}_3 & 0 & 0 & 0 \\
		0 & 0 & 0 & C_1 & d_1 & 0 \\
		0 & 0 & 0 & 0 & -C_2 & d_0 \\
		0 & 0 & 0 & 0 & 0 & C_3
	\end{bmatrix} , \\
	\bar{K}_i &=
	\begin{bmatrix}
		\bar{d}_i & 0 & 0 & 0 \\
		0 & C_{i-2} & d_1 & 0 \\
		0 & 0 & -C_{i-1} & d_0 \\
		0 & 0 & 0 & C_i
	\end{bmatrix}
		\quad (3 < i < n-2) , \\
	\bar{K}_{n-2} &=
	\begin{bmatrix}
		\bar{d}_{n-2} & 0 & 0 & 0 \\
		0 & C_{n-4} & d_1 & 0 \\
		0 & 0 & -C_{n-3} & d_0
	\end{bmatrix} , \\
	\bar{K}_{n-1} &=
	\begin{bmatrix}
		\bar{d}_{n-1} & 0 & 0 \\
		0 & C_{n-3} & d_1
	\end{bmatrix} , \\
	\bar{K}_{i} &= 0 \quad (n\le i) .
\end{align}
\end{subequations}
where $f(r)$ is defined in~\eqref{eq:gen-f} and $f_1 = r f'(r)$, $\bar{d}_i$ and $d_i$ denote the exterior derivatives on
$i$-forms, on $\bar{\M}$ and $\M$ respectively, while $D$ and $C_i$ are
the covariant derivative and the Calabi complex
operators~\eqref{eq:calabi} on $(\S,\Omega)$, and we have also used the
operators $J_1$~\eqref{eq:schw-J1},
$J_2$~\eqref{eq:schw-J2}, $H_{J_1}$~\eqref{eq:schw-HJ1-tHJ1},
$H_{J_2}$~\eqref{eq:schw-HJ2}. (See Appendix~\ref{sec:schw-notation} for
a more complete summary of the notation.)
\end{thm}

While we have unambiguously defined the operators $J_1$, $J_2$,
$H_{J_1}$, and $H_{J_2}$, we have not computed them explicitly. For our
purposes here, it is sufficient that they exist and satisfy a few
defining properties. Of course, in individual cases, they could be
easily computed using computer algebra.

\begin{proof}
The proof is very much parallel to the proof of Theorem~\ref{thm:flrw}.
We start with the knowledge that the complex~\eqref{eq:schw-calabi} is a
full compatibility complex. Then, gluing together (from left to right)
the diagrams~\eqref{eq:schw-lift1}, \eqref{eq:schw-lift2},
\eqref{eq:schw-lift3} and~\eqref{eq:schw-lift4}, we observe that the
glued diagrams satisfy the hypotheses of
Lemma~\ref{lem:compat-sufficient}. This implies, that $\tK_i$ is a full
compatibility complex as well, which in turn implies that so is
$\bar{K}_i$, whose operators we have explicitly listed
in~\eqref{eq:schw-complex}.
\end{proof}

The non-vanishing ranks of the vector bundles in the $\bar{K}_i$ complex
have the following pattern, which can be compared to similar table for
the constant curvature (Section~\ref{sec:cc}) and FLRW cases
(Section~\ref{sec:flrw}, where $m=n-1$, for comparison):

\begin{center}
\begin{tabular}{c|cccc}
	& $n=4$ & $n=5$ & $n=6$ & $n\ge 4$ \\ \hline
	& $4$ & $5$ & $6$ & $(n-2)+1+1$ \\
	$\bar{K}_0$ & \\
	& $10$ & $15$ & $21$ & $\frac{(n-1)(n-2)}{2} + 2(n-2) + 3$ \\
	$\bar{K}_1$ & \\
	& $18$ & $35$ & $64$ & $\frac{n-2}{2}\binom{n-1}{3} + (n-1)(n-2) + (n-2) + \binom{n}{2} + (n-2) + 1$ \\
	$\bar{K}_2$ & \\
	& $12$ & $35$ & $95$ &
		$(n-2)\binom{n-1}{4}
		+ (n-2)\binom{n-1}{3}
		+ \frac{(n-1)(n-2)}{2}
		+ \binom{n}{3} + (n-2) + 1$ \\
	$\bar{K}_3$ & \\
	& $2$ & $17$ & $81$ &
		$\frac{3(n-2)}{2}\binom{n-1}{5}
		+ 2(n-2)\binom{n-1}{4}
		+ \frac{(n-2)}{2}\binom{n-1}{3}
		+ \binom{n}{4}$ \\
	$\vdots$ & & & & $\vdots$ \\
	& & & &
		$\frac{(n-2)(i-1)}{2}\binom{n-1}{i+1}
		+ (n-2)(i-2)\binom{n-1}{i}
		+ \frac{(n-2)(i-3)}{2}\binom{n-1}{i-1}
		+ \binom{n}{i}$ \\
	\llap{($3<i$)\quad}
	$\bar{K}_i$ & \\
	& & & &
		$\frac{(n-2)i}{2}\binom{n-1}{i+2}
		+ (n-2)(i-1)\binom{n-1}{i+1}
		+ \frac{(n-2)(i-2)}{2}\binom{n-1}{i}
		+ \binom{n}{i+1}$ \\
	$\vdots$ & & & & $\vdots$ \\
	$\bar{K}_{n-1}$ \\
	& $2$ & $4$ & $7$ & $\frac{(n-2)(n-3)}{2} + 1$
\end{tabular}
\end{center}

\begin{rem}
In $n=4$ dimensions, it is well-known~\cite{jezierski,swaab} that, for
practical purposes, taking the linearized Einstein equations into
account, the gauge invariant degrees of freedom for linear perturbations
on the Schwarzschild background reduce to the Regge-Wheeler (axial) and
Zerilli (polar) scalars, or equivalently the complex Teukolsky scalar.
It is even possible to give the Regge-Wheeler and Zerilli scalars local
and manifestly gauge-invariant definitions, based on the linearization
of curvature tensors vanishing on the Schwarzschild
background~\cite{dotti-schw}. However, it is also known that there exist
so-called \emph{algebraically special} modes that are not pure gauge but
lie in the kernel of these gauge-invariants~\cite{whiting-price}. Hence,
this small set of invariants cannot be considered complete in our sense.
In our construction, $\bar{K}_1$ has $18$ independent components
(without taking the linearized Einstein equations into account, though).
But our construction proves that they form a complete set of local
gauge-invariants.
\end{rem}

\begin{rem} \label{rem:schw-K1}
In analogy with Remark~\ref{rem:flrw-K1} about FLRW geometries, it is
worth noting that the $\bar{K}_i$ complex presented above is not
continuously deformable through the class of gST spacetimes to the $M=0$
case, which corresponds to the constant curvature limit. The main
reason, again, is that the operators $J_1$ and $J_2$, introduced in
Equations~\eqref{eq:schw-J1} and~\eqref{eq:schw-J2}, are proportional to
$1/M$ and hence diverge in that limit. These operators, together with
their factorizations $H_{J_1}$ and $H_{J_2}$ appear in several places in
the formulas~\eqref{eq:schw-complex} for the operators $\bar{K}_i$ (for
$i\ge 1$). Thus, also in this case, it would be difficult to compare the
local gauge-invariant components of our $\bar{K}_1$ operator to the
components of the linearized Riemann operator, which would also be given
by long and unenlightening expressions. A more fruitful comparison would
be to try to express the components of our $\bar{K}_1$ in terms of the
linearized IDEAL characterization tensors that were recently constructed
for the gST geometries~\cite{kh-gst}.

However, the intuition proposed in the second paragraph of
Remark~\ref{rem:flrw-K1} still largely applies to the components of our
$\bar{K}_1$. In particular, our $J_1$ operator is directly analogous to
the $J$ operator introduced for FLRW geometries. On the other hand, the
$J_2$ operator did not have a direct analogy, so the way it induces
gauge invariant components of $\bar{K}_1$ is slightly different.
\end{rem}

\section{Discussion} \label{sec:discuss}

In this work, we have studied the construction of the compatibility
complex (Definition~\ref{def:compat}) $K_l$, $l=0,1,2,\ldots$, for a
linear differential operator $K_0$ of regular finite type
(Definition~\ref{def:fintype}). The construction proceeds by putting the
operator $K_0$ into a canonical form of a flat connection and then
lifting the resulting twisted de~Rham complex to a compatibility complex
for $K_0$ (Theorem~\ref{thm:fintype-compat}). Our primary and motivating
example of an operator of regular finite type is the Killing operator $K_{ab}[v]
= \nabla_a v_b + \nabla_b v_a$ on a Lorentzian (or even
pseudo-Riemannian) manifold $(M,g)$. Once known, the components of the
first compatibility operator $K_1$ can be interpreted (as discussed in
the Introduction) as a complete set of local gauge-invariant observables
in linearized gravity on $(M,g)$.

We have applied the abstract construction of Section~\ref{sec:compat} to
several physically motivated examples: flat (Minkowski) and constant
curvature (de~Sitter or anti-de~Sitter) spacetimes in
Section~\ref{sec:cc}, cosmological (FLRW) spacetimes in
Section~\ref{sec:flrw}, (Schwarzschild-Tangherlini) spherically
symmetric black hole spacetimes%
	\footnote{The family of spacetimes considered Section~\ref{sec:schw}
	is actually richer than just asymptotically flat spherically symmetric
	black holes (the Schwarzschild-Tangherlini ones). More generally, it
	allows for a non-zero cosmological constant and also allows to
	substitute spherical symmetry for planar or pseudo-spherical symmetry, which
	respectively give rise Taub's plane symmetric spacetimes or to
	pseudo-Schwarzschild solutions.} %
in Section~\ref{sec:schw}. In each case, we have kept the dimension $n =
\dim M$ general, allowing at least $n\ge 4$. While the contents of
Section~\ref{sec:cc} are well-known (they were previously reviewed in
more detail in~\cite{kh-calabi}), the Killing compatibility complexes
constructed in Sections~\ref{sec:flrw} and~\ref{sec:schw} are new.

One may wish to compare the main result for FLRW geometries,
Theorem~\ref{thm:flrw}, with the recent works~\cite{fhh,cdk,fhk}, which
were the first to (a) construct, (b) give a geometric interpretation to and
(c) prove completeness for the first compatibility operator $K_1$ in a
context very similar the one considered in  Section~\ref{sec:flrw} (the
difference is that here we do not include the presence of a dynamical
scalar inflaton field on an cosmological FLRW geometry). The systematic
approach developed in this work can also be easily applied in the
presence of an inflaton field. Then, the systematically constructed
compatibility operator $K_1$ would be necessarily equivalent to what was
obtained in~\cite{fhh,cdk,fhk}. The difference is that our systematic
construction automatically comes with a proof of completeness, while the
previous proof of completeness given in~\cite{fhh} relied very heavily
on parallels with known results for the flat and constant curvature
cases~\cite{higuchi,kh-calabi}, without an obvious way to generalize it.
On the other hand, our systematic construction does not give a
Stewart-Walker-like (cf.\ the introduction to Section~\ref{sec:killing})
geometric interpretation to $K_1$ as a linear local gauge-invariant
observable. On the other hand, the approach put forward
in~\cite{cdk,fhk}, of constructing a \emph{candidate} $K_1$ by
linearizing an \emph{IDEAL characterization} of the background geometry,
automatically gives $K_1$ a Stewart-Walker-like geometric interpretation,
but does not automatically prove completeness.%
	\footnote{Although, the only possibility we know in which completeness
	might fail is when the IDEAL characterization tensors vanish at
	\emph{quadratic} or \emph{higher} order when approaching the isometry
	class of the characterized geometry in the space of metrics. Then
	their \emph{linearization} might fail to capture all of the linear
	invariants.} %
Thus, we see great
potential in joining the methods of the current work with those
of~\cite{cdk,fhk} to construct universal Killing compatibility operators
(equivalently, complete sets of linear local gauge-invariant
observables) on a variety of backgrounds, while getting the benefits of
straightforward geometric interpretation and of a systematic way to
prove completeness.

For the Schwarzschild black hole (and its higher dimensional
generalizations), the Regge-Wheeler and Zerilli local gauge-invariants
have been known for a long time~\cite{ki-master}. Other local
gauge-invariants have also been proposed (see~\cite{jezierski,swaab,ab-kerr}
for a brief review). However, to our knowledge, no
claim of completeness has ever been made for an explicit set of local
gauge-invariants on Schwarzschild. Thus, even our construction of the
first compatibility $K_1$ operator in Section~\ref{sec:schw} appears to
be new. On the other hand, the $4$-dimensional Schwarzschild black hole
does have a known IDEAL characterization~\cite{fs-schw}, recently
extended to higher dimensions~\cite{kh-gst}, so as was argued in the
previous paragraph its linearization would have provided a good
candidate for $K_1$. To our knowledge, this has not been done explicitly
in the literature. Again, comparing that heuristic construction with our
systematic approach would be very interesting.

The next logical step is to apply our methods to the Kerr black hole and
higher dimensional (Myers-Perry) generalizations. As a first step, we
intend to construct a Killing compatibility complex for the Kerr
geometry~\cite{aabkw}, thus providing a proof of completeness for the
list of local gauge-invariants recently proposed in~\cite{ab-kerr}.

Once the Killing compatibility complex is known on a given geometry,
this information has interesting applications to the symplectic and
Poisson structures on the space of solutions of linearized
gravity~\cite[Sec.5]{kh-calabi}.

\paragraph{Acknowledgments.}
The author thanks S.~Aksteiner, L.~Andersson and T.~B\"ackdahl for many
useful discussions on gauge-invariant observables, as well as
B.~Kruglikov and W.~Seiler for pertinent feedback at the early stages of
this work. Also, the author was partially supported by the GA\v{C}R
project 18-07776S and RVO: 67985840.

\appendix

\section{Flat connection form for PDEs of regular finite type} \label{sec:fintype}

We have chosen to express our definition of a PDE of \emph{(regular)
finite type} (Definition~\ref{def:fintype}) directly in terms of a
\emph{flat connection} (Definition~\ref{def:flat-conn}). Elsewhere in
the literature, the definition is given in different terms, but the
equivalence with the form of a flat connection is well-known, even for
nonlinear equations, once the notion of regularity is made precise. For
instance, the contents of Rmk.2.3.3, Rmk.2.3.6, and Ex.2.3.17
of~\cite{seiler-inv} concern precisely the equivalence between these two
possible definitions.

For the convenience of the reader, we give an elementary proof of this
equivalence for linear equations, which is all that we will need. When
we speak of \emph{equivalence} below, we mean in the sense of one
operator complexes of Definition~\ref{def:homalg}. The length of our
proof mostly reflects the amount of notation that we have needed to
introduce along the way to make the argument as explicit as possible.

For the following definition, we need to quickly introduce the notion of
jets and jet bundles~\cite[Secs.2.1--2]{seiler-inv}. Given a vector
bundle $V\to M$, the \emph{$N$-jet} $j^N_x v$ at $x\in M$ of a section
$v\in \Secs(V)$ is the equivalence class of sections that have the same
Taylor expansion about $x$ up to order $N$ in local adapted coordinates
on $V\to M$. The definition clearly does not depend on the choice of
adapted coordinates. Denote by $J_x^N V$ the vector space of all
$N$-jets at $x$ and let $J^N V = \bigsqcup_{x\in M} J_x^N V$ be the
\emph{$N$-jet bundle} of $V$, which can naturally be given the structure
of a smooth vector bundle $J^N V \to M$, with $J^0 V = V$. By throwing
away higher terms of Taylor series expansions, we can define natural
projections $\pi^N_{N'}\colon J^N V \to J^{N'} V$ when $N\ge N'$. By
assigning to a section $v\in \Secs(V)$ its own $N$-jet at each point of
$M$, we can define the natural \emph{$N$-jet extension} differential
operator $j^N\colon \Secs(V) \to \Secs(J^N V)$, which is universal in
the sense that for any differential operator $D\colon \Secs(V) \to
\Secs(W)$ of order at most $N$, there exists a unique vector bundle map
$p^0(D) \colon J^N V \to W$ such that $D[v] = d(j^N v)$. Extending this
notation, we denote by $p^l(D) = p^0(j^l \circ D)$ the \emph{$l$-th
prolongation} of $D$.

\begin{dfn} \label{def:fintype-orig}
Let $V\to M$ and $W\to M$ be vector bundles and $K\colon \Secs(V) \to
\Secs(W)$ a linear differential operator. The PDE $K[v] = 0$ is said to
be of \emph{finite type} when (a) locally there exists an integer $N <
\oo$, a vector bundle morphism $\kappa\colon J^NV \to J^{N+1}V$ and a
differential operator $\lambda\colon \Secs(W) \to \Secs(J^{N+1}V)$ such
that $j^{N+1} v - \kappa(j^N v) = \lambda[K[v]]$ for any $v \in
\Secs(V)$, and moreover \emph{regular} when (b) locally the dimension of
the solution space is finite and constant.
\end{dfn}

\begin{prp} \label{prp:fintype-conn}
If a differential operator $K \colon \Secs(V) \to \Secs(W)$ between
vector bundles $V\to M$, $W \to M$ defines a PDE of regular finite type
$K[v]=0$, then $K$ is equivalent to a flat connection operator
$\bar{\DD}$ on some vector bundle $\bar{U}\to M$.
\end{prp}

\begin{proof}
Our proof will consist of three steps: (a) equivalence of the original
differential operator $K$ to a connection $\DD$ on $J^N V$ together with
a non-differential constraint $E\colon J^N V \to W'$, (b) equivalence to
the restriction $\tilde{\DD}$ of $\DD$ to the sub-bundle
$\tilde{U}\hookrightarrow J^N V\to M$ satisfying the non-differential
constraint $E(\tilde{u}) = 0$, and (c) equivalence to the restriction
$\bar{\DD}$ of $\tilde{\DD}$ to the sub-bundle $\bar{U}\hookrightarrow
\tilde{U} \to M$ spanned by flat sections.

Since all of our definitions and claims are local, we might as well work
in local adapted coordinates on $V$, $W$ and any other vector bundles.
For instance, we will use coordinates $(x^a)$ on $M$, $(x^a,v^\alpha)$
on $V$ and $(x^a,v^\alpha, v^{\alpha,1}, \ldots, v^{\alpha,N})$ on $J^N V$,
such that $v_{\alpha,k}(j^N \phi(x)) = \del^k v^\alpha(\phi(x))$, where
$\del^k$ stands for all possible independent partial derivatives of
order $k$ with respect to the $(x^a)$ coordinates, with similar
notations used for other bundles. Also, when there is no confusion, we
will denote a general section of a vector bundle $V$ by $v$, a general
section of $W$ by $w$, and so on. We will denote a general section of
$J^NV$ by $v^{(N)} = (v, v^1, \ldots, v^N)$.

Before proceeding, let us establish some notation. Namely, supposing
that $K$ is a differential operator of order $k$, there is a unique
non-differential bundle map that factors $K$ through $k$-jets, which we
denote by $p^0(K)(j^k v) = K[v]$ and similarly $p^l(K)(j^{k+l} v) = j^l
K[v]$. Also, we will need the algebraic operators $\iota_k$ defined by
the identity%
	\footnote{The non-triviality of the identity stems from the fact that
	$\del^k(\del^l v)$ has more components than $\del^{k+l} v$, if we
	do not symmetrize the partial derivatives between the $\del^k$ and
	$\del^l$ operators.} %
$\del^k(\del^l v) =
\iota_k(\del^{k+l} v)$, as well as the differential operators
$\Delta^N_k$ defined by the identity
\begin{equation}
	\Delta^N_k[\del v - \iota_1(v^1), \ldots, \del v^{N-1} - \iota_1(v^N),
		\del v^N - \iota_1(v^{N+1})]
	= \del^k v - \iota_k(v^k) ,
\end{equation}
for $k\le N$. These operators basically encode the identities $\del^0 v
- v = 0$, $\del^1 v - \iota_1(v^1) = \del v - \iota_1(v^1)$, $\del^2 v -
\iota_2(v^2) = \del(\del v - \iota_1(v^1)) + (\del v^1 - \iota_1(v^2))$,
\ldots .

(a)
Essentially, all the information that we will need to establish the
first equivalence is contained in the bundle map $\kappa$ and the
differential operator $\lambda$ from Definition~\ref{def:fintype}. The
non-trivial information is contained in the highest order components,
$\del^{N+1} v - \kappa^{N+1}(v,\del v, \ldots, \del^N v) =
\lambda^{N+1}[K[v]]$, which implies the identity
\begin{equation}
	\del(\del^N v) - \iota_1(\kappa^{N+1}(v,\del v, \ldots, \del^N v))
		= \iota_1(\lambda^{N+1}[K[v]]) .
\end{equation}
That last identity will be a crucial piece of our definition of a
connection on $J^N V$.

To complete the necessary definitions, we will need the differential
operator $P_1\colon \Secs(W) \to \Secs(T^*M\otimes_M J^N V)$ given by
\begin{equation}
	P_1[w] = \begin{bmatrix}
		0 \\ \vdots \\
		0 \\
		\iota_1(\lambda^{N+1}[w])
	\end{bmatrix} .
\end{equation}
Next, the bundle $W'$ and the bundle map $E\colon J^N \to W'$ are chosen
so that $\ker E = \pi^{N'}_0 (\ker p^{N'}(K))$ for some $N'$. Since we
are allowed to specify $N$ and $N'$ as we like, we can pick them so that
$N' > 0$, $N > k$. The meaning is that $E$ takes into account all
integrability conditions of order $N$ that can be obtained by prolonging
the equation $K[v] = 0$ by $N'$ differentiations. By construction, there
must exist a differential operator $P_0\colon \Secs(W) \to \Secs(W')$
such that $P_0[K[v]] = E(j^N v)$. Also, since we have presumed that $N >
k$, we have $\pi^N_k(\ker E) \subset \ker p^0(K)$ and there must exist a
bundle map $\bar{P}_0 \colon W'\to W$ such that $P_0 E = p^0(K)
\pi^N_k$. It remains to define the connection on $J^N V$, which we give
by the formula
\begin{equation}
	\DD \begin{bmatrix} v \\ v^1 \\ \vdots \\ v^{N-1} \\ v^N \end{bmatrix}
	= \begin{bmatrix}
		\del v - \iota_1(v^1) \\
		\del v^1 - \iota_1(v^2) \\
		\vdots \\
		\del v^{N-1} - \iota_1(v^N) \\
		\del v^N - \iota_1(\kappa^{N+1}(v,v^1, \ldots, v^N))
	\end{bmatrix} .
\end{equation}

Having introduced all the necessary notation. The desired equivalence is
explicitly exhibited by the diagram
\begin{equation}
\begin{tikzcd}[column sep=4cm,row sep=4cm]
	v \ar{r}{K}
		\ar[swap,shift right]{d}{j^N} \&
	w \ar[swap,shift right]{d}{\begin{bmatrix} P_1 \\ P_0 \end{bmatrix}}
		\ar[dashed,bend left]{l}{0}
	\\
	v^{(N)}
		\ar[swap]{r}{\begin{bmatrix} \DD \\ E \end{bmatrix}}
		\ar[swap,shift right]{u}{\pi^N_0} \&
	\begin{bmatrix} v^{(N)}_1 \\ w' \end{bmatrix}
		\ar[swap,shift right]{u}{\begin{bmatrix} p^0(K)\Delta^N_k & \bar{P}_0 \end{bmatrix}}
		\ar[swap,dashed,bend right]{l}{\begin{bmatrix} -\Delta^N_N & 0 \end{bmatrix}}
\end{tikzcd} ,
\end{equation}
where $v^N_1$ denotes a general section of $T^*M \otimes_M J^N V \to M$.
To prove that we have an equivalence, we must verify all the conditions
required by Definition~\ref{def:homalg}. The commutativity of the
squares formed by solid arrows follows from direct computations. One is
involves only the defining properties of $P_0$ and $P_1$, while the
other uses the definitions of $\Delta^N_k$ and $\bar{P}_0$:
\begin{align*}
	\begin{bmatrix}
		p^0(K)\Delta_k^N & \bar{P}_0
	\end{bmatrix}
	\begin{bmatrix} \DD \\ E \end{bmatrix}
	- K\pi^N_0
	&= p^0(K) (\Delta^N_k \DD) + (\bar{P}_0 E) - p^0(K)\pi^N_k (j^N \pi^N_0) \\
	&= p^0(K) \pi^N_k (j^N \pi^N_0 - \id) + p^0(K) \pi^N_k (\id - j^N \pi^N_0)
	= 0 .
\end{align*}
The remaining checks involve the homotopy corrections:
\begin{align*}
	\pi_0^N j^N &= \id - 0 ,
	\\
	j^N\pi_0^N
	&= \id + (j^N\pi_0^N - \id) \\
	&= \id + \Delta^N_N \DD ,
\end{align*}
\begin{align*}
	\left( \id
	- \begin{bmatrix}
		p^0(K)\Delta_k^N & \bar{P}_0
	\end{bmatrix}
	\begin{bmatrix}
		P_1 \\ P_0
	\end{bmatrix}
	- 0
	\right) K
	&= K - p^0(K) \Delta^N_k (P_1 K) - \bar{P}_0 (P_0 K) \\
	&= K - p^0(K) (\Delta^N_k \DD) j^N - (\bar{P}_0 E) j^N \\
	&= K - p^0(K) (j^k \pi^N_0 - \pi^N_k) j^N - p^0(K) j^k
	= 0 ,
\end{align*}
\begin{multline*}
	\left( \begin{bmatrix} \id & 0 \\ 0 & \id \end{bmatrix}
	- \begin{bmatrix}
		P_1 \\ P_0
	\end{bmatrix}
	\begin{bmatrix}
		p^0(K)\Delta_k^N & \bar{P}_0
	\end{bmatrix}
	- \begin{bmatrix} \DD \\ E \end{bmatrix}
		\begin{bmatrix} -\Delta^N_N & 0 \end{bmatrix}
	\right)
	\begin{bmatrix} \DD \\ E \end{bmatrix}
	\\
	=
	\begin{bmatrix} \DD \\ E \end{bmatrix}
	- \begin{bmatrix}
		P_1 \\ P_0
	\end{bmatrix}
		p^0(K) \left( (j^k\pi^N_0 - \pi^N_k) + \pi^N_k \right)
	- \begin{bmatrix} \DD \\ E \end{bmatrix}
		( \id - j^N\pi^N_0 )
	\\
	=
	\left(
	\begin{bmatrix} \DD \\ E \end{bmatrix} j^N
	- \begin{bmatrix} P_1 \\ P_0 \end{bmatrix} K
	\right)
	\pi^N_0
	=
	\begin{bmatrix} 0 \\ 0 \end{bmatrix} .
\end{multline*}

(b)
The next step is to eliminate the non-differential $E(v,v^1,\ldots, v^N)
= 0$ constraint. By introducing a sub-bundle $\iota\colon \tilde{U}
\hookrightarrow J^NV \to M$ such that $\tilde{U} = \ker E$. In general
$\ker E$ need not be a vector bundle (the fiber ranks may be
non-constant over $M$), hence $\tilde{U}$ might not exist as a bundle.
However, from requirement in Definition~\ref{def:fintype-orig}(b), we
know that the dimension of the solution space of $K[v] = 0$ is finite
and locally constant, which by part (a) of our proof also applies to the
solution space of $\DD[v^{(N)}] = 0$ under the $E$-constraint.
In fact, the dimension of the solution space on a neighborhood of $x\in
M$ is bounded from above by $\dim \ker_x E$, simply because a $\DD$-flat
section is uniquely determined by its value at any one point. The only
reason that an element $\tilde{u}_x \in \ker_x E$ might not correspond
to a local solution is that there might exist some higher order
differential consequence of $\DD[v^{(N)}] = 0$ (or equivalently
of $K[v] = 0$) that imposes further integrability conditions on $j^N v$,
which $\tilde{u}_x$ may not satisfy. However, from the theory of formal
integrability of PDEs (\emph{Cartan-Kuranishi}
theorem~\cite[Sec.7.4]{seiler-inv}), it is well known that past a
certain finite differential order $N'$, no further constraints on $j^Nv$
will appear from considering $\del^{N'} K[v] = 0$ or higher order
differential consequences. Let us use this order $N'$ (or any higher
one) to influence the definition of the $E$-constraint that we
introduced in part (a) of the proof. That is, we are free to assume that
$E$ has been chosen such that every element $\tilde{u}_x \in \ker_x E$
defines a unique local solution%
	\footnote{The existence of such a solution is guaranteed by applying
	$\DD$-parallel transport.} %
of $K[v]=0$ with $j^N v(x) = \tilde{u}_x$. In other words, $\dim \ker_x
E$ is equal to the local dimension of the solution space about $x\in M$.
But then, by the regular finite type hypothesis on $K[v]=0$, we know that $\dim
\ker_x E$ is locally constant, meaning that $\ker E$ is indeed a vector
bundle, which we can denote by $\iota \colon \tilde{U} \to J^N V \to M$.

Since we are working locally, we are free to presume that there also
exists a projection bundle map $q \colon J^N V \to \tilde{U}$ such that
$q\iota = \id$. In the other direction, we have the identity $(\id -
\iota q) \iota = 0$, which means that there must exist a bundle map
$h\colon W' \to J^NV$ such that $\iota q = \id - hE$. Further, we can
define a connection operator $\tilde{\DD}$ on $\tilde{U}$ by the formula
\begin{equation}
	\tilde{\DD} \tilde{u} = q_1 \DD \iota(\tilde{u}) ,
\end{equation}
where we have introduced the convenient notation $q_1 = \id\otimes q
\colon T^*M \otimes_M J^NV \to T^*M\otimes_M \tilde{U}$. We will use the
same convention also for $E_1 = \id \otimes E$ and $\iota_1 = \id
\otimes \iota$. We finally have all the ingredients to exhibit the next
equivalence
\begin{equation}
\begin{tikzcd}[column sep=4cm,row sep=4cm]
	v^{(N)} \ar{r}{\begin{bmatrix} \DD \\ E \end{bmatrix}}
		\ar[swap,shift right]{d}{q} \&
	\begin{bmatrix} v^{(N)}_1 \\ w' \end{bmatrix}
		\ar[swap,shift right]{d}{\begin{bmatrix} q_1 & -q_1 \DD h \end{bmatrix}}
		\ar[dashed,bend left]{l}{\begin{bmatrix} 0 & h \end{bmatrix}}
	\\
	\tilde{u}
		\ar[swap]{r}{\tilde{\DD}}
		\ar[swap,shift right]{u}{\iota} \&
	\tilde{u}_1
		\ar[swap,shift right]{u}{\begin{bmatrix} \iota_1 \\ 0 \end{bmatrix}}
		\ar[swap,dashed,bend right]{l}{0}
\end{tikzcd} .
\end{equation}
For the commutativity of the solid arrow squares we first need one more
identity. Let us define
\begin{equation}
	E' = h_1 E_1 \DD - \DD h E
\end{equation}
and note that it is a non-differential operator (as follows from the
basic properties of the connection operator). From its definition,
$E'(v^{(N)}) = 0$ is also an integrability condition. But by the
discussion from the preceding paragraph, $E(V^{(N)}) = 0$ already takes
into account all possible integrability conditions. Hence, we must be
able to factor $E' = h' E$ for some bundle map $h' \colon W'\to T^*M
\otimes_M W'$, from which follows the identity
\begin{equation}
	h_1 E_1 \DD = \DD h E + h' E = (\DD h + h') E .
\end{equation}
Hence,
\begin{align*}
	\begin{bmatrix} \DD \\ E \end{bmatrix} \iota
	- \begin{bmatrix} \iota_1 \\ 0 \end{bmatrix} \tilde{\DD}
	&= \begin{bmatrix}
		\DD\iota - (\iota_1 q_1) \DD \iota \\
		E \iota \\
		\end{bmatrix}
	= \begin{bmatrix}
		(h_1 E_1 \DD) \iota \\
		0
		\end{bmatrix}
	= \begin{bmatrix}
		(\DD h + h') (E\iota) \\
		0
		\end{bmatrix}
	= \begin{bmatrix} 0 \\ 0 \end{bmatrix} ,
	\\
	\tilde{\DD} q
	- \begin{bmatrix} q_1 & -q_1 \DD h \end{bmatrix}
		\begin{bmatrix} \DD \\ E \end{bmatrix}
	&= q_1 \DD (\iota q) - q_1 \DD (\id - h E)
	= 0 .
\end{align*}
The first set of identities involving the homotopy
corrections also easily follows from the definition of the $\iota$ and
$q$ bundle maps.
We check the remaining ones by direct computation:
\begin{equation*}
	\id - \begin{bmatrix} q_1 & -q_1 \DD h \end{bmatrix}
		\begin{bmatrix} \iota_1 \\ 0 \end{bmatrix}
	= \id - q_1 \iota_1 = 0 ,
\end{equation*}
\begin{align*}
	\left( \begin{bmatrix} \id & 0 \\ 0 & \id \end{bmatrix}
	- \begin{bmatrix} \iota_1 \\ 0 \end{bmatrix}
		\begin{bmatrix} q_1 & -q_1 \DD h \end{bmatrix}
	- \begin{bmatrix} \DD \\ E \end{bmatrix}
		\begin{bmatrix} 0 & h \end{bmatrix}
	\right)
	&= \begin{bmatrix}
		\id - \iota_1 q_1 & \iota_1 q_1 \DD h - \DD h \\
		0 & \id - Eh
		\end{bmatrix}
	\\
	&= \begin{bmatrix}
		h_1 E_1 & - (h_1 E_1 \DD) h \\
		0 & \id - Eh
		\end{bmatrix}
	= \begin{bmatrix}
		h_1 E_1 & - (\DD h + h') E h \\
		0 & \id - Eh
		\end{bmatrix}
	\\
	&= \begin{bmatrix}
		\id & (\DD h + h') \\
		0 & \id
		\end{bmatrix}
		\begin{bmatrix}
		h_1 E_1 & - (\DD h + h') \\
		0 & \id - Eh
		\end{bmatrix} ,
\end{align*}
where the last factor has the property
\begin{equation*}
	\begin{bmatrix}
		h_1 E_1 & - (\DD h + h') \\
		0 & \id - Eh
	\end{bmatrix}
	\begin{bmatrix} \DD \\ E \end{bmatrix}
	= \begin{bmatrix}
		(h_1 E_1 \DD - \DD h E - h' E) \\
		E(\id - hE)
	\end{bmatrix}
	= \begin{bmatrix}
		0 \\
		(E \iota) q
	\end{bmatrix}
	= \begin{bmatrix}
		0 \\
		0
	\end{bmatrix} .
\end{equation*}

(c)
At this point, we know that the local solutions of $K[v] = 0$ are in
bijection with the $\tilde{\DD}$-flat local sections of $\tilde{U} \to
M$. In principle, it is now sufficient to check that $\DD$ is flat (if
it were not flat, then the rank of $\tilde{U}$ could not coincide with
the dimension of the local solution space of $K[v]=0$, though the two do
coincide by construction from part (b) of our proof). However, we will
take a slightly indirect route and show a more general result, that will
also be referred to in our discussion of the Killing equation in
Section~\ref{sec:killing}. Namely, provided the local solutions of
$\tilde{\DD} \tilde{u} = 0$ span a sub-bundle $\bar{\iota} \colon
\bar{U} \hookrightarrow \tilde{U} \to M$, we will show that the
restriction of $\tilde{\DD}$ to $\bar{U} \to M$ is flat and the original
$\tilde{\DD} \tilde{u} = 0$ equation is equivalent to the new $\bar{\DD}
\bar{u} = 0$ equation.

In our case, from the regular finite type assumption on $K[v] = 0$, we know that
the local solution space has locally constant (finite) dimension, which
is easily seen to be equivalent to the local solutions of $\tilde{\DD}
\tilde{u} = 0$ spanning a sub-bundle.

Now, under our hypotheses and since we are working locally, we can
choose a frame on $\bar{U} \to M$ which corresponds to flat sections
$\bm{\tilde{u}}_\beta$ on $\tilde{U}$. Namely, we define
$\bar{\iota}(\bar{u}) = \bar{u}^\beta \bm{\tilde{u}}_\beta$. Locally, there
also exists the projection bundle map $\bar{q}\colon \tilde{U} \to
\bar{U} \to M$, which satisfies $\bar{q} \bar{\iota} = \id$ and hence acts as
$\bar{q}\left(\bar{u}^\beta \bm{\tilde{u}}_\beta\right) = \bar{u}$.
Hence, using again the notation $\bar{\iota}_1 = 1\otimes \bar{\iota}$, we get
the identity
\begin{equation}
\tilde{\DD}[\bar{\iota}(\bar{u})]
= d\bar{u}^\beta \otimes \bm{\tilde{u}}_\beta
	+ \bar{u}^\beta (\tilde{\DD} \bm{\tilde{u}}_\beta)
= d\bar{u}^\beta \otimes \bm{\tilde{u}}_\beta
= \iota_1 \left(\bar{\DD} \bar{u}\right) ,
\end{equation}
where $d$ is simply the exterior derivative acting on the scalars
$\bar{u}^\beta$ and we have defined $\bar{\DD}$ to act on the frame
components of $\bar{u}$ as $(\bar{\DD}\bar{u})^\beta = d \bar{u}^\beta$.
The operator $\bar{\DD}$ is clearly a flat connection on the bundle
$\bar{U} \to M$. It remains only to exhibit the equivalence between the
equations $\tilde{\DD}\tilde{u} = 0$ and $\bar{\DD}\bar{u} = 0$.

As we already discussed in part (b) of our proof, when $\bar{\iota}$ is
not surjective, there must be some integrability conditions that follow
from the differential consequences of $\tilde{\DD}\tilde{u} = 0$. In
other words, there exists a differential operator $\bar{\lambda}$ such
that $\bar{\lambda} \tilde{\DD}$ is a non-differential operator,
satisfying $(\bar{\lambda} \tilde{\DD}) \iota = 0$, with $\iota \colon
\bar{U} \hookrightarrow \ker \bar{\lambda}\tilde{\DD}$ actually being an
isomorphism. Again, as before, this means that there exists an operator
$\bar{h}$ such that $\bar{\iota} \bar{q} = \bar{h} (\bar{\lambda}
\bar{\DD})$. Recalling again the notation, $\tilde{U}_1 = T^*M \otimes_M
\tilde{U}$ and $\bar{U}_1 = T^*M \otimes_M \bar{U}$, as well as
$\bar{q}_1 = \id \otimes \bar{q}$ and $\bar{\iota}_1 = \id \otimes
\bar{\iota}$. With that in mind, the desired equivalence is explicitly
given by the diagram
\begin{equation}
\begin{tikzcd}[column sep=4cm,row sep=4cm]
	\tilde{u} \ar{r}{\tilde{\DD}}
		\ar[swap,shift right]{d}{\bar{q}} \&
	\tilde{u}_1
		\ar[swap,shift right]{d}{\bar{q}_1 (\id-\tilde{\DD} \bar{h} \bar{\lambda})}
		\ar[dashed,bend left]{l}{\bar{h} \bar{\lambda}}
	\\
	\bar{u}
		\ar[swap]{r}{\bar{\DD}}
		\ar[swap,shift right]{u}{\bar{\iota}} \&
	\bar{u}_1
		\ar[swap,shift right]{u}{\bar{\iota}_1}
		\ar[swap,dashed,bend right]{l}{0}
\end{tikzcd} .
\end{equation}
The arguments to check all the required identities are similar to those
in part (b). We check that the solid arrows form commutative squares by
direct computation:
\begin{align*}
	\bar{\DD} \bar{q}
	&= \bar{q}_1 \tilde{\DD} (\bar{\iota} \bar{q})
	= \bar{q}_1 \tilde{\DD} - \bar{q}_1 \tilde{\DD} \bar{h} \bar{\lambda} \tilde{\DD}
	= \bar{q}_1 (\id - \tilde{\DD} \bar{h} \bar{\lambda}) \tilde{\DD} , \\
	\tilde{\DD} \bar{\iota}
	&= \bar{\iota}_1 \bar{\DD} .
\end{align*}

To check some of the identities with the homotopy corrections, we need
one more identity. For ease of notation, define $\bar{E} = \bar{\lambda}
\tilde{\DD}$, which by assumption is a non-differential operator which
incorporates all the integrability conditions of the equation
$\tilde{\DD}\tilde{u} = 0$. Just as in part (b), since all integrability
conditions must factor through $\bar{E}$, there must exist an operator
$\bar{h}'$ such that $\bar{h}_1 \bar{E}_1 \tilde{\DD} - \tilde{\DD}
\bar{h} \bar{E} = \bar{h}' \bar{E}$, which implies the identity
\begin{equation}
	[\bar{h}_1 \bar{E}_1 - (\tilde{\DD} \bar{h} + \bar{h}') \bar{\lambda}]
		\tilde{\DD} = 0 .
\end{equation}
Hence, we can verify that
\begin{align*}
	\id - \bar{q}\bar{\iota} &= 0 , \\
	\id - \bar{\iota} \bar{q} - \bar{h}\bar{\lambda} \tilde{\DD} &= 0 , \\
	(\id - \bar{q}_1 (\id - \tilde{\DD}\bar{h}\bar{\lambda}) \bar{\iota}_1) \bar{\DD}
	&= (\id - (\bar{q}_1 \bar{\iota}_1)) \bar{\DD}
		+ \bar{q}_1 \tilde{\DD} \bar{h} \bar{\lambda} (\bar{\iota}_1 \bar{\DD}) \\
	&= \bar{q}_1 \tilde{\DD} \bar{h} ((\bar{\lambda} \tilde{\DD}) \bar{\iota})
	= 0 , \\
	(\id - \bar{\iota}_1 \bar{q}_1 (\id - \tilde{\DD}\bar{h}\bar{\lambda})
		- \tilde{\DD} \bar{h} \bar{\lambda}) \tilde{\DD}
	&= (\id - \bar{\iota}_1 \bar{q}_1) (\id - \tilde{\DD}\bar{h}\bar{\lambda}) \tilde{\DD} \\
	&= \bar{h}_1 \bar{E}_1 \tilde{\DD}
		- \bar{h}_1 \bar{E}_1 \tilde{\DD} (\bar{h} \bar{\lambda} \tilde{\DD}) \\
	&= \bar{h}_1 \bar{E}_1 \tilde{\DD}
		- (\bar{h}_1 \bar{E}_1 \tilde{\DD}) (\id - \bar{\iota} \bar{q}) \\
	&= (\tilde{\DD} \bar{h} + \bar{h}') (\bar{E} \bar{\iota}) \bar{q}
	= 0 .
\end{align*}
This concludes the proof.
\end{proof}

Since according to Definition~\ref{def:fintype-orig} the flat section
equation $\DD v = 0$ for a flat connection $\DD$ is itself of regular finite
type (with $N=0$), Proposition~\ref{prp:fintype-conn} shows that
Definitions~\ref{def:fintype} and~\ref{def:fintype-orig} are clearly
equivalent and can be used interchangeably.

The reader might notice that the structure of parts (b) and (c) in
the proof of Proposition~\ref{prp:fintype-conn} is rather similar. The
reason that we have included both of them in detail is that part (c) can
basically be read independently and establishes the following (of course
also well-know) more specific result:

\begin{lem} \label{lem:conn-flat-conn}
Let $\DD$ be a connection on a vector bundle $V$. If the local solutions
of the flat section equation $\DD v = 0$ span a sub-bundle
$W\hookrightarrow V$, then the restriction $\bar{\DD} = \DD|_W$ of $\DD$
to $W$ is a flat connection on $W$. Moreover, $\DD$ is equivalent to
$\bar{\DD}$ in the sense of one operator complexes
(Definition~\ref{def:homalg}).
\end{lem}

\section{Notation reference} \label{sec:notation}

\subsection{Constant curvature spacetime} \label{sec:cc-notation}

\begin{tabular}{lll}
	$\alpha$
		& curvature constant
		& Section~\ref{sec:cc}
	\\[3ex]
	$C_i$, $C_0 = K$
		& Calabi compatibility complex for Killing operator $K$
		& \eqref{eq:calabi}
	\\[3ex]
	$S\odot T$
		& Kulkarni-Nomizu product
		& \eqref{eq:kn-prod}
\end{tabular}

\subsection{FLRW spacetimes} \label{sec:flrw-notation}

\begin{tabular}{lp{6cm}l}
	$(M,g) = (I\times F, -dt^2 + f^2 \tilde{g}^F)$
		& FLRW geometry, with scale factor $f$
		& Section~\ref{sec:flrw}
	\\[3ex]
	$\alpha$
		& curvature constant
		& Section~\ref{sec:flrw}
	\\[3ex]
	$U_a$
		& unit covector normal to $F$ factor
		& Section~\ref{sec:flrw}
	\\[3ex]
	$R_{abcd}, R_{ab}, \R$
		& background Riemann, Ricci and scalar curvatures on $(M,g)$
		& \eqref{eq:flrw-curv}
	\\[3ex]
	$\del_t, \tnabla$
		& derivative operators extended from $I$ and $F$ to $M$
		& Section~\ref{sec:flrw}
	\\[3ex]
	$\tilde{\Delta}, \tilde{\div}, \tilde{\tr}$
		& Laplacian, divergence and trace extended from $F$ to $M$
		& \eqref{eq:flrw-lapdivtr}
	\\[3ex]
	$v_a = Af U_b + f^2 \tX_a$
		& covector parametrization, $U^a\tX_a = 0$
		& \eqref{eq:flrw-v-param}
	\\[3ex]
	$h_{ab} = p U_a U_b - 2f^2 U_{(a} \tY_{b)} + f^2 \tZ_{ab}$
		& symmetric 2-tensor parametrization, $U^a\tY_a = 0 = U^a\tZ_{ab}$
		& \eqref{eq:flrw-h-param}
	\\[3ex]
	$K$
		& Killing operator on $(M,g)$
		& \eqref{eq:flrw-killing}
	\\[3ex]
	$\tC_i$, $\tC_0 = \tilde{K}$
		& extension of Calabi and Killing operators from $(F,g^F)$ to $M$
		& \eqref{eq:flrw-calabi}
	\\[3ex]
	$J, \tilde{J}$
		& operator to extract $A = J\circ K[v]$, subcomponent $\tilde{J}$
		& \eqref{eq:flrw-J}, \eqref{eq:flrw-tJ}, \eqref{eq:flrw-J-expl}
	\\[3ex]
	$H_J, \tilde{H}_J$
		& factorization of $J, \tilde{J}$
		& \eqref{eq:flrw-HJ}, \eqref{eq:flrw-tHJ}, \eqref{eq:flrw-HJ-expl}
\end{tabular}

\subsection{Schwarzschild-Tangherlini spacetimes} \label{sec:schw-notation}

\begin{tabular}{lp{6.5cm}l}
	$(\bar{\M}, \bar{g}) = (\M \times \S, g + r^2 \Omega), \bar{\nabla}_\mu$
		& generalized Schwarzschild-Tangherlini (gST) spacetime,
			covariant derivative
		& Section~\ref{sec:schw}
	\\[3ex]
	$\bar{R}_{abcd}, \bar{R}_{ab}, \bar{\R}$
		& background Riemann, Ricci and scalar curvatures on
			$(\bar{\M},\bar{g})$
		& \eqref{eq:schw-curv-riem}, \eqref{eq:schw-curv-ricc}
	\\[3ex]
	$\bar{T}_{abcd}$
		& $\Lambda$-shifted background Riemann curvature on
			$(\bar{\M},\bar{g})$
		& \eqref{eq:schw-curv-lambda}
	\\[3ex]
	$(\S,\Omega), D_A$
		& constant curvature factor,
			covariant derivative extended to $\bar{\M}$
		& Section~\ref{sec:schw}
	\\[3ex]
	$(\M,g), \nabla_a$
		& radio-temporal factor,
			covariant derivative extended to $\bar{\M}$
		& Section~\ref{sec:schw}
	\\[3ex]
	$R_{abcd}, R_{ab}, \R$
		& background Riemann, Ricci and scalar curvatures on $(\M,g)$
		& \eqref{eq:schw-curv-rt}
	\\[3ex]
	$-f(r)$
		& the $\bar{g}_{tt}$ metric component
		& \eqref{eq:gen-g}, \eqref{eq:gen-f}
	\\[3ex]
	$f_1(r) = r f'(r), f_2(r) = r'f_1(r)$
		& derivatives of $f(r)$
		& \eqref{eq:gen-f1}, \eqref{eq:gen-f2}
	\\[3ex]
	$M, \Lambda, \alpha$
		& mass, cosmological, curvature constants
		& \eqref{eq:gen-f}
	\\[3ex]
	$t_a = -f dt_a$
		& timelike Killing covector on $(\M,g)$
		& Section~\ref{sec:schw}
	\\[3ex]
	$r, r_a = dr_a$
		& radial coordinate and covector on $(\M,g)$
		& Section~\ref{sec:schw}
	\\[3ex]
	$v_\mu \to \begin{bmatrix} u_t\, f dt_a + u_r dr_a \\ r (r X_A) \end{bmatrix}$
		& covector parametrization on $\bar{\M}$
		& \eqref{eq:schw-vh-param}
	\\[3ex]
	$h_{\mu\nu} \to \begin{bmatrix}
		p \, r_a r_b - 2t_{(a} w_{b)} & r (r Y_{aB}) \\
		r (r Y_{bA}) & r^2 Z_{AB}
	\end{bmatrix}$
		& symmetric 2-tensor parametrization on $\bar{\M}$
		& \eqref{eq:schw-vh-param}
	\\[3ex]
	$\bar{K}$
		& Killing operator on $(\bar{\M},\bar{g})$
		& \eqref{eq:schw-killing}
	\\[3ex]
	$C_i$, $C_0 = K$
		& extension of Calabi and Killing operators from $(\S,\Omega)$ to $\bar{\M}$
		& \eqref{eq:schw-killing}
	\\[3ex]
	$J_1$
		& operator to extract $u_r = J_1\circ \bar{K}[v]$
		& \eqref{eq:schw-J1}
	\\[3ex]
	$J_2$
		& operator to extract $\nabla_a \frac{v_B}{r} = J_2\circ \bar{K}[v]_{aB}$
		& \eqref{eq:schw-J2}
	\\[3ex]
	$H_{J_1}, \tilde{H}_{J_1}$
		& factorization of $J_1$
		& \eqref{eq:schw-HJ1-tHJ1}
	\\[3ex]
	$H_{J_2}, \tilde{H}_{J_2}$
		& factorization of parts of $J_2$
		& \eqref{eq:schw-HJ2}, \eqref{eq:schw-tHJ2}
\end{tabular}

\bibliographystyle{utphys-alpha}
\bibliography{killing}

\providecommand{\href}[2]{#2}\begingroup\raggedright\begin{thebibliography}{10}

\bibitem{ab-kerr}
S.~Aksteiner and T.~B\"{a}ckdahl, ``All local gauge invariants for
  perturbations of the {Kerr} spacetime,''
  \href{http://dx.doi.org/10.1103/PhysRevLett.121.051104}{{\em Physical Review
  Letters} {\bfseries 121} (2018) 051104},
  \href{http://arxiv.org/abs/1803.05341}{{\ttfamily arXiv:1803.05341}}.

\bibitem{aabkw}
L.~Andersson, S.~Aksteiner, T.~B\"{a}ckdahl, I.~Khavkine, and B.~Whiting,
  ``Compatibility complex for black hole spacetimes,'' 2018.
\newblock In preparation.

\bibitem{btz}
M.~Ba\~{n}ados, M.~Henneaux, C.~Teitelboim, and J.~Zanelli, ``Geometry of the
  2+1 black hole,'' \href{http://dx.doi.org/10.1103/physrevd.48.1506}{{\em
  Physical Review D} {\bfseries 48} (1993) 1506--1525},
  \href{http://arxiv.org/abs/gr-qc/9302012}{{\ttfamily arXiv:gr-qc/9302012}}.

\bibitem{bardeen}
J.~M. Bardeen, ``Gauge-invariant cosmological perturbations,''
  \href{http://dx.doi.org/10.1103/physrevd.22.1882}{{\em Physical Review D}
  {\bfseries 22} (1980) 1882--1905}.

\bibitem{bcpds}
M.~L. Bedran, M.~O. Calv\~{a}o, F.~M. Paiva, and I.~Dami\~{a}o Soares, ``Taub's
  plane-symmetric vacuum spacetime reexamined,''
  \href{http://dx.doi.org/10.1103/physrevd.55.3431}{{\em Physical Review D}
  {\bfseries 55} (1997) 3431--3439},
  \href{http://arxiv.org/abs/gr-qc/9608058}{{\ttfamily arXiv:gr-qc/9608058}}.

\bibitem{janet}
Y.~A. Blinkov, C.~F. Cid, V.~P. Gerdt, W.~Plesken, and D.~Robertz, ``The
  {MAPLE} package {Janet}: {I}. polynomial systems. {II}. linear partial
  differential equations,'' in {\em Proceedings of the 6th International
  Workshop on Computer Algebra in Scientific Computing, Passau (Germany)},
  V.~G. Ganzha, E.~W. Mayr, and E.~V. Vorozhtsov, eds., pp.~31--54.
\newblock Institut f\"{u}r Informatik, Technische Universit\"{a}t M\"{u}nchen,
  Garching, 2003.
\newblock \url{https://wwwb.math.rwth-aachen.de/Janet/}.

\bibitem{cahen-wallach}
M.~Cahen and N.~Wallach, ``Lorentzian symmetric spaces,''
  \href{http://dx.doi.org/10.1090/s0002-9904-1970-12448-x}{{\em Bulletin of the
  American Mathematical Society} {\bfseries 76} (1970) 585--591}.

\bibitem{calabi}
E.~Calabi, \href{http://dx.doi.org/10.1090/pspum/003}{``On compact,
  {R}iemannian manifolds with constant curvature. {I},''} in {\em Differential
  Geometry}, C.~B. Allendoerfer, ed., vol.~3 of {\em Proceedings of Symposia in
  Pure Mathematics}, pp.~155--180.
\newblock AMS, Providence, RI, 1961.

\bibitem{cdk}
G.~Canepa, C.~Dappiaggi, and I.~Khavkine, ``{IDEAL} characterization of
  isometry classes of {FLRW} and inflationary spacetimes,''
  \href{http://dx.doi.org/10.1088/1361-6382/aa9f61}{{\em Classical and Quantum
  Gravity} {\bfseries 35} (2018) 035013},
  \href{http://arxiv.org/abs/1704.05542}{{\ttfamily arXiv:1704.05542}}.

\bibitem{dotti-schw}
G.~Dotti, ``Nonmodal linear stability of the {Schwarzschild} black hole,''
  \href{http://dx.doi.org/10.1103/physrevlett.112.191101}{{\em Physical Review
  Letters} {\bfseries 112} (2014) 191101},
  \href{http://arxiv.org/abs/1307.3340}{{\ttfamily arXiv:1307.3340}}.

\bibitem{eastwood}
M.~Eastwood, \href{http://dx.doi.org/10.1007/978-0-387-73831-4\_3}{``Notes on
  projective differential geometry,''} in {\em Symmetries and Overdetermined
  Systems of Partial Differential Equations}, M.~Eastwood and W.~Miller, eds.,
  vol.~144 of {\em The IMA Volumes in Mathematics and its Applications}, ch.~3,
  pp.~41--60.
\newblock Springer, New York, NY, 2008.

\bibitem{fs-schw}
J.~J. Ferrando and J.~A. S\'{a}ez, ``An intrinsic characterization of the
  {Schwarzschild} metric,''
  \href{http://dx.doi.org/10.1088/0264-9381/15/5/014}{{\em Classical and
  Quantum Gravity} {\bfseries 15} (1998) 1323--1330}.

\bibitem{fhh}
M.~B. Fr\"{o}b, T.-P. Hack, and A.~Higuchi, ``Compactly supported linearised
  observables in single-field inflation,''
  \href{http://dx.doi.org/10.1088/1475-7516/2017/07/043}{{\em Journal of
  Cosmology and Astroparticle Physics} {\bfseries 2017} (2017) 043},
  \href{http://arxiv.org/abs/1703.01158}{{\ttfamily arXiv:1703.01158}}.

\bibitem{fhk}
M.~B. Fr\"{o}b, T.-P. Hack, and I.~Khavkine, ``Approaches to linear local
  gauge-invariant observables in inflationary cosmologies,''
  \href{http://dx.doi.org/10.1088/1361-6382/aabcb7}{{\em Classical and Quantum
  Gravity} {\bfseries 35} (2018) 115002},
  \href{http://arxiv.org/abs/1801.02632}{{\ttfamily arXiv:1801.02632}}.

\bibitem{gg83}
J.~Gasqui and H.~Goldschmidt, ``D\'{e}formations infinit\'{e}simales des
  espaces riemanniens localement sym\'{e}triques. {I},''
  \href{http://dx.doi.org/10.1016/0001-8708(83)90090-7}{{\em Advances in
  Mathematics} {\bfseries 48} (1983) 205--285}.

\bibitem{gg88}
J.~Gasqui and H.~Goldschmidt,
  \href{http://dx.doi.org/10.1007/978-94-009-3057-5\_14}{``Complexes of
  differential operators and symmetric spaces,''} in {\em Deformation Theory of
  Algebras and Structures and Applications}, M.~Hazewinkel and M.~Gerstenhaber,
  eds., vol.~247 of {\em NATO ASI Series}, pp.~797--827.
\newblock Kluwer, Dordrecht, 1988.

\bibitem{geroch-killing}
R.~Geroch, ``Limits of spacetimes,''
  \href{http://dx.doi.org/10.1007/bf01645486}{{\em Communications in
  Mathematical Physics} {\bfseries 13} (1969) 180--193}.

\bibitem{goldschmidt-lin}
H.~Goldschmidt, ``Existence theorems for analytic linear partial differential
  equations,'' \href{http://dx.doi.org/10.2307/1970689}{{\em The Annals of
  Mathematics} {\bfseries 86} (1967) 246--270}.

\bibitem{higuchi}
A.~Higuchi, ``Equivalence between the {Weyl}-tensor and gauge-invariant
  graviton two-point functions in {Minkowski} and {de Sitter} spaces,'' 2013.
\newblock \href{http://arxiv.org/abs/1204.1684}{{\ttfamily arXiv:1204.1684}}.

\bibitem{ki-stab}
A.~Ishibashi and H.~Kodama, ``Stability of higher-dimensional {Schwarzschild}
  black holes,'' \href{http://dx.doi.org/10.1143/ptp.110.901}{{\em Progress of
  Theoretical Physics} {\bfseries 110} (2003) 901--919},
  \href{http://arxiv.org/abs/hep-th/0305185}{{\ttfamily arXiv:hep-th/0305185}}.

\bibitem{jezierski}
J.~Jezierski, ``{`Peeling property'} for linearized gravity in null
  coordinates,'' \href{http://dx.doi.org/10.1088/0264-9381/19/9/310}{{\em
  Classical and Quantum Gravity} {\bfseries 19} (2002) 2463--2490},
  \href{http://arxiv.org/abs/gr-qc/0111030}{{\ttfamily arXiv:gr-qc/0111030}}.

\bibitem{kh-peierls}
I.~Khavkine, ``Covariant phase space, constraints, gauge and the {P}eierls
  formula,'' \href{http://dx.doi.org/10.1142/s0217751x14300099}{{\em
  International Journal of Modern Physics A} {\bfseries 29} (2014) 1430009},
  \href{http://arxiv.org/abs/1402.1282}{{\ttfamily arXiv:1402.1282}}.

\bibitem{kh-calabi}
I.~Khavkine, ``The {Calabi} complex and {Killing} sheaf cohomology,''
  \href{http://dx.doi.org/10.1016/j.geomphys.2016.06.009}{{\em Journal of
  Geometry and Physics} {\bfseries 113} (2017) 131--169},
  \href{http://arxiv.org/abs/1409.7212}{{\ttfamily arXiv:1409.7212}}.

\bibitem{kh-gst}
I.~Khavkine, ``{IDEAL} characterization of higher dimensional spherically
  symmetric black holes,'' 2018.
\newblock \href{http://arxiv.org/abs/1807.09699}{{\ttfamily arXiv:1807.09699}}.

\bibitem{ki-master}
H.~Kodama and A.~Ishibashi, ``A master equation for gravitational perturbations
  of maximally symmetric black holes in higher dimensions,''
  \href{http://dx.doi.org/10.1143/ptp.110.701}{{\em Progress of Theoretical
  Physics} {\bfseries 110} (2003) 701--722},
  \href{http://arxiv.org/abs/hep-th/0305147}{{\ttfamily arXiv:hep-th/0305147}}.

\bibitem{kis-n+m}
H.~Kodama, A.~Ishibashi, and O.~Seto, ``Brane world cosmology: Gauge-invariant
  formalism for perturbation,''
  \href{http://dx.doi.org/10.1103/physrevd.62.064022}{{\em Physical Review D}
  {\bfseries 62} (2000) 064022},
  \href{http://arxiv.org/abs/hep-th/0004160}{{\ttfamily arXiv:hep-th/0004160}}.

\bibitem{lobo-mimoso}
F.~S.~N. Lobo and J.~P. Mimoso, ``Possibility of hyperbolic tunneling,''
  \href{http://dx.doi.org/10.1103/physrevd.82.044034}{{\em Physical Review D}
  {\bfseries 82} (2010) 044034},
  \href{http://arxiv.org/abs/0907.3811}{{\ttfamily arXiv:0907.3811}}.

\bibitem{martel-poisson}
K.~Martel and E.~Poisson, ``Gravitational perturbations of the {Schwarzschild}
  spacetime: A practical covariant and gauge-invariant formalism,''
  \href{http://dx.doi.org/10.1103/physrevd.71.104003}{{\em Physical Review D}
  {\bfseries 71} no.~10, (2005) 104003},
  \href{http://arxiv.org/abs/gr-qc/0502028}{{\ttfamily arXiv:gr-qc/0502028}}.

\bibitem{oneill}
B.~O'Neill, {\em Semi-Riemannian Geometry With Applications to Relativity},
  vol.~103 of {\em Pure and Applied Mathematics}.
\newblock Elsevier Science, 1983.

\bibitem{seiler-inv}
W.~M. Seiler, {\em Involution: The Formal Theory of Differential Equations and
  its Applications in Computer Algebra}, vol.~24 of {\em Algorithms and
  Computation in Mathematics}.
\newblock Springer, 2010.

\bibitem{swaab}
A.~G. Shah, B.~F. Whiting, S.~Aksteiner, L.~Andersson, and T.~B\"{a}ckdahl,
  ``Gauge-invariant perturbations of {Schwarzschild} spacetime,'' 2016.
\newblock \href{http://arxiv.org/abs/1611.08291}{{\ttfamily arXiv:1611.08291}}.

\bibitem{spencer}
D.~C. Spencer, ``Overdetermined systems of linear partial differential
  equations,'' \href{http://dx.doi.org/10.1090/s0002-9904-1969-12129-4}{{\em
  Bulletin of the American Mathematical Society} {\bfseries 75} (1969)
  179--240}.

\bibitem{stephani-sols}
H.~Stephani, D.~Kramer, M.~MacCallum, C.~Hoenselaers, and E.~Herlt,
  \href{http://dx.doi.org/10.1017/CBO9780511535185}{{\em Exact Solutions of
  {Einstein's} Field Equations}}.
\newblock Cambridge University Press, Cambridge, 2003.

\bibitem{stewart-walker}
J.~M. Stewart and M.~Walker, ``Perturbations of {Space-Times} in general
  relativity,'' \href{http://dx.doi.org/10.1098/rspa.1974.0172}{{\em
  Proceedings of the Royal Society of London. A. Mathematical and Physical
  Sciences} {\bfseries 341} (1974) 49--74}.

\bibitem{tarkhanov}
N.~N. Tarkhanov, \href{http://dx.doi.org/10.1007/978-94-011-0327-5}{{\em
  Complexes of Differential Operators}}, vol.~340 of {\em Mathematics and Its
  Applications}.
\newblock Kluwer, Dordrecht, 1995.

\bibitem{taub}
A.~H. Taub, ``Empty space-times admitting a three parameter group of motions,''
  \href{http://dx.doi.org/10.2307/1969567}{{\em The Annals of Mathematics}
  {\bfseries 53} (1951) 472--490}.

\bibitem{wald-gr}
R.~M. Wald, {\em General Relativity}.
\newblock University of Chicago Press, Chicago, 1984.

\bibitem{weibel}
C.~A. Weibel, \href{http://dx.doi.org/10.1017/CBO9781139644136}{{\em An
  introduction to homological algebra}}, vol.~38 of {\em Cambridge Studies in
  Advanced Mathematics}.
\newblock Cambridge University Press, Cambridge, 1994.

\bibitem{whiting-price}
B.~F. Whiting and L.~R. Price, ``Metric reconstruction from {Weyl} scalars,''
  \href{http://dx.doi.org/10.1088/0264-9381/22/15/003}{{\em Classical and
  Quantum Gravity} {\bfseries 22} (2005) S589--S604}.

\end{thebibliography}\endgroup

\end{document}